\documentclass[final,12pt]{clear2025} 

\title[Estimating Heterogeneous Causal Effects with Tree-Based Methods under Imperfect Compliance]{Estimating Heterogeneous Causal Effects with Tree-Based Methods under Imperfect Compliance}
\usepackage{times}
\usepackage{algorithm} 
\usepackage{multirow}
\usepackage{dsfont} 
\usepackage{mathtools} 
\usepackage{booktabs} 
\usepackage{bm} 
\usepackage{physics} 
\input{ee.sty} 
\usepackage{float}
\usepackage{comment}
\usepackage{thmtools}
\usepackage{enumitem}
\usepackage{array}
\usepackage{rotating}
\usepackage[section]{placeins}
\usepackage{subcaption}
\usepackage{dsfont}
\usepackage{longtable} 
\usepackage{pdflscape}
\newenvironment{notes}
  {\par\medskip
   \begin{minipage}{\linewidth}
   \footnotesize\emph{Notes: }}
{\end{minipage}}
\renewcommand{\proofname}{Proof}
\renewenvironment{proof}[1][\proofname]
{%
  \par\noindent
  {\bfseries\upshape #1.}\ %
}
{  \jmlrQED}

\clearauthor{
 \Name{Karolina Gliszczy\'nska-Schroeder} \Email{karolina.gliszczynska(at)vwl.uni-due.de}\\
 \addr Chair of Econometrics, Faculty of Business Administration and Economics, University of Duisburg-Essen, Universitätsstraße 12, 45117 Essen, Germany%
}

\begin{document}


\newtheorem{assum}{Assumption}


\maketitle


\begin{abstract}%
We study the impact of conditional complier average causal effect (CCACE) estimation methods on the performance of subgroup discovery and heterogeneous causal effect estimation under imperfect compliance. Building on the Bayesian Causal Forest with Instrumental Variable (BCF-IV; \citet{bargaglibcf}) method, we introduce a two-step, model-agnostic approach that allows any suitable machine learning method to be used for the CCACE estimation in the first step. Specifically, we implement two non-Bayesian tree-based methods, both using forest-based learners: DRRF-IV, a debiased transformed-outcome regression-forest approach, and a GRF-based IV adaptation of the generalized random forest framework. Through a simulation study, we assess the precision, bias, and the ability to correctly identify the underlying subgroup structure of the proposed methods relative to BCF-IV.  The results show that the non-Bayesian methods perform competitively across the considered simulation settings, with performance improving for larger sample sizes and moderately large treatment effects, while reducing computational runtime. We apply our new methods by revisiting an empirical study that examines the effect of prompt admission to intensive care units (ICU) on 28-day mortality across 48 UK National Health Service hospitals. While previous work finds no significant overall treatment effect, we investigate whether subgroups of patients may benefit more from prompt ICU admission.

\end{abstract}

\begin{keywords}%
    Heterogeneous treatment effects,  instrumental variables, subgroup
discovery, causal forest
\end{keywords}

\begin{jelclass}
    C14, C21, C26, C52
\end{jelclass}

\section{Introduction}
\setcounter{footnote}{0}
Randomized controlled trials (RCTs) have been considered the gold standard for causal inference \citep{RCT_gold_2021}. However, RCTs are often infeasible in practice, due to ethical or financial constraints \citep{ RCT_goldstandard, CI_and_Obsdata}. As a consequence, researchers frequently rely on observational data, which are often affected by treatment endogeneity, that is, by correlation between the treatment variable and unobserved factors influencing the outcome. Instrumental variable (IV) methods \citep{Wooldridge} offer a strategy for addressing endogeneity  by using instruments that affect the treatment but have no direct effect on the outcome. Traditionally, instrumental variable methods have been used primarily to identify local causal effects for complier subpopulations, most notably the Local Average Treatment Effect (LATE) \citep{Angrist_Imbens_1994, ACR_IV, angrist_imbens_iv_96}. Although some recent contributions study identification of the population-wide Average Treatment Effect (ATE) under IV assumptions, these approaches typically require stronger conditions \citep{wang_tchetgen_2018}. While these methods provide internally valid causal estimates, they summarize effects at an aggregate level and may conceal heterogeneity in treatment effects across subpopulations. Specifically, this means that the effect of treatment for compliers may differ across patient or covariate groups, even when the overall LATE or CACE is close to zero.

This paper contributes to the literature on causal forest-based methods for heterogeneous treatment effect estimation \citep{athey_recursive_2016, wager_estimation_2018, athey_generalized_2019,oprescu2019orthogonal,hahn2020bayesian} in the context of instrumental variable settings. In particular, we focus on treatment endogeneity and imperfect compliance, meaning that treatment assignment does not necessarily coincide with the actual received treatment. In this setting, the target parameter is the conditional complier average causal effect (CCACE), rather than the standard CATE. Specifically, we study subgroup discovery in tree-based IV models, aiming to identify interpretable subsamples for which complier treatment effects might differ.

Recent contributions have combined tree-based methods with treatment endogeneity using instrumental variables. For example, the Generalized Random Forest (GRF) framework of \citet{athey_generalized_2019} provides a general approach for estimation problems defined by local moment conditions, including an IV framework as a special case. While GRF-IV can estimate heterogeneous treatment effects, it does not directly produce a single interpretable subgroup partition, since predictions are aggregated across many trees \citep{bargaglibcf}.

Closely related, both \citet{Wang_2021_CT_IV} and \citet{stoffi_gnecco_CTIV} develop causal tree approaches for IV settings (CT-IV), which estimate treatment effects at the leaf level. In \citet{Wang_2021_CT_IV}, the single-tree CT-IV serves as the base learner for the broader instrumental variable forest (IVF). While the tree component yields an interpretable partition, the full forest estimator aggregates across many trees. 
In parallel, \citet{stoffi_gnecco_CTIV} proposed an alternative CT-IV framework, which address the problem of treatment endogeneity by incorporating a propensity-score-like correction, replacing the treatment variable with an instrument. Although these CT-IV approaches differ in implementation, they primarily use tree-based partitions as a way to estimate heterogeneous treatment effects. The resulting leaves can be interpreted as subgroups, but subgroup discovery is not separated from the estimation task and the primary focus is on improving the accuracy of heterogeneous treatment effect estimation. The Bayesian Causal Forest with IV (BCF-IV) \citep{bargaglibcf} extends this line of work by introducing a two-step approach that directly supports subgroup detection, but relies on Bayesian modeling assumptions. Furthermore, in its original formulation and evaluation, subgroup discovery is studied primarily in settings with binary covariates.

Our work builds directly on the BCF-IV framework but replaces its Bayesian discovery step with a model-agnostic discovery step in the spirit of \citet{Two_step_pragmatic_subgroup_discovery}.  We keep the overall two-step structure largely unchanged, because it usefully separates the task of identifying subgroups from the task of estimating subgroup-specific causal effects. Following the honesty idea of \citet{athey_recursive_2016}, the sample is split into two disjoint subsamples. The first subsample is used to estimate heterogeneous treatment effects, which are then treated as outcomes in a CART algorithm \citep{cart84} to get a tree structure capturing treatment effect heterogeneity. In the second step, the tree structure is held fixed and used to estimate subgroup-specific causal effects via IV regressions within the terminal nodes. 

Thus, the main modification in our paper is therefore not the two-step logic itself, but the estimator used in the discovery step. Recent work has already explored replacing the original BCF-IV estimator with alternative Bayesian tree learners. For example, \citet{maßmann2026shrinkagebayesiancausalforest} introduce Shrinkage Bayesian Causal Forest with Instrumental Variables (SBCF-IV), which replaces the original BCF-IV discovery step by a sparsity-inducing SoftBART estimator. Our contribution takes this idea one step further by generalizing the discovery step itself. In particular, we allow any suitable machine learning method to estimate treatment effects, thereby removing the restriction to Bayesian models in the first step. Furthermore, we introduce two tree-based implementations within this framework. The first is the Debiased Robust Random Forest with IV (DRRF-IV), which is inspired by the double/debiased machine learning approach of \citet{chernozhukov_doubledebiased_2018}, but uses estimated nuisance parameters like  propensity scores and conditional outcome expectations to construct a transformed outcome. The second is a GRF-IV adaptation for subgroup identification. In addition, we refine existing evaluation metrics to accommodate continuous covariates.  We are interested in how effectively these proposed estimators capture heterogeneous treatment effects under varying subgroup structures and treatment effect sizes, and how the choice of method has a systematic impact in terms of accuracy and interpretability. 

The paper is structured as follows. Section~\ref{ch:Tree IV} reviews the potential outcome framework and outlines the required modifications for the IV setting. It further introduces the general two-step framework and presents the DRRF-IV and GRF-IV implementations. Section~\ref{ch:sim_study} reports results from a simulation study that compares the performance of the proposed estimators across multiple heterogeneity scenarios and treatment effect sizes.
Section~\ref{ch:emp_app} provides an empirical application using ICU Spotlight data \citep{Keele_IV}, where prompt admission to an intensive care unit (ICU) may affect patient mortality, and bed availability serves as an instrument. Building on their previous work, which finds no statistically significant average effects, we investigate heterogeneity in treatment effects across patient groups and demonstrate how the methods identify subgroups. 
Section \ref{ch:Discussion} concludes and discusses potential extensions.

\section{Tree-based methods with instrumental variables}
\label{ch:Tree IV}
We use the potential outcome framework and notation of Rubin’s causal model \citep{rubin_estimating_1974, rubin_bayesian_1978, imbens_rubin_2010}. Consider a set of $N$ units, indexed by $i=1,\dots N$.  For each unit $i$ let $D_i \in \{0,1\}$ be the binary treatment indicator and let $X_i \in \mathbb{R}^K$ be a $K$-dimensional covariate or feature vector. The potential outcome for unit $i$ under treatment $d \in \{0,1\}$ is denoted by $Y_i(d)$, where $Y_i(0)$ represents the outcome under control, and $Y_i(1)$ the outcome under treatment. The individual treatment effect (ITE) for a given unit is then defined as
\begin{equation}\label{eq:ITE}
\tau_i \coloneqq Y_i(1) - Y_i(0).
\end{equation}
The Causal inference framework relies on a set of foundational assumptions, which we provide, before relaxing unconfoundedness in the IV framework introduced below.  The first assumption is the Stable Unit Treatment Value Assumption (SUTVA), which rules out two potential complications. First, it excludes interference between units, so that the potential outcome of unit $i$ is affected only by its own treatment status and not by the treatment assignment of other units. Second, it requires that the treatment is well defined, meaning that there are no different or hidden versions of the same treatment level.

\begin{restatable}[Stable unit treatment value assumption (SUTVA)]{assum}{sutvaAssumption}
\label{assump:SUTVA}
\begin{align*}
\label{SUTVA}
\text{If } D_i = d \text{, then } Y_i(d) = Y_i^{obs} \text{ , } \forall d \in \{0, 1 \} \text{ , } \forall ~ i= 1, ..., N . 
\end{align*}
\end{restatable}
Under SUTVA, each unit has two potential outcomes, but only one of them can be observed for a given unit. Hence, the observed outcome can be written as
$$Y_i^{obs} = Y_i(1)D_i + Y_i(0) (1-D_i)$$.
The second assumption is unconfoundedness, which requires that, conditional on $X_i$, the treatment is independent of the potential outcomes. 

\begin{restatable}[Unconfoundedness]{assum}{unconfoundednessAssumption}
\label{assump:unconfoundedness}
\begin{align*}
D_i \perp\!\!\!\perp \left(Y_i(1), Y_i(0)\right) \mid X_i.
\end{align*}
\end{restatable}

The third assumption is positivity (also referred to as overlap). Let the propensity score denote the conditional probability of receiving treatment given covariates \begin{equation}
\label{eq:propensity}
p(x) \coloneqq \operatorname{Pr}(D_i=1|X_i=x).
\end{equation} 
Then the positivity or overlap assumption requires that for all $x$ in the support of $X_i$, the probabilities are bounded away from 0 and 1. 
\begin{restatable}[Overlap]{assum}{overlapAssumption}
\label{assump:overlap} 
\begin{align*}
0 < p(X_i=x) < 1 \quad \forall x \text{ in support of } X_i.
\end{align*}
\end{restatable}
Since the ITE in Equation~\eqref{eq:ITE} is not observable and therefore not directly identifiable, the literature commonly focuses on aggregated treatment effects such as the average treatment effect (ATE):
\begin{equation} \label{eq:ATE}
    \tau := \mathbb{E}[Y_i(1)-Y_i(0)],
\end{equation}
or the conditional average treatment effect (CATE)
\begin{equation} \label{eq:CATE}
    \tau(x) := \mathbb{E}[Y_i(1)-Y_i(0)\mid X_i = x].
\end{equation}
Note, that under the Assumptions~\ref{assump:SUTVA}–\ref{assump:overlap}, CATE can be identified from observable outcomes $Y_i^{\text{obs}}$
\begin{align}
\mathbb{E}[Y_i^{\text{obs}} \mid D_i = 1, X_i = x]
-
\mathbb{E}[Y_i^{\text{obs}} \mid D_i = 0, X_i = x].
\end{align}
In observational studies, the treatment variable $D_i$ is often endogenous \citep{ZHONG2021108967}. Even when accounting for observed covariates, treatment decisions may still be influenced by unobserved factors that are also related to $Y_i$. Endogeneity violates the unconfoundedness assumption and can lead to biased estimates \citep{Dominici2021From}. To address this problem, instrumental variable (IV) methods are a popular alternative for identifying causal effects \citep{Keele_IV}. Let $Z_i$ denote a binary instrument representing treatment assignment, for unit $i$. Following the notation of \citet{bargaglibcf}, we distinguish between the assigned treatment $Z_i$ and the actual received treatment $D_i$, allowing for imperfect compliance. Let $D_i(z)$ denote the treatment that unit $i$ would receive under assignment $Z_i=z$. The realized treatment is given by $D_i = D_i(Z_i)$.
The received treatment is a function of the treatment assignment $D_i(z)$. Then, potential outcomes can be written as $Y_i(z) :=Y_i(z, D_i(z))$. Depending on how units respond to their treatment assignments, we can categorize them into four subgroups: 

\begin{enumerate}
    \item \textbf{Compliers}: Units who take the treatment if and only if assigned to it $D_i(Z_i = 0) = 0$ and $D_i(Z_i = 1) = 1$.
    \item \textbf{Defiers}: Units who do the opposite of their assignment; $ D_i(Z_i = 0) = 1$, $ D_i(Z_i = 1) = 0$.
    \item \textbf{Always-takers}: Units who always receive treatment regardless of assignment; $D_i(Z_i = 0) = 1$, $ D_i(Z_i = 1) = 1$.
    \item \textbf{Never-takers}: Units who never receive treatment, regardless of assignment; $ D_i(Z_i = 0) = 0$, $ D_i(Z_i = 1) = 0$.
\end{enumerate}

To enable causal interpretation in the presence of endogeneity, we replace Assumption~\ref{assump:unconfoundedness} with a set of IV assumptions, following \citet{stoffi_gnecco_CTIV}. The first assumption is monotonicity, which rules out the existence of defiers.

\begin{restatable}[Monotonicity]{assum}{monotonicity}
\label{assump:monotonicity}
\begin{align*}
D_i(1) \geq D_i(0) \quad \text{for all } i = 1, \ldots, N.
\end{align*}
\end{restatable}

 In the case of one-sided noncompliance, such as when only those assigned to treatment can access it, this assumption is automatically satisfied. In two-sided noncompliance settings, the monotonicity assumption is generally plausible, though it cannot typically be verified empirically by the researcher (see \citet{stoffi_gnecco_CTIV} for a detailed discussion). The next assumption is the existence of compliers. This assumption requires that the complier subpopulation exists with positive probability.
 
\begin{restatable}[Existence of Compliers]{assum}{existenceCompliers}
\label{assump:existence_compliers}
The complier subpopulation exists with positive probability
\begin{align*}
\operatorname{Pr}\left(D_i(1) > D_i(0)\right) > 0.
\end{align*}
\end{restatable}
Next, we assume unconfoundedness of $Z_i$. Analogous to Assumption~\ref{assump:unconfoundedness}, this assumption requires that the instrument is as good as randomly assigned, conditional on the covariates. It ensures that the instrument is independent of both the potential outcomes and the potential treatment statuses, given $X_i$.
\begin{restatable}[Unconfoundedness of the IV]{assum}{UnconfoundednessIV}
\label{assump:Unconfoundedness_IV}
\begin{align*}
Z_i \perp\!\!\!\perp \Big( \{Y_i(z,d): z,d\in\{0,1\}\}, D_i(0), D_i(1) \Big)\mid X_i.
\end{align*}
\end{restatable}

Further we have the exclusion restriction, which asserts that the instrument affects the outcome only through its influence on the treatment. 
\begin{restatable}[Exclusion Restriction]{assum}{ExclusionRestriction}
\label{assump:Exclusion_Restriction}
\begin{align*}
Y_i(z,d) = Y_i(d) \quad \forall z,d \in \{0,1\}.
\end{align*}
\end{restatable}
That is, conditional on the treatment received, the instrument has no direct effect on the outcome. Lastly, we define the instrument propensity score as

\begin{equation} \label{eq:instrument_propscore}
    p_Z(x) := \operatorname{Pr}(Z_i = 1 \mid X_i=x),
\end{equation}
and assume overlap for the instrument. 
\begin{restatable}[Instrument Positivity]{assum}{InstrumentPositivity}
\label{assump:overlap_iv}
\begin{align*} 
0 < p_Z(x) < 1 \quad \forall ~x ~\text{ in the support of } X_i.
\end{align*}
\end{restatable}
Under Assumptions\footnote{Technically, one doesn't need Assumption \ref{assump:overlap_iv} here to identify CACE. However, we need it to construct the DRRF-IV estimator in Section \ref{sec:RF with IV}.} \ref{assump:monotonicity}-\ref{assump:overlap_iv}, the causal effect of the treatment can be identified for the subpopulation of compliers. This effect is known as the Complier Average Causal Effect (CACE) or LATE \citep{Angrist_Prischke_2008}, and captures the average effect of the treatment on the outcome $Y_i$ for units who comply with their assignment \citep{imbens_causal_2015}
\begin{align} \label{eq::CACE}
\tau_{\text{CACE}} := \frac{\mathbb{E}[Y_i \mid Z_i = 1] - \mathbb{E}[Y_i \mid Z_i = 0]}{\mathbb{E}[D_i \mid Z_i = 1] - \mathbb{E}[D_i \mid Z_i = 0]} = \frac{\text{ITT}}{\pi_C}.
\end{align}

We focus on the Conditional Complier Average Causal Effect (CCACE), which is conceptually analogous to the CATE in Equation~\eqref{eq:CATE}, but extends the CACE by accounting for heterogeneity across observed covariates. The CCACE is defined as follows
\begin{align} \label{eq::CCACE}
    \tau_{\text{CCACE}}(x) := \frac{\mathbb{E}[Y_i \mid Z_i = 1, X_i = x] - \mathbb{E}[Y_i \mid Z_i = 0, X_i = x]}{\mathbb{E}[D_i \mid Z_i = 1, X_i = x] - \mathbb{E}[D_i \mid Z_i = 0, X_i = x]} = \frac{\text{ITT}(x)}{\pi_C(x)},
\end{align}
 where the numerator $ \text{ITT}(x) $ denotes the conditional intention-to-treat effect and the denominator $ \pi_C(x) $ represents the conditional proportion of compliers \citep{angrist_imbens_iv_96}. 

\subsection{General Two-Step Approach}
\label{ch:2SP}


The pragmatic two-step framework \citep{Two_step_pragmatic_subgroup_discovery} is a model-agnostic approach for identifying interpretable subgroups.  In the first step, heterogeneous treatment effect estimates $\hat{\tau}(x)$ are obtained using a suitable machine learning method, such as causal forests, Bayesian causal forests (BCF) \citep{hahn2020bayesian}, or meta-learners \citep{kunzel_metalearners_2019}, based on the full covariate vector $X_i$. In the second step, a CART model regresses the estimated effects $\widehat{\tau}(X_i)$ on a pre-specified subset of interpretable covariates $X_i^{\text{interp}}\subseteq X_i$ to construct subgroups. Note, that treatment effects are not re-estimated within the discovered leaves, but it is possible to estimate the ATE of each subgroup by separate regression models, see Algorithm \ref{algo:komura_two_step} in Appendix \ref{app:2step_explaination}. 

\citet{bargaglibcf} propose the BCF-IV algorithm, which uses an honest sample splitting approach together with Bayesian Causal Forests in the discovery step to detect heterogeneity, followed by an inference step that estimates the treatment effect $\tau_{\text{CCACE}}(x)$ from Equation~\eqref{eq::CCACE}, within the discovered subgroups.  Our contribution builds on the BCF-IV framework, while relaxing its reliance on Bayesian first-stage estimation. Inspired by the model-agnostic idea of \citet{Two_step_pragmatic_subgroup_discovery}, we generalize the discovery step and allow any suitable machine learning method to estimate $\widehat{\text{ITT}}(x)$ and
$\hat{\pi}_C(x)$, which are used solely to construct an interpretable partition
of the covariate space via CART. While the framework is agnostic with respect to the choice of
first-stage estimator, we focus on tree-based methods, which are described in Sections~\ref{sec:RF with IV}–\ref{sec:BCF_IV}, for estimating
$\widehat{\text{ITT}}(x)$ and $\hat{\pi}_C(x)$.
Holding the resulting partition fixed, the CCACE is then
re-estimated within each leaf using IV regression on an
independent inference sample, yielding the final
$\hat{\tau}_{\text{CCACE}}(x)$.
Algorithm~\ref{algo:Two-Step} summarizes the procedure.

\FloatBarrier 
\vspace{0.5cm}
\begin{algorithm}[h] 
\caption{Two-Step Procedure with Instrumental Variables}
\label{algo:Two-Step}
\textbf{Inputs:} $N$ units $i = 1, \dots, N$ with data $(X_i, Z_i, D_i, Y_i)$ \\
\textbf{Outputs:} 
\begin{itemize}
    \item Tree structure identifying heterogeneity in causal effects
    \item  Estimates $\hat\tau_{\text{CCACE}}(x)$
\end{itemize}

\textbf{Step 0: Honest Sample Splitting}
\begin{itemize}
    \item Randomly split the data into a discovery set $\mathcal{I}_{\text{dis}}$ and an inference set $\mathcal{I}_{\text{inf}}$.
\end{itemize}
\textbf{Step 1: Discovery (on $\mathcal{I}_{\text{dis}}$)}
\begin{enumerate}
    \item Estimate $\widehat{\text{ITT}}(x)$ using an estimator described in Sections~\ref{sec:RF with IV}–\ref{sec:BCF_IV}
    \item Estimate compliance probabilities $\hat{\pi}_C(x)$ using an estimator described in Sections~\ref{sec:RF with IV}–\ref{sec:BCF_IV}
    \item Compute CCACE for $\mathcal{I}_{\text{dis}}$ \[
\hat\tau^{\,dis}_{\text{CCACE}}(x)
=
\widehat{\text{ITT}}(x)/\hat\pi_C(x)
\]
\item Fit a decision tree to $(\hat\tau^{\,dis}_{\text{CCACE}}(x), X_i)$
 to identify subgroups
\end{enumerate}

\textbf{Step 2: Inference (on $\mathcal{I}_{\text{inf}}$)}
\begin{enumerate}
    \item For each terminal node $\mathcal{L}_j$, $j=1,\ldots,J$, estimate the subgroup-specific effect $\hat{\tau}_{\text{CCACE},\mathcal{L}_j}$.
    \item Run weak-instrument tests and discard leaves with weak instrument strength from the set of reported subgroups.
\end{enumerate}
\end{algorithm}
\FloatBarrier 
\vspace{0.5cm}

Let $\mathcal{L}_1,\ldots,\mathcal{L}_J$ denote the $J$ terminal nodes or leaves, of the discovered tree, where each leaf corresponds to one discovered subgroup.
In particular, following \citet{bargaglibcf}, and noting that 
$Z_i, D_i \in \{0,1\}$, for each terminal node 
$\mathcal{L}_j$ obtained in the discovery step, the subgroup-specific CCACE is defined as
\begin{align}
{\tau}_{\text{CCACE},\mathcal{L}_j}
=
\frac{
{\mathbb{E}}[Y_i \mid Z_i=1,\, X_i\in\mathcal{L}_j]
-
{\mathbb{E}}[Y_i \mid Z_i=0,\, X_i\in\mathcal{L}_j]
}{
{\mathbb{E}}[D_i \mid Z_i=1,\, X_i\in\mathcal{L}_j]
-
{\mathbb{E}}[D_i \mid Z_i=0,\, X_i\in\mathcal{L}_j]
}.
\end{align}

Under Assumptions \ref{assump:monotonicity}-\ref{assump:Exclusion_Restriction} and with a binary $Z_i$ and $D_i$,
the corresponding IV estimator $\widehat{\tau}_{\text{CCACE},\mathcal{L}_j}$ coincides with the two-stage least squares (2SLS)
estimator for the CCACE within each subgroup (see \citet{bargaglibcf} for details).
Finally, for $X_i \in \mathcal{L}_j$, we define
\[
\widehat{\tau}_{\text{CCACE}}(X_i) =  \widehat{\tau}_{\text{CCACE},\mathcal{L}_j}.
\]

Subgroups with weak instrument strength are discarded from the set of reported valid subgroups in order to mitigate potential bias arising from weak identification or a small proportion of compliers within a leaf. To control for false discoveries due to multiple subgroup-level inference, we adjust $p$-values using Holm’s method \citep{Holm1979}, as implemented in \citet{bargaglibcf}. Alternative familywise error rate or false discovery rate controlling procedures are beyond the scope of this paper.

 The next sections illustrate how different estimators can be incorporated into the discovery step of Algorithm \ref{algo:Two-Step}. Section~\ref{sec:RF with IV} introduces a debiased robust transformed-outcome regression forest for estimating $\text{ITT}(x)$. Section~\ref{sec:CF with IV} present GRF-IV by \citet{athey_generalized_2019}, which directly targets $\tau_{\text{CCACE}}(x)$. Finally, Section~\ref{sec:BCF_IV} considers the original BCF-IV procedure of \citet{bargaglibcf}. We aim to investigate how alternative estimation strategies influence subgroup discovery within the general two-step framework.

\subsection{Transformed Outcome Regression with Instrumental Variables}
\label{sec:RF with IV}

In this section, we derive an estimator for $\text{ITT}(x)$ from Equation \eqref{eq::CCACE}, which is based on the  double/debiased machine learning (DML; \citet{chernozhukov_doubledebiased_2018}) and transformed-outcome (uplift) approaches \citep{gutierrez_causal_2017}. Before we extend this idea to the IV setting by constructing an orthogonal IV score in transformed–outcome form, we briefly discuss both methods. The DML framework builds on two central ideas from semiparametric inference: sample splitting and the use of Neyman–orthogonal score functions, which are robust to first-order nuisance estimation errors. In our setting, for the parameter of interest $\theta(x)=\text{ITT}(x)$, orthogonality is defined as follows.

\begin{definition} \label{def:neymann_orthogonal}
Let $\psi_i(\theta(x),\eta)$ denote a score function for unit $i$, identifying
$\theta(x)$ through the conditional moment condition
\[
\mathbb E[\psi_i(\theta(x),\eta)\mid X_i=x]=0,
\]
where $\eta$ is a nuisance function.
The score is called Neyman-orthogonal with respect to $\eta$
if
\[
\left.
\frac{\partial}{\partial\lambda}
\mathbb E[
\psi_i(\theta(x),\eta+\lambda h)
\mid X_i=x]
\right|_{\lambda=0}
=0
\]
for any perturbation $h(\cdot)$.
\end{definition}

To illustrate orthogonal score construction, consider the partially linear model
\begin{align}
    Y_i&=\tau D_i+g(X_i)+\varepsilon_i,\qquad &&\mathbb E[\varepsilon_i\mid X_i,D_i]=0, \notag \\
    D_i&=p(X_i)+v_i,\qquad &&\mathbb E[v_i\mid X_i]=0,
\end{align}
with $\tau$ from Equation~\eqref{eq:ATE}.
Let $p(x):=\mathbb{E}[D_i\mid X_i=x]$ and $m(x):=\mathbb{E}[Y_i\mid X_i=x]$ be the nuisance functions. Throughout this section, $m(x)$ and $p(x)$ denote the nuisance functions evaluated at a generic covariate value $x$, while $m(X_i)$ and $p(X_i)$ denote the same functions evaluated at the observed covariates of unit $i$.   Under the partially linear model, we have
\begin{align}
\tilde Y_i &:= Y_i- m(X_i), \\
\tilde D_i &:= D_i- p(X_i).
\intertext{and therefore}
Y_i-m(X_i)& =\tau(D_i-p(X_i))+\varepsilon_i.
\end{align}
Hence, estimation of $\tau$ reduces to a final regression of $\tilde Y_i$ on $\tilde D_i$, which corresponds to the orthogonal score function
\[
\psi_i(\tau,m,p)
=
\big(Y_i-m(X_i)-\tau(D_i-p(X_i))\big)\,(D_i-p(X_i))
\]
and satisfies the moment condition
$\mathbb E[\psi_i(\tau,m,p)\mid X_i=x]=0$. Moreover, this score is also orthogonal with respect to the nuisance functions $(m,p)$ in the sense of Definition \ref{def:neymann_orthogonal}.

Transformed-outcome methods, on the other hand, construct an outcome transformation $Y_i^\star$ whose conditional expectation equals the causal effect of interest. This allows standard machine learning methods to be applied directly to $Y_i^\star$, since
\begin{align*}
\mathbb{E}[Y_i^\star\mid X_i=x] = \tau(x).
\end{align*}
For example, \citet{wager_estimation_2018} note that the CATE can be estimated by applying standard machine learning methods to a transformed outcome
\begin{align}
Y_i^\star :=
\frac{D_iY_i}{\hat{p}(X_i)}
-
\frac{(1-D_i)Y_i}{1-\hat{p}(X_i)},\end{align}
where $\hat p(x)$ denotes a consistent estimator of the propensity score. 
Under Assumptions ~\ref{assump:SUTVA}-\ref{assump:overlap}, this transformation satisfies
$$\mathbb{E}[Y_i^\star\mid X_i=x] =\tau(x).$$
Furthermore,  \citet{stoffi_gnecco_CTIV}  transform $Y_i$ using the instrument $Z_i$ and the corresponding propensity score $p_Z(X_i)$ from Equation \eqref{eq:instrument_propscore} and define
\begin{align} \label{eq:IV_tot}
Y_i^{\text{IV}} := \frac{(Z_i - p_Z(X_i)) Y_i}{p_Z(X_i)(1 - p_Z(X_i))},
\end{align}
which satisfies $\mathbb E\!\left[Y_i^{\text{IV}}\mid X_i=x\right] =
\text{ITT}(x)$. 

Building on Equation \eqref{eq:IV_tot}, we consider a residualized version of the outcome based on $m(x)$
and defined the transformed outcome as
\begin{align} \label{eq:ystariv}
Y_i^{\star\text{IV}} := \frac{(Z_i - p_Z(X_i))(Y_i - m(X_i))}{p_Z(X_i)(1 - p_Z(X_i))}.
\end{align}
Note that residualizing the outcome does not affect the identification of the target parameter, since
\begin{align*} 
\mathbb E[(Z_i-p_Z(x))m(x)\mid X_i=x]=0.
\end{align*}
Following the logic of prior work on transformed outcome trees for treatment effect estimation \citep{gutierrez_causal_2017, athey_recursive_2016}, we can formalize the result as follows.
\begin{theorem}
\label{thm:dr_itt}
Let $ m(x)=\mathbb{E}[Y_i\mid X_i=x]$ and $p_Z(x)=\operatorname{Pr}(Z_i=1\mid X_i=x)$ and suppose Assumptions~\ref{assump:monotonicity}–\ref{assump:Exclusion_Restriction} hold. 
Then $Y_i^{\star\operatorname{IV}}$ satisfies
\[
\mathbb{E}[Y_i^{\star\operatorname{IV}} \mid X_i=x] = \mathbb{E}[Y_i \mid Z_i = 1, X_i=x] - \mathbb{E}[Y_i \mid Z_i = 0, X_i=x] = \operatorname{ITT}(x).
\]
\end{theorem}

\begin{proof}
See Appendix \ref{sec:proof}.
\end{proof}

Equation \eqref{eq:ystariv} is related to the orthogonal score
construction used in the DML framework. While DML orthogonalizes both $D_i$ and $Y_i$ via residual-on-residual regression and solves moment conditions explicitly, our approach use the orthogonal IV score as a pseudo–outcome. 
\begin{lemma}
\label{lem:orthogonality}
Let $\theta(x):=\text{ITT}(x)$ and define the score function
\[
\psi_i(\theta(x),m,p_Z)
:=
\frac{(Z_i-p_Z(X_i))(Y_i-m(X_i))}
     {p_Z(X_i)(1-p_Z(X_i))}
-\theta(x).
\]
The corresponding conditional moment functional is
\[
\Psi(\theta(x),m,p_Z)
:=
\mathbb{E}\!\left[
\psi_i(\theta(x),m,p_Z)
\mid X_i=x
\right].
\]
Then $\Psi(\theta(x),m,p_Z)$ is Neyman--orthogonal with respect to $m(\cdot)$. That is, for any perturbation $h(\cdot)$,
\[
\left.
\frac{\partial}{\partial \lambda}
\Psi(\theta(x),m+\lambda h,p_Z)
\right|_{\lambda=0}
=0.
\]
\end{lemma}
\begin{proof}
See Appendix \ref{sec:proof}.
\end{proof}
Taken together, Theorem~\ref{thm:dr_itt} and Lemma~\ref{lem:orthogonality} show that $Y_i^{\star\operatorname{IV}}$ can be used as a pseudo-outcome for $\operatorname{ITT}(x)$. Since \[
\mathbb{E}[Y_i^{\star\operatorname{IV}}\mid X_i=x]=\operatorname{ITT}(x),
\]
we can estimate $\operatorname{ITT}(x)$ by applying a regression forest with $Y_i^{\star\operatorname{IV}}$ as the response. The orthogonality result shows that this pseudo-outcome is locally insensitive to first-order errors in $m(\cdot)$. This estimator is then used in Step~1.1 of Algorithm~\ref{algo:Two-Step}. The resulting procedure is summarized in Algorithm~\ref{algo:TOT_IV}, with implementation details provided in Appendix~\ref{app:implementation_TSIV}.

To estimate the conditional compliance rate $\pi_C(x)$, we fit a separate regression model for $D_i$ as a function of $(X_i,Z_i)$. Let
\[
\mu_D(x,z):=\mathbb E[D_i\mid X_i=x,Z_i=z].
\]
We then estimate
\begin{align*}
 \pi_C(x)&=\mu_D(x,1)-\mu_D(x,0)
 \intertext{by}
\widehat\pi_C(x)&=\hat\mu_D(x,1)-\hat\mu_D(x,0).
\end{align*}
Combined with the transformed-outcome estimator $\widehat{\mathrm{ITT}}(x)$, this yields
\[
\hat\tau_{\mathrm{CCACE}}(x)
=
\frac{\widehat{\mathrm{ITT}}(x)}{\widehat\pi_C(x)}, 
\]
which we use to compute CCACE for $\mathcal{I}_{\text{dis}}$.

\begin{remark}
We use the term “robust” in the sense of orthogonal or debiased score construction. An theoretical extension to doubly robust estimators, which combine outcome modeling and propensity score weighting to obtain consistent estimates if at least one of the two nuisance components is correctly specified, can be found in Appendix \ref{sec:Theory}.
\end{remark}
To complete the discovery-stage estimator $\widehat\tau^{\,dis}_{\text{CCACE}}(x)$, we also specify the splitting criterion used for tree construction. Since estimation is now based on the transformed outcome $Y_i^{\star \text{IV}}$ rather than the raw outcome $Y_i$, the tree construction is also performed using $Y_i^{\star \mathrm{IV}}$ as the response variable. So the usual regression-tree splitting criterion can be written in terms of between-child-node variation in the transformed outcome. In the present setting, this gives the transformed-outcome tree criterion  
\[
\text{TOT}_{\text{IV}} := \frac{N_L (\bar{Y}_L^\star - \bar{Y}^\star)^2 + N_R (\bar{Y}_R^\star - \bar{Y}^\star)^2}{N_L + N_R},
\]
where $L$ and $R$ denote the left and right child nodes with sample sizes $N_L$ and $N_R$, respectively, and $\bar{Y}_L^\star$, $\bar{Y}_R^\star$, and $\bar{Y}^\star$ are the corresponding means of $Y_i^{\star \text{IV}}$. This criterion favors splits that target heterogeneity in $\text{ITT}(x)$.

Finally, our preceding results provide the framework for a discovery-stage estimator $\widehat\tau^{\,dis}_{\text{CCACE}}(x)$. We refer to this implementation as the Debiased Robust Random Forest with IV (DRRF-IV). Relative to the general two-step procedure in Algorithm~\ref{algo:Two-Step}, we change how Step~1 is carried out. We estimate $\operatorname{ITT}(x)$ using the debiased transformed outcome $Y_i^{\star\operatorname{IV}}$ and the conditional compliance rate $\pi_C(x)$,  to obtain $\widehat{\tau}^{\,\operatorname{dis}}_{\operatorname{CCACE}}(x)$ for the discovery sample. These estimated effects are then passed to fit a decision tree in Step~1.4 of Algorithm~\ref{algo:Two-Step}. Algorithm~\ref{algo:TOT_IV} summarizes the DRRF-IV algorithm.

\FloatBarrier 
\begin{algorithm}[h] 
\small
\caption{Debiased Transformed Outcome Tree (DRRF-IV) for Estimating $\widehat\tau^{\,dis}_{\text{CCACE}}$}
\label{algo:TOT_IV}
\KwIn{$\mathcal{I}_{dis}$, number of trees $B$}
\KwOut{Estimated CCACE for $\mathcal{I}_{\text{dis}}$}
\textbf{Step 1:} 
\vspace{0.3cm}
    \begin{itemize}
\item[1.] Estimate nuisance function $\hat p_Z(x)$ via logistic regression, $\widehat m(x)$  via regression forest and $\hat \mu_D(x,z)$ via causal forest (see Appendix \ref{app:estimation_nuisance}). 
\item[2.] Compute the debiased transformed outcome for all units:
\[
\widehat{Y}_i^{\star \text{IV}} = \frac{(Z_i - \hat{p}_Z(X_i))(Y_i - \widehat{m}(X_i))}{\hat{p}_Z(X_i)(1 - \hat{p}_Z(X_i))}.
\]
\end{itemize}
\textbf{Step 2:}
\vspace{0.3cm}
\For{$b = 1$ \KwTo $B$}{
 Split $\mathcal{I}_{dis}$ into two disjoint sets:
    \begin{itemize}
        \item A training set $\mathcal{I}_{b}^{train}$ for tree construction,
        \item An estimation set $\mathcal{I}_{b}^{est}$ for leaf-weight computation.
    \end{itemize}
    
    \textbf{Tree Growing (on $\mathcal{I}_{b}^{train}$):}
    \begin{itemize}
        \item Train a regression tree using the transformed outcome $\widehat{Y}_i^{\star \text{IV}}$ and TOT-IV splitting criterion:
        \[
        \text{TOT}_{\text{IV}} = \frac{N_L (\bar{Y}_L^\star - \bar{Y}^\star)^2 + N_R (\bar{Y}_R^\star - \bar{Y}^\star)^2}{N_L + N_R}
        \]
    \end{itemize}

    \textbf{Weighting (on $\mathcal{I}_{b}^{est}$):}
    \begin{itemize}
         \item For each $x = X_i$, $i \in \mathcal I_{dis}$, identify the corresponding leaf $L_b(x)$
        \item Assign weights:
        \[
        \alpha_{bi}(x) = \frac{\mathbf{1}(X_i \in L_b(x))}{|L_b(x)|}, \quad i \in \mathcal{I}_{b}^{est}
        \]
    \end{itemize}
}
\textbf{Final Estimation:}
\begin{itemize}
    \item Aggregate weights: $\alpha_i(x) = \frac{1}{B} \sum_{b=1}^B \alpha_{bi}(x)$
  \item Estimate the conditional intention-to-treat effect and the conditional compliance rate:
    \[
    \widehat{\operatorname{ITT}}(x)
    =
    \sum_{i \in \mathcal{I}_{\operatorname{dis}}}
    \alpha_i(x)\widehat{Y}_i^{\star\operatorname{IV}},
    \qquad
    \widehat{\pi}_C(x)
    =
    \widehat{\mu}_D(x,1)-\widehat{\mu}_D(x,0).
    \]
    \item Compute the discovery-sample CCACE estimate:
    \[
   \widehat\tau^{\,dis}_{\text{CCACE}}(x)=\frac{\widehat{\mathrm{ITT}}(x)}{\widehat{\pi}_C(x)}. \]
\end{itemize}
\end{algorithm}
\FloatBarrier

\subsection{Causal Forests with instrumental variables}
\label{sec:CF with IV}

To estimate heterogeneous treatment effects under endogeneity or irregular treatment assignment, \citet{athey_generalized_2019} extend the GRF framework to an IV setting. The GRF-IV estimator targets $\tau_{\text{CCACE}}(x)$,\footnote{In \citet{athey_generalized_2019}, this estimand is referred to as the Conditional Local Average Treatment Effect (CLATE). We maintain the notation $\tau_{\text{CCACE}}(x)$ for consistency.} under the structural model
\[Y_i = \eta(x) + \tau_{\text{CCACE}}(x) \cdot D_i + \varepsilon_i,\]
where $\varepsilon_i$ is an error term and $\eta(x)$ is a nuisance function. Under Assumptions ~\ref{assump:monotonicity}-\ref{assump:Exclusion_Restriction}, the treatment effect $\tau_{\text{CCACE}}(x)$ is identified by

\[\tau_{\text{CCACE}}(x) = \frac{\text{Cov}[Y_i, Z_i \mid X_i = x]}{\text{Cov}[D_i, Z_i \mid X_i = x]}.\]

The local estimates are obtained using gradient-based pseudo-outcomes derived from orthogonal moment conditions. Each tree solves a local estimating equation, and predictions are aggregated using forest weights obtained under an honest sample-splitting scheme. These pseudo-outcomes isolate variation in the treatment effect rather than in the outcome itself (see Algorithm~\ref{algo:Algorithm_GRF_IV} in Appendix \ref{app:GRF-IV} or \citet{athey_generalized_2019}). 

In our setting, GRF-IV is used in the discovery step of Algorithm~\ref{algo:Two-Step} to obtain unit-level estimates $\widehat{\tau}_{\text{CCACE}}(x)$. These estimates are then used in the inference step of Algorithm~\ref{algo:Two-Step} to identify interpretable subgroups. While GRF-IV itself is an established method for estimating heterogeneous treatment effects with instruments, our contribution lies in embedding it within a general two-step framework for subgroup identification. 

Both GRF-IV and DRRF-IV aim to estimate heterogeneous causal effects in IV settings, but they rely on different estimation strategies. GRF-IV is a moment-based approach that estimates $\tau_{\text{CCACE}}(x)$ by solving local orthogonal estimating equations and aggregating the resulting estimates using forest weights. In contrast, DRRF-IV, introduced in Section~\ref{sec:RF with IV}, constructs a debiased/robust transformed outcome whose conditional expectation equals the causal parameter of interest. This transformation allows heterogeneous effects to be estimated directly via standard regression forests applied to the transformed outcome.

\subsection{Bayesian instrumental variable causal forest} 
\label{sec:BCF_IV}
We now briefly describe the Bayesian Causal Forest with Instrumental Variables (BCF-IV) algorithm by \citet{bargaglibcf}, which extends the Bayesian Additive Regression Trees  (BART; \citet{Chipman_2010};  \citet{Hill}) framework to instrumental-variable settings with imperfect compliance. Following their notation, for the estimation of the conditional intention-to-treat, the outcome is modeled by including $Z_i$ into a BART ensemble
\[
Y_i = f(Z_i, X_i) + \varepsilon_i 
\;\approx\; \mathcal{T}_1(Z_i, X_i) + \cdots + \mathcal{T}_q(Z_i, X_i) + \varepsilon_i,\qquad 
\varepsilon_i \sim \mathcal{N}(0,\sigma^2), 
\]
where each $\mathcal{T}_j$ $j \in \{1 \dots q\}$ denotes a binary regression tree together with its associated 
structure and leaf parameters. 
The conditional expected value of $Y_i$ can then be written as
\[\mathbb{E}[Y_i\mid Z_i=z,X_i=x]=g(z,x),
\]
from which the conditional intention-to-treat effect is given by
\[
\operatorname{ITT}(x)=g(1,x)-g(0,x).
\]
Note, that the functions $f$ and $g$ are related but notationally distinct, $f$ denotes the BART regression function in the outcome equation, whereas $g$ denotes the corresponding conditional outcome mean. 

Adapting the model to an irregular assignment mechanism, BCF-IV decomposes the conditional mean of the outcome into a baseline component and a heterogeneous intention-to-treat component
\begin{equation}
\label{eq:BCFIV_conditional_mean}
\mathbb{E}[Y_i \mid Z_i = z, X_i = x]
=
\mu(p_Z(x), x) + \text{ITT}(x) z,   
\end{equation}
where the baseline function  $\mu(p_Z(x), x) $ captures variation in $Y_i$ unrelated to
$Z_i$. Both $\mu(\cdot)$ and $\text{ITT}$  are modeled through separate BART priors, and $\text{ITT}$ is constrained to have shallower trees.
To estimate the conditional proportion of compliers, BCF-IV fits a second BART model for treatment receipt,
\[
D_i = \ell(Z_i,X_i) + \psi_i
\;\approx\;
\tilde{\mathcal{T}}_1(Z_i,X_i) + \cdots + \tilde{\mathcal{T}}_q(Z_i,X_i) + \psi_i,
\qquad
\psi_i \sim \mathcal{N}(0,\theta^2),
\]
where $\tilde{\mathcal{T}}_j$ denotes a separate BART ensemble. As in \citet{bargaglibcf}, this notation is used as a compact representation of the treatment model. Since $D_i$ is binary, the conditional mean $\delta(z,x)=\mathbb E[D_i\mid Z_i=z,X_i=x]$ corresponds to the conditional probability $\operatorname{Pr}(D_i=1\mid Z_i=z,X_i=x)$.  The conditional proportion of compliers is then obtained as
\[
\pi_C(x) = \delta(1,x) - \delta(0,x).
\]
Combining $\text{ITT}(x)$ and the compliance rate yields 
\[
\tau_{\mathrm{CCACE}}(x)
=
\frac{\text{ITT}(x)}{\pi_C(x)}.
\]

In the original BCF-IV procedure proposed by \citet{bargaglibcf}, the estimates $\widehat{\text{ITT}}(x)$ and $\widehat{\pi}_C(x) $ obtained in the discovery sample are combined to form $\hat\tau^{\,dis}_{\text{CCACE}}(x)$. These fitted effects are then used to train a CART tree for identifying interpretable subgroups; see Algorithm~\ref{algo:BCF_IV} for details. Algorithm~\ref{algo:Two-Step} extends this idea by allowing alternative estimators of $\widehat{\text{ITT}}(x)$ and $\widehat{\pi}_C(x)$, including BCF-IV, GRF-IV, and DRRF-IV,  to be incorporated within the same subgroup discovery framework.
\vspace{0.5cm}
\FloatBarrier
\begin{algorithm}[H]
\small
\caption{BCF-IV}
\label{algo:BCF_IV}
\textbf{Inputs:} $N$ units $i = 1, \dots, N$ with data $(X_i, Z_i, D_i, Y_i)$ \\
\textbf{Outputs:}
\begin{itemize}
    \item Tree structure identifying heterogeneity in causal effects
    \item Estimates $\hat{\tau}_{\text{CCACE},\mathcal{L}_j}$ for the discovered subgroups
\end{itemize}

\textbf{Step 0: Honest Sample Splitting}
\begin{itemize}
    \item Randomly split the data into a discovery set $\mathcal{I}_{\text{dis}}$ and an inference set $\mathcal{I}_{\text{inf}}$.
\end{itemize}

\textbf{Step 1: Discovery (on $\mathcal{I}_{\text{dis}}$)}
\begin{enumerate}
    \item Estimate the conditional intention-to-treat effect $\widehat{\text{ITT}}(x)$ using BCF.
    \item Estimate the conditional compliance rate
    \[
    \hat{\pi}_C(x) = \hat{\delta}(1,x) - \hat{\delta}(0,x),
    \]
    where $\hat{\delta}(z,x)$ is obtained from a BART model for treatment receipt.
    \item Compute
    \[
    \hat{\tau}^{\,dis}_{\text{CCACE}}(x)
    =
    \frac{\widehat{\text{ITT}}(x)}{\hat{\pi}_C(x)}.
    \]
    \item Fit a decision tree to $\bigl(\hat{\tau}^{\,dis}_{\text{CCACE}}(x), X_i\bigr)$ to discover subgroups.
\end{enumerate}

\textbf{Step 2: Inference (on $\mathcal{I}_{\text{inf}}$)}
\begin{enumerate}
    \item For each terminal node $\mathcal{L}_j$, $j=1,\ldots,J$, estimate the subgroup-specific effect $\hat{\tau}_{\text{CCACE},\mathcal{L}_j}$.
    \item Run weak-instrument tests and discard leaves with weak instrument strength from the set of reported subgroups.
\end{enumerate}
\end{algorithm}
\FloatBarrier
\vspace{0.5cm}


\section{Simulation Study}
\label{ch:sim_study}
To investigate the performance of the proposed methods in the presence of imperfect compliance, we conduct a simulation study with two heterogeneity settings based on the data generating process by \citet{bargaglibcf}. The aim is to evaluate how well the proposed estimators capture heterogeneous treatment effects across varying subgroup structures and treatment effect sizes. As a benchmark, we compare performance against the established BCF-IV algorithm.

\subsection{Setup}
\label{ch:sim_study_setup}

For each design, we simulate $MC=500$ independent samples. In the baseline specification, the sample size is $N=2000$, and we additionally consider larger samples $N=10000$. For each observation $i$, the covariates $X_{i1},\dots,X_{i10}$ are generated independently from a standard normal distribution,
\[
X_{ij} \sim \mathcal{N}(0, 1), \quad \text{for } j = 1, \dots, 10.
\]
The instrument $Z_i$ is drawn independently from a binomial distribution, $Z_i$ $\sim\text{Binom}(0.5)$. We consider one-sided noncompliance. Units assigned to control never receive treatment, whereas among units assigned to treatment, compliance occurs with probability $0.75$. Formally,
\[
D_i(0) = 0, \quad D_i(1) \sim \text{Bernoulli}(0.75).
\]
Potential outcomes under control are generated as
\[
Y_i(0)\sim\mathcal{N}(0,1),
\]
and potential outcomes under treatment are given by
\[
Y_i(1)=Y_i(0)+D_i(1)\tau_{\text{CCACE}}(X_i).
\]
where $\tau_{\text{CCACE}}(X_i)$ denotes the individual-level complier average causal effect, which varies across covariate profiles. We consider two heterogeneity scenarios.

\begin{align*}
\text{\textbf{1. Strong heterogeneity scenario:}} \\\quad 
\tau_{\text{CCACE}}(X_i) &= 
\begin{cases}
  k   & \text{if } X_{i1} < -0.5 \text{ and } X_{i2} < -0.5, \\
 -k   & \text{if } X_{i1} > 0.5 \text{ and } X_{i2} > 0.5, \\
 0    & \text{otherwise}.
\end{cases} \\
\text{\textbf{2. Slight heterogeneity scenario:}} \\\quad 
\tau_{\text{CCACE}}(X_i) &= 
\begin{cases}
  k     & \text{if } X_{i1} < 0 \text{ and } X_{i2} < 0, \\
 -k     & \text{if } X_{i1} > 0 \text{ and }  X_{i2} > 0, \\
 0.5k   & \text{if } X_{i1} > 0 \text{ and }  X_{i2} < 0, \\
 -0.5k  & \text{otherwise}.
\end{cases}
\end{align*}
\noindent
Here, $k \in \{0.5, 1, 1.5, 2.5, 3, 3.5, 4\}$ denotes the treatment effect size.
 The two heterogeneity scenarios differ in the complexity of the induced subgroup structure. In the strong heterogeneity design, the covariate space contains two clearly separated active subgroups with nonzero effects and a larger null-effect region otherwise. The slight design has a more diffused heterogeneity: instead of two clearly isolated regions, the covariate space is partitioned into four subgroups with different effect magnitudes and signs, which introduces finer bands between the extreme subgroups and therefore makes subgroup identification more difficult. Overall, these treatment effect scenarios induce the following true subgroup leaves:
 
\begin{description}
    \item[Strong heterogeneity:] Only two active subgroups with nonzero effects:
    \[
    \mathcal{L}_1 = \{X_i: X_{i1} < -0.5 \text{ and } \ X_{i2} < -0.5\}, \quad
    \mathcal{L}_2 = \{X_i: X_{i1} > 0.5 \text{ and }\ X_{i2} > 0.5\}
    \]
    All other units are assigned to the null-effect leaf:
  \[
  \mathcal{L}_0 = \left\{ X_i \,\middle|\, X_i \notin \mathcal{L}_1 \cup \mathcal{L}_2 \right\}
  \]
    \item[Slight heterogeneity:] Four subgroups defined by combinations of $X_{i1}$ and $X_{i2}$:
    \[
    \begin{aligned}
    \mathcal{L}_1 &= \{X_i: X_{i1} < 0 \text{ and }\ X_{i2} < 0\}, \\
    \mathcal{L}_2 &= \{X_i: X_{i1} > 0 \text{ and }\ X_{i2} > 0\}, \\
    \mathcal{L}_3 &= \{X_i: X_{i1} > 0 \text{ and } \ X_{i2} < 0\}, \\
    \mathcal{L}_4 &= \{X_i: X_{i1} < 0\text{ and } \ X_{i2} > 0\}.
    \end{aligned}
    \]
 \end{description}   

We evaluate performance along two dimensions. First, we evaluate subgroup discovery by measuring how accurately units are assigned to their true subgroup. Since exact split points are not directly comparable across methods in designs with continuous covariates, we adopt a unit-level evaluation based on the predicted and true leaves. Details on the matching procedure are provided in Appendix~\ref{app:evaluation_criteria}. The aligned true and predicted subgroup labels can be summarized by the confusion matrix in Table~\ref{tab:confusion_metrics_multiclass_main}.
 
\begin{table}[h]
\centering
\begin{tabular}{lcc}
\toprule
 & \textbf{True $\mathcal{L}_j$} & \textbf{True not $\mathcal{L}_j$} \\
\midrule
\textbf{Predicted $\mathcal{L}_j$}     & True Positive (TP$_{\mathcal{L}_j}$) & False Positive (FP$_{\mathcal{L}_j}$) \\
\textbf{Predicted not $\mathcal{L}_j$} & False Negative (FN$_{\mathcal{L}_j}$) & True Negative (TN$_{\mathcal{L}_j}$) \\
\bottomrule
\end{tabular}
\caption{Classification scheme for subgroup detection in the multi-class setting.}
\label{tab:confusion_metrics_multiclass_main}
\end{table}

Based on these aligned subgroup labels, we compute the true detection rate (TDR), the false detection rate (FDR), and the overall detection rate (ODR). In this section, FDR denotes the false detection rate and should not be confused with the false discovery rate used in multiple-testing procedures. We refer to these quantities as detection metrics because they evaluate whether units are assigned to their correct subgroups.
\begin{align*}
\text{TDR}_{\mathcal{L}_j}
&=
\frac{\text{TP}_{\mathcal{L}_j}}
{\text{TP}_{\mathcal{L}_j}+\text{FN}_{\mathcal{L}_j}},
\qquad
\text{FDR}_{\mathcal{L}_j}
=
\frac{\text{FP}_{\mathcal{L}_j}}
{\text{FP}_{\mathcal{L}_j}+\text{TN}_{\mathcal{L}_j}}, \qquad
\text{ODR}=
\frac{1}{N}\sum_{j=0}^{J}\mathrm{TP}_{\mathcal L_j}.
\end{align*}

Second, we consider also significance-based subgroup discovery metrics. These are stricter, because a subgroup assignment is counted as successful only if the observation is assigned to the correct leaf and the estimated treatment effect is statistically significant in that leaf. In particular, we follow the original implementation of \citet{bargaglibcf}, which reports leaf-level IV estimates, their \(p\)-values, and multiple-testing-adjusted \(p\)-values. Specifically, for each discovered leaf \(\mathcal L_j\), the null hypothesis
\[
H_{0,j}: \tau_{\operatorname{CCACE},\mathcal L_j}=0
\]
is tested using the leaf-level IV estimate computed on the inference sample. The resulting \(p\)-values are adjusted across leaves using Holm's method, and a leaf is treated as significant if its adjusted \(p\)-value satisfies $p^{\operatorname{adj}}_{\mathcal L_j}<0.05.$

We use this reported significance classification when constructing the significance-based discovery metrics. To distinguish these quantities from the purely detection-based counts above, we denote the corresponding confusion-matrix entries by $*$. Thus, $\text{TP}_{\mathcal{L}_j}^*$ counts observations that are correctly assigned to $\mathcal{L}_j$ and for which the estimated treatment effect is significant, with $\text{FP}_{\mathcal{L}_j}^*$, $\text{FN}_{\mathcal{L}_j}^*$, and $\text{TN}_{\mathcal{L}_j}^*$ defined analogously. The resulting significance-based metrics are
\begin{align*}
\text{TPR}_{\mathcal{L}_j}
&=
\frac{\text{TP}_{\mathcal{L}_j}^*}
{\text{TP}_{\mathcal{L}_j}^*+\text{FN}_{\mathcal{L}_j}^*},
\qquad
\text{FPR}_{\mathcal{L}_j}
=
\frac{\text{FP}_{\mathcal{L}_j}^*}
{\text{FP}_{\mathcal{L}_j}^*+\text{TN}_{\mathcal{L}_j}^*}, \\
\text{SigRate}_{\mathcal{L}_j}
&=
\operatorname{Pr}(\text{detected significant}\mid \text{true label}=\mathcal{L}_j).
\end{align*}
Accordingly, TPR and FPR should be interpreted as significance-adjusted analogues of TDR and FDR. When the significance rate is close to one, the detection-based and significance-based measures become similar. Additionally, we evaluate estimation accuracy by using the MSE, MAE and bias. For each observation $i \in \mathcal{I}_{\mathrm{inf}} $, we compute
\begin{align}
\text{MSE}
&=
\frac{1}{|\mathcal{I}_{\mathrm{inf}}|}
\sum_{i \in \mathcal{I}_{\mathrm{inf}}}
\left(
\widehat{\tau}_{\text{CCACE},i}
-
\tau^{\star}_{\text{CCACE},i}
\right)^2, \\
\text{MAE}
&=
\frac{1}{|\mathcal{I}_{\mathrm{inf}}|}
\sum_{i \in \mathcal{I}_{\mathrm{inf}}}
\left|\widehat{\tau}_{\text{CCACE},i}
-
\tau^{\star}_{\text{CCACE},i}
\right |, \\
\text{Bias}
&=
\frac{1}{|\mathcal{I}_{\mathrm{inf}}|}
\sum_{i \in \mathcal{I}_{\mathrm{inf}}}
\left(
\widehat{\tau}_{\text{CCACE},i}
-
\tau^{\star}_{\text{CCACE},i}
\right),
\end{align}
where \(\widehat{\tau}_{\text{CCACE},i}\) denotes the estimated subgroup effect for observation \(i\), and \(\tau^{\star}_{\text{CCACE},i}\) is the corresponding CCACE based on the subgroup to which the observation was assigned by the method. Note that, $\tau^{\star}_{\text{CCACE},i}$ may differ from the true individual subgroup effect \(\tau_{\operatorname{CCACE}}(X_i)\) when the method assigns observation $i$ to an incorrect subgroup. This choice is intentional, since subgroup identification is already evaluated by the detection-based and significance-based metrics, whereas MSE, MAE, and bias are used to assess the accuracy of the effect estimate associated with the subgroup selected by the method.

\subsection{Results}
\label{ch:sim_study_results}
Figures~\ref{fig:combined_TDR_10000}–\ref{fig:combined_ODR_10000} report the average detection and significance metrics for subgroup discovery for $N=10000$, while Table~\ref{tab:mse_bias_10000_all} reports estimation accuracy for the CCACE. The corresponding tables and leafwise results are given in Appendix~\ref{app:Additional_MC}. The results for $N=2000$ are also reported in Appendix  \ref{app:Additional_MC} and lead to a similar conclusion, but subgroup discovery and estimation accuracy are generally weaker, particularly for DRRF-IV and GRF-IV. Overall, these results suggest that when the sample size is rather small, a larger effect size is needed to perform on par with BCF-IV.

We first consider the strong heterogeneity setting.  At the smallest effect size, $k=0.5$, BCF-IV achieves the strongest subgroup discovery performance, with $\operatorname{TDR}=0.786$ and $\operatorname{TPR}=0.594$, compared to $\operatorname{TDR}=0.632$ and 0.643 for DRRF-IV and GRF-IV, and $ \operatorname{TPR}=0.518$ and 0.499, respectively. As the effect size increases, all methods improve quickly. By $k=1$, $\operatorname{TPR}$ rises to 0.882 for BCF-IV, 0.848 for DRRF-IV, and 0.833 for GRF-IV, and by $k=4$ all methods attain $\operatorname{TDR}$ above 0.93 and $\operatorname{TPR}$ above 0.92. So overall, from $k=0.5$ to $k=4$, $\operatorname{TPR}$ increases by about 56\% for BCF-IV, 80\% for DRRF-IV, and 87\% for GRF-IV. False discovery rates remain higher than in the slight heterogeneity case, with $\operatorname{FDR}=0.150$ for BCF-IV and about 0.20 for the forest-based methods, but decrease substantially, while significance rates increase from 0.519, 0.427, and 0.415 at $k=0.5$ to at least 0.984, 0.986, and 0.988, respectively.  Estimation accuracy is also very similar across methods. BCF-IV has a slight advantage at small effect sizes, but the gap is small: at $k=0.5$, MSE is 0.007 for BCF-IV versus 0.009 and 0.008 for DRRF-IV and GRF-IV. For larger $k$, MSE and bias remain close across all three methods. MAE is small throughout, typically between 0.065 and 0.08, indicating that all methods estimate subgroup effects quite accurately once subgroup assignment is sufficiently reliable.

The slight heterogeneity setting follows a similar pattern. While subgroup discovery is already strong at moderate effect sizes, the gap between the methods is more visible when $k$ is small. At $k=0.5$, BCF-IV again performs best, with $\operatorname{TDR}=0.875$ and $\operatorname{TPR}=0.825$, while $\operatorname{FDR}$ and $\operatorname{FPR}$ remain close to zero.  By contrast, DRRF-IV and GRF-IV still struggle to correctly identify the true subgroups at this effect size, with $\operatorname{TDR}=0.381$ and 0.262, respectively, but improve rapidly as $k$ increases. By $k=1.5$, all three methods achieve $\operatorname{TDR}$ and $\operatorname{TPR}$ above 0.98, and from $k=2$ onward the differences are negligible. Significance rates are already high at $k=0.5$, ranging from 0.945 for BCF-IV to 0.958 and 0.956 for DRRF-IV and GRF-IV, and equal 1.0 for all methods at larger effect sizes. BCF-IV has the lowest MSE at small effect sizes, with 0.007 at $k=0.5$, compared with 0.011 for DRRF-IV and 0.018 for GRF-IV. Thus, relative to BCF-IV, the initial MSE is about 57\% higher for DRRF-IV and more than 150\% higher for GRF-IV. From about $k=1.5$ onward, however, all methods achieve essentially identical MSE levels around 0.006.

Overall, BCF-IV shows the strongest subgroup discovery and the lowest estimation error at small effect sizes. This finding is consistent with the BCF literature, which emphasizes the value of shrinkage and regularization when the signal is weak \citep{hahn2020bayesian}. BCF-IV models $\mu(p_Z(x), x)$ and $\text{ITT}(x)$ from Equation \eqref{eq:BCFIV_conditional_mean} separately and allows different levels of regularization for each component. When the true heterogeneity structure is simple or heterogeneity is modest, this leads to more stable estimates and, in our setting, lowers the risk of spurious subgroup detection. By contrast, DRRF-IV and GRF-IV are more sensitive to small effect sizes. As shown in \citet{hou2025honesty}, GRF can struggle to discover heterogeneity under honest estimation, since honesty uses less data for tree construction. When the effect size is small and, as a result, also CCACE, the forest becomes less able to detect the true subgroup structure, even though honesty helps protect against overfitting. Consequently, the GRF-based methods are less effective at identifying the covariates that truly drive heterogeneity when the signal is weak. As $k$ increases to a moderate level, however, the heterogeneity signal becomes clearer, and the performance differences across methods largely disappear.

\FloatBarrier
\begin{table}[!htbp]
\caption{MSE, bias, and MAE under strong and slight heterogeneity.}
\label{tab:mse_bias_10000_all}
\centering
\setlength{\tabcolsep}{10pt}
\resizebox{.95\linewidth}{!}{
\begin{tabular}{l r rrr @{\hskip 1.8cm} rrr}
\\[-1.8ex]\hline 
\hline \\[-1.8ex] 
\multicolumn{2}{c}{ } & \multicolumn{6}{c}{heterogeneity scenario} \\
\cmidrule(l{3pt}r{3pt}){3-8}
\multicolumn{2}{c}{ } & \multicolumn{3}{c}{strong} & \multicolumn{3}{c}{slight} \\
\cmidrule(l{3pt}r{3pt}){3-5} \cmidrule(l{3pt}r{3pt}){6-8}
Method & $k$ & MSE & Bias & MAE & MSE & Bias & MAE\\
\midrule
\multirow{8}{*}{BCF-IV}
 & 0.5 & 0.007 & -0.002 & 0.065 & 0.007 & -0.004 & 0.063\\
 & 1.0 & 0.008 & -0.003 & 0.070 & 0.006 & -0.001 & 0.061\\
 & 1.5 & 0.008 & -0.004 & 0.070 & 0.006 & -0.002 & 0.060\\
 & 2.0 & 0.008 & -0.004 & 0.070 & 0.006 & -0.003 & 0.060\\
 & 2.5 & 0.009 & -0.004 & 0.072 & 0.006 & -0.001 & 0.059\\
 & 3.0 & 0.010 & -0.004 & 0.077 & 0.006 & -0.002 & 0.061\\
 & 3.5 & 0.010 & -0.003 & 0.075 & 0.006 & -0.002 & 0.061\\
 & 4.0 & 0.011 & -0.005 & 0.079 & 0.006 &  0.000 & 0.060\\
\addlinespace
\multirow{8}{*}{DRRF-IV}
 & 0.5 & 0.009 & -0.001 & 0.070 & 0.011 &  0.000 & 0.076\\
 & 1.0 & 0.008 & -0.004 & 0.067 & 0.008 & -0.001 & 0.065\\
 & 1.5 & 0.008 & -0.002 & 0.068 & 0.006 & -0.001 & 0.062\\
 & 2.0 & 0.008 & -0.000 & 0.070 & 0.006 & -0.004 & 0.061\\
 & 2.5 & 0.008 & -0.004 & 0.071 & 0.006 & -0.002 & 0.061\\
 & 3.0 & 0.010 & -0.000 & 0.074 & 0.006 & -0.001 & 0.064\\
 & 3.5 & 0.009 & -0.001 & 0.075 & 0.006 & -0.003 & 0.060\\
 & 4.0 & 0.010 & -0.004 & 0.078 & 0.006 & -0.001 & 0.061\\
\addlinespace
\multirow{8}{*}{GRF-IV}
 & 0.5 & 0.008 & -0.005 & 0.066 & 0.018 & -0.002 & 0.089\\
 & 1.0 & 0.008 & -0.004 & 0.067 & 0.014 & -0.006 & 0.072\\
 & 1.5 & 0.008 & -0.003 & 0.069 & 0.006 & -0.001 & 0.059\\
 & 2.0 & 0.008 & -0.001 & 0.070 & 0.006 & -0.002 & 0.060\\
 & 2.5 & 0.009 & -0.001 & 0.072 & 0.006 & -0.001 & 0.061\\
 & 3.0 & 0.009 & -0.001 & 0.072 & 0.006 & -0.003 & 0.060\\
 & 3.5 & 0.010 & -0.002 & 0.077 & 0.006 & -0.001 & 0.060\\
 & 4.0 & 0.011 & -0.003 & 0.080 & 0.006 & -0.001 & 0.061\\
\bottomrule
\end{tabular}}
\caption*{\footnotesize\textit{Notes:} Entries report average estimation accuracy measures across simulation replications for $N=10000$. Lower values indicate better performance.}
\end{table}
\FloatBarrier
\FloatBarrier
\begin{figure}[!htbp]
\centering
\begin{minipage}{0.95\textwidth}
  \centering
  \caption{TDR and TPR under strong and slight heterogeneity}
  \label{fig:combined_TDR_10000}
  \includegraphics[width=0.82\textwidth]{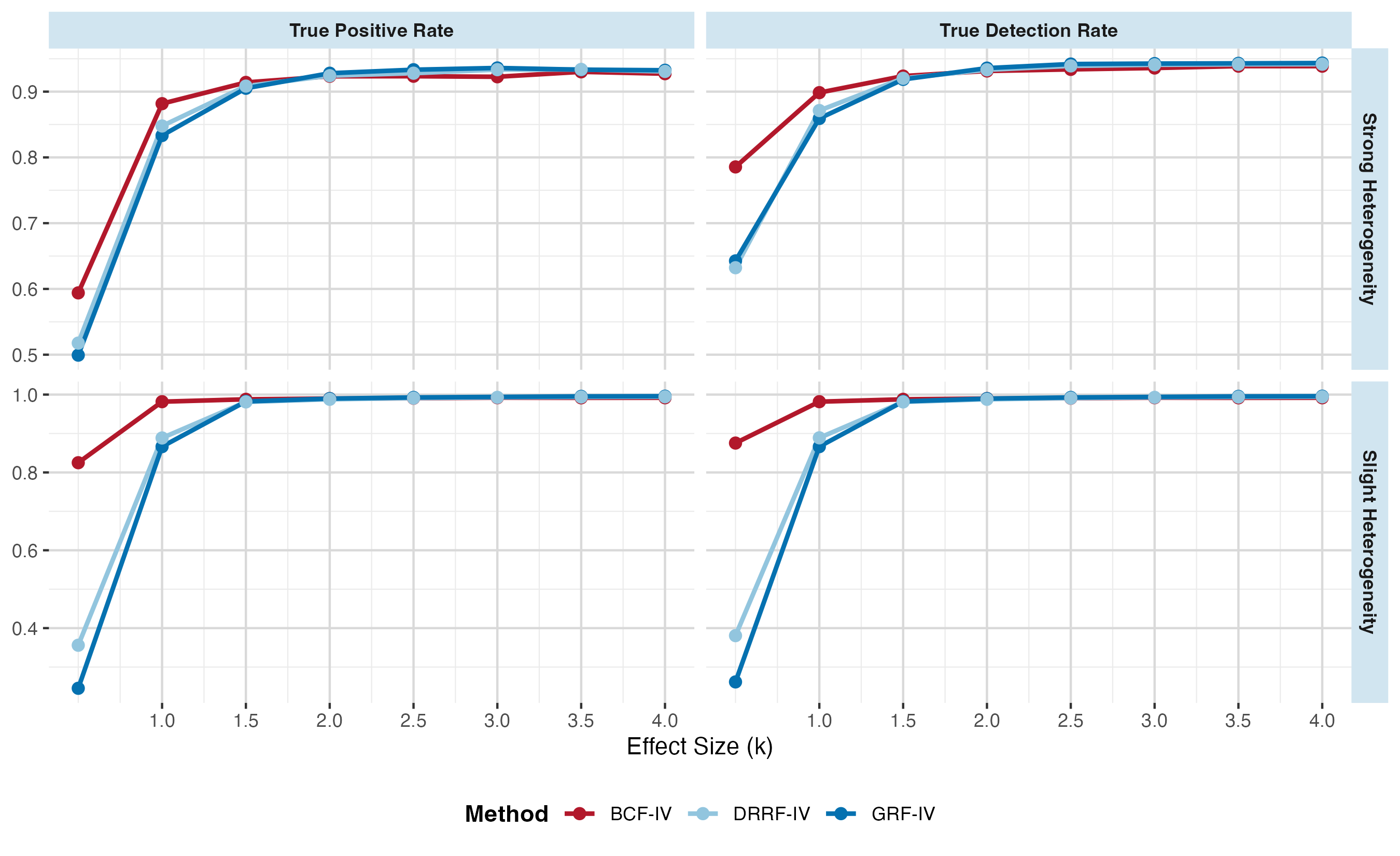}
  \caption*{\scriptsize\textit{Notes:} Results are averaged over 500 Monte Carlo replications with \(N=10{,}000\). The horizontal axis shows the treatment effect size \(k\). Higher values indicate better subgroup detection. From \(k=2\) onward, the curves nearly coincide for both heterogeneity scenarios.}
\end{minipage}
\vspace{0.6cm}
\begin{minipage}{0.95\textwidth}
  \centering
  \caption{FDR and FPR under strong and slight heterogeneity}
  \label{fig:combined_FDR_10000}
  \includegraphics[width=0.82\textwidth]{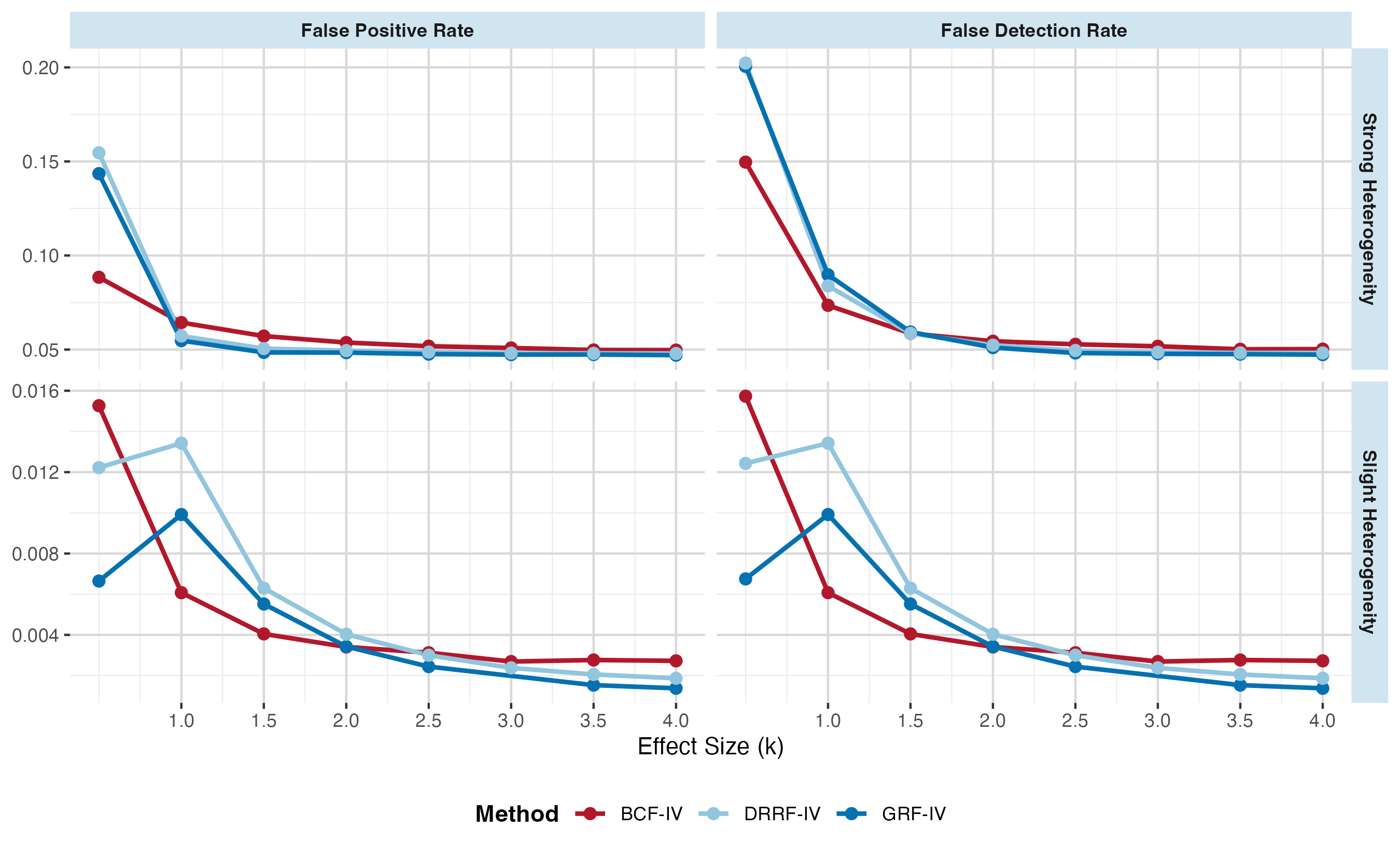}
  \caption*{\scriptsize\textit{Notes:} Results are averaged over 500 Monte Carlo replications with \(N=10{,}000\). Lower values indicate better performance. Starting from a lower level than BCF-IV, the non-Bayesian methods exhibit a spike at \(k=1\) in the slight heterogeneity case. From \(k=2\) onward, the curves nearly coincide for both heterogeneity scenarios.}
\end{minipage}
\end{figure}
\FloatBarrier

\FloatBarrier
\begin{figure}[!htbp]
  \centering
  \begin{minipage}{0.95\textwidth}
  \caption{ODR and SigRate under strong and slight heterogeneity}  
    \label{fig:combined_ODR_10000}
 \includegraphics[width=0.82\textwidth]{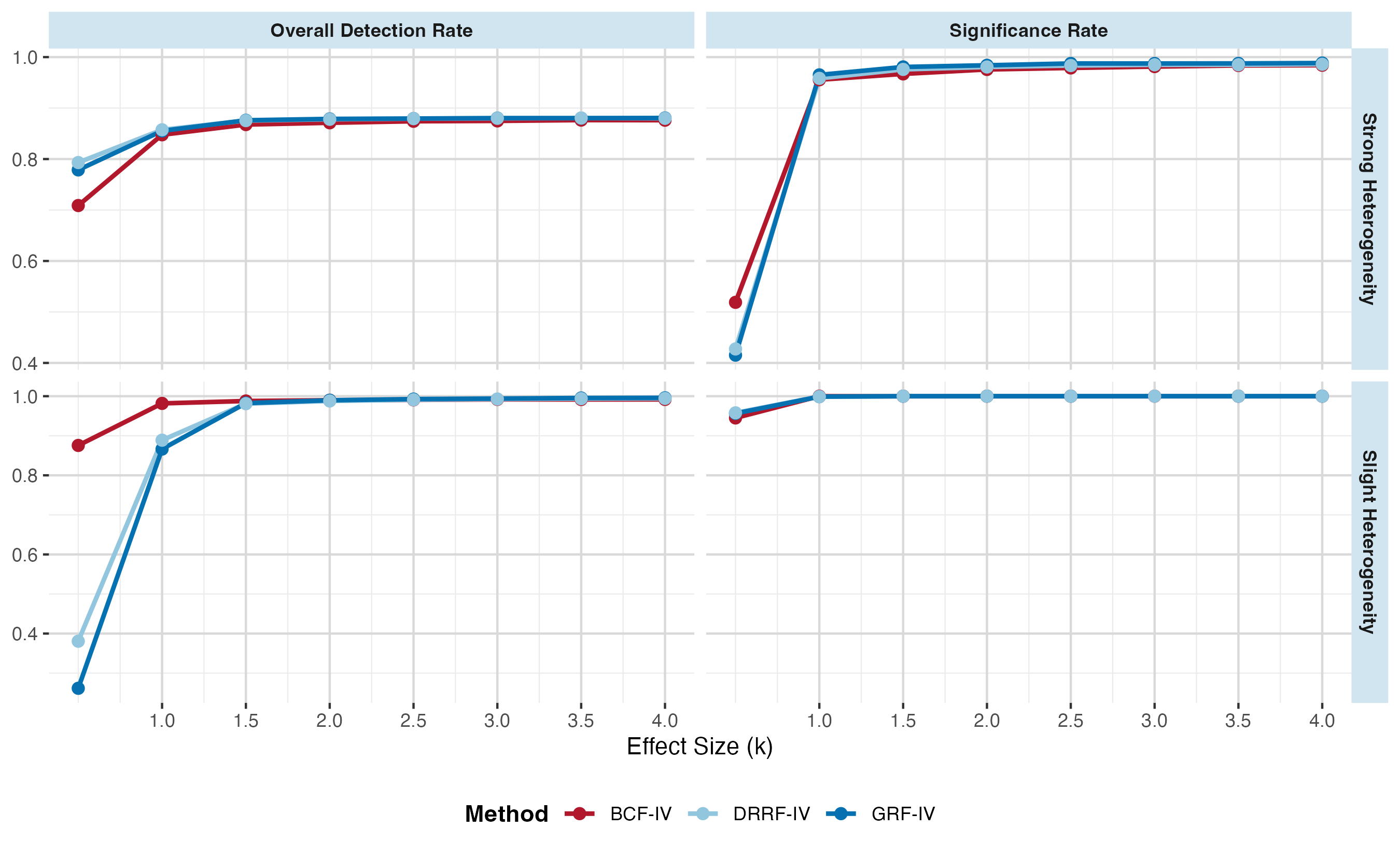} 
\caption*{\scriptsize\textit{Notes:} Results are averaged over 500 Monte Carlo replications with $N=10000$. Higher values indicate better performance. SigRate measures the frequency with which observations in a true subgroup are assigned to a leaf whose estimated effect is statistically significant. } 
\end{minipage}
\end{figure} 
\FloatBarrier
\begin{samepage}
In addition to statistical performance, we also compare the computational runtime of the three methods. Detailed runtime results are reported in Appendix~\ref{app:runtime}. Overall, GRF-IV is by far the fastest method across all considered designs. For example, when \(N=10{,}000\), average runtimes are approximately \(2.6\)--\(2.8\) seconds for GRF-IV, compared to approximately \(14.3\)--\(15.2\) seconds for DRRF-IV and \(23.5\)--\(23.8\) seconds for BCF-IV. Thus, although the modified BCF-IV implementation used in this work is faster and more stable than the original implementation by \citet{bargaglibcf}, it remains the most computationally demanding approach. The higher runtime of BCF-IV is mainly driven by the Bayesian estimation of \(\operatorname{ITT}(x)\), whereas the additional computational cost of DRRF-IV is primarily due to the transformation and rewritten splitting rule. 
\end{samepage}
\section{Empirical application}
\label{ch:emp_app}
There are indications that prompt admission to the intensive care units (ICUs) may affect patient mortality. \citet{ICU_impact} demonstrated that delays in ICU admission were associated with higher mortality, even for patients who were eventually admitted. These findings emphasize that critically ill patients benefit from early access to intensive care. However, as \citet{Keele_IV} highlight, observational studies of ICU admission face methodological challenges, as patients who are admitted promptly to the ICU often present with greater severity of illness, according to their medical characteristics. Consequently, standard regression-based methods may yield biased estimates, which has led researchers to adopt IV approaches to study the effects of critical care \citet{ICU_Zubizarreta}. 

We use data from the SPOTlight study, a prospective cohort of 13,011 adult patients from 48 UK National Health Service (NHS) hospitals, collected between November 2010 and December 2011.  Because the original data could not be publicly released, \citet{Keele_IV} provide a pseudo-dataset generated by resampling the original data with replacement until obtaining a resample that was close to the distribution of the original data. Furthermore, they applied IV methods to estimate the causal effect of ICU admission on mortality. The approach relies on ICU bed availability at the time of clinical assessment, which serves as an instrument to address unmeasured confounding in the assignment of ICU care. While their findings suggest that prompt ICU admission may reduce short-term mortality, the estimated effects were not statistically significant. Building on that work, we investigate heterogeneity in treatment effects across patient groups using our two-step algorithms. Specifically, we aim to identify patient subgroups that benefit more or less from prompt ICU admission, even when the overall average effect is close to zero. Identifying such variation may influence clinical decision-making and more efficient allocation of critical care resources. We therefore first conduct a primary single-run analysis using each method once. Second, to examine the robustness of subgroup findings, we repeat the entire procedure 100 times with different inference and discovery splits. Note that this analysis is exploratory in nature, as subgroup structures may vary across repetitions. The goal is not to make general claims but to examine patterns and highlight subgroups that appear consistently across runs.

\subsection{Data preprocessing}

Building on the pseudo-data provided by \citet{Keele_IV}, we focus on the short- to medium-term effects of ICU admission and use 28-day mortality (\texttt{dead28}) as the outcome variable. Alternative outcomes such as 7-day or 90-day mortality are not considered in this analysis. We also exclude patients older than 93 years, resulting in a sample of 12,903 patients. From the original pseudo-data, which contains 40 variables, we select the 18 covariates, which were also used in the \texttt{ivmodel} \texttt{R} package \citep{ivmodel}. The treatment is defined as ICU admission within the hospital stay (\texttt{icu\_bed}) and the instrument is determined by ICU bed availability at the time of assessment (\texttt{open\_bin}), with a value of one if fewer than four ICU beds were available and zero otherwise. Thus, following the original coding, \(\texttt{open\_bin}=1\) corresponds to lower ICU bed availability. This coding is the reverse of a convention in which \(Z_i=1\) denotes greater treatment availability, but it is kept here to remain consistent with the original data documentation. The dataset also contains patient characteristics, including demographic variables (age, sex), comorbidities (sepsis diagnosis, peri-arrest), and clinical severity measuress: the ICNARC physiological score (\texttt{icnarc\_score}), the National Early Warning Score (\texttt{news\_score}), and the Sequential Organ Failure Assessment score (\texttt{sofa\_score}). Additionally, indicators are provided for the observed level of care at assessment and the recommended level of care after assessment.

Lastly, we exclude \texttt{site} from the main analysis, since our goal is to uncover heterogeneity in treatment effects across patient characteristics, whereas including hospital indicators would shift part of the heterogeneity analysis toward between-hospital differences. As a robustness check, we conducted an additional analysis with hospital "fixed effects", and report these results in Appendix \ref{app:additional_application}. In total, the final dataset consists of 12,903 patients and 18 variables, including the outcome, treatment, instrument and baseline covariates. Summary statistics for the main variables are reported in Tables \ref{tab:summary_cont_ICU} and \ref{tab:summary_cat_ICU} in Appendix \ref{sec:Figures}. 

\subsection{Results}

The overall estimate of the CACE is approximately $-0.07$ for all methods, which is comparable to the estimates reported by \citep{Keele_IV}. This suggests that prompt ICU admission may slightly reduce 28-day mortality among compliers. While this indicates a potential survival benefit, the effect is not statistically significant. Moreover, the resulting subgroup effects are also not statistically significantly. Nevertheless, examining the subgroup patterns remains informative as the patterns of discovered subgroups varied across approaches.

In total, GRF-IV produced 22 subgroups, DRRF-IV identified sixteen subgroups, whereas BCF-IV identified four subgroups. The latter method tended to form much larger subgroups, with maximum leaf sizes encompassing up to $90.3\%$ of the sample, compared to $16.3\% $ for DRRF-IV and $14.6\%$ for GRF-IV.  Note that the weak-IV diagnostics are reported for each leaf in Appendix~\ref{app:Additional_Application_Results}, but are not used to remove leaves from the displayed trees in Figures \ref{fig:tree_drrf_iv} and \ref{fig:tree_grf_bcf_iv}. Using a \(5\%\) threshold for the weak-IV diagnostic, 9 of the 16 DRRF-IV leaves and 14 of the 22 GRF-IV leaves are flagged as weakly identified. For BCF-IV, diagnostics were available for three of the four terminal leaves, of which two were flagged as weakly identified. These diagnostics indicate that several leaf-level estimates should be interpreted cautiously and not as reliable evidence of subgroup-specific causal effects.

The estimated CACEs also varied across methods, driven by different variables that contribute to heterogeneity. Since 28-day mortality is binary, the true subgroup-specific CCACE is a causal risk difference and is bounded by \([-1,1]\) \citep{chiba2007bounds}. However, the IV estimator is a ratio of the estimated intention-to-treat effect and the estimated compliance rate. Consequently, leaf-level estimates can fall outside this logical range when the estimated compliance rate is close to zero, when the leaf contains relatively few observations, or when the first-stage relationship between the instrument and treatment receipt is weak. The weak-IV diagnostics reported in Appendix~\ref{app:Additional_Application_Results} help identify such cases. Very large leaf-level estimates should therefore be interpreted as evidence of instability in the corresponding leaves rather than as clinically meaningful effect sizes.

For example, for DRRF-IV, the raw leaf-level estimates ranged from \(-57\) among patients aged at least 77 years with ICNARC scores below 21, to \(8.6\), among patients in the same branch who additionally had low NEWS scores and for whom normal ward care was recommended after assessment. Both extreme estimates occurred in leaves flagged as weakly identified by the weak-IV diagnostic. Among the DRRF-IV leaves not flagged by this diagnostic, the estimates ranged from \(-0.21\) to \(0.56\). 
 
 GRF-IV produced estimates varying from $-1.62$ (for patients with an ICNARC score of less than 14, no sepsis, a NEWS score of less than 6 and an age greater than 78) to $6.61$ (for patients with an ICNARC score between 14 and 24, a SOFA score of less than 3 and an age greater than 83). Again, both extreme estimates occurred in leaves flagged as weakly identified by the weak-IV diagnostic. Among the GRF-IV leaves, which are not flagged, the estimates ranged from \(-0.49\) to \(0.35\). By contrast, BCF-IV yielded the narrowest range of raw leaf-level estimates, from \(-0.28\) to \(0.90\). However, BCF-IV identified only four terminal subgroups, one of which contained approximately \(90\%\) of the sample. Just for three of the four terminal leaves diagnostics were available, and only one of these was not flagged by the weak-IV diagnostic. This non-flagged leaf corresponded to patients over 63 years with SOFA scores over 5 and a recommendation for level 3 care after assessment, with an estimated CCACE of \(0.90\). 

 \FloatBarrier
 \begin{figure}[!htbp]
    \centering
    \caption{Decision tree generated by DRRF-IV}
    \label{fig:tree_drrf_iv}
    \includegraphics[width=0.95\textwidth]{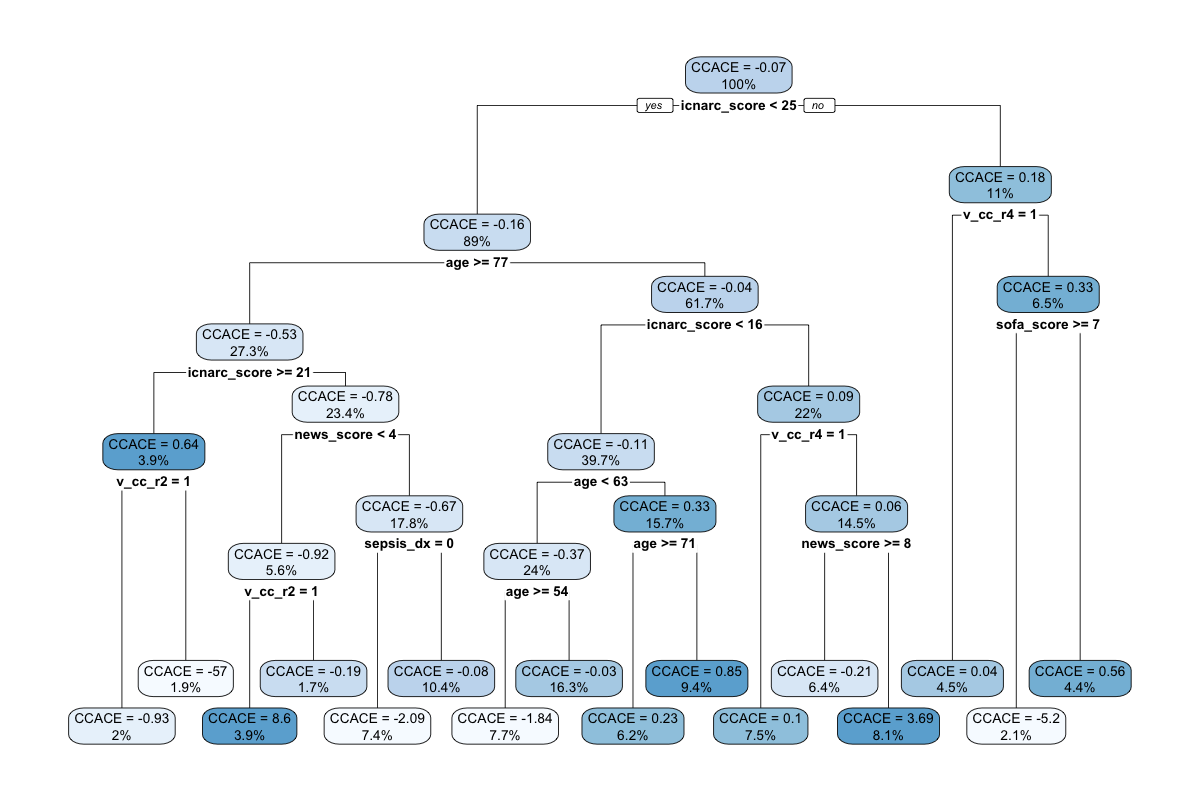}
    \caption*{\scriptsize\textit{Notes:} The figure shows the subgroup tree obtained from the DRRF-IV procedure. Terminal nodes correspond to discovered patient subgroups, with subgroup-specific CCACE estimates reported in the leaves. DRRF-IV identifies 16 terminal subgroups. Using a \(5\%\) threshold for the weak-IV diagnostic, 7 of the 16 leaves are not flagged as weakly identified. Lighter colors indicate smaller CCACE estimates, while darker colors indicate larger CCACE estimates.}
\end{figure}
\FloatBarrier

From a clinical perspective, some of these splitting patterns are plausible. Generally, patients showing signs of illness, such as sepsis, organ dysfunction (moderate SOFA scores) or intermediate ICNARC scores, or those with advanced age, tended to benefit from prompt ICU admission. However, very elderly patients and/or patients with only mild physiological derangements (low SOFA scores or intermediate ICNARC scores) exhibited an increased risk of mortality. A review of prognostic factors among very elderly ICU patients emphasizes that poor outcomes are not explained by age alone, but also by severity of illness \citep{ICU_Factors}.
Moreover, for relatively healthy or only moderately ill patients, the risks associated with ICU care (delirium, nosocomial infections, or other complications) may outweigh the potential benefits. Again, these patterns should be viewed as exploratory signals rather than confirmatory evidence.

\FloatBarrier

\begin{figure}[p]
    \centering
    \caption{Decision trees generated by GRF-IV and BCF-IV}
    \label{fig:tree_grf_bcf_iv}

    \begin{minipage}[t]{\textwidth}
        \centering
        \includegraphics[
            width=\textwidth,
            height=0.48\textheight,
            keepaspectratio
        ]{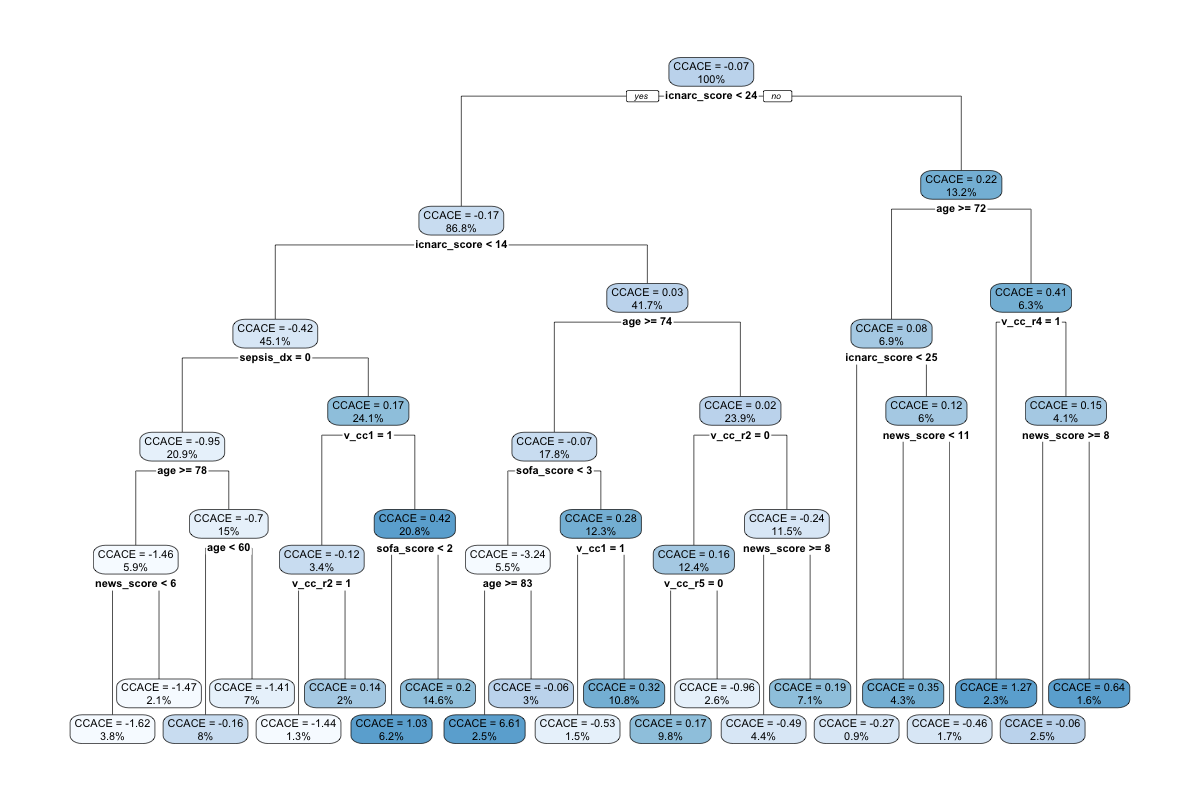}

        \smallskip
        (a) GRF-IV
    \end{minipage}

    \vspace{0.35cm}

    \begin{minipage}[t]{0.95\textwidth}
        \centering
        \includegraphics[
            width=0.65\textwidth,
            height=0.28\textheight,
            keepaspectratio
        ]{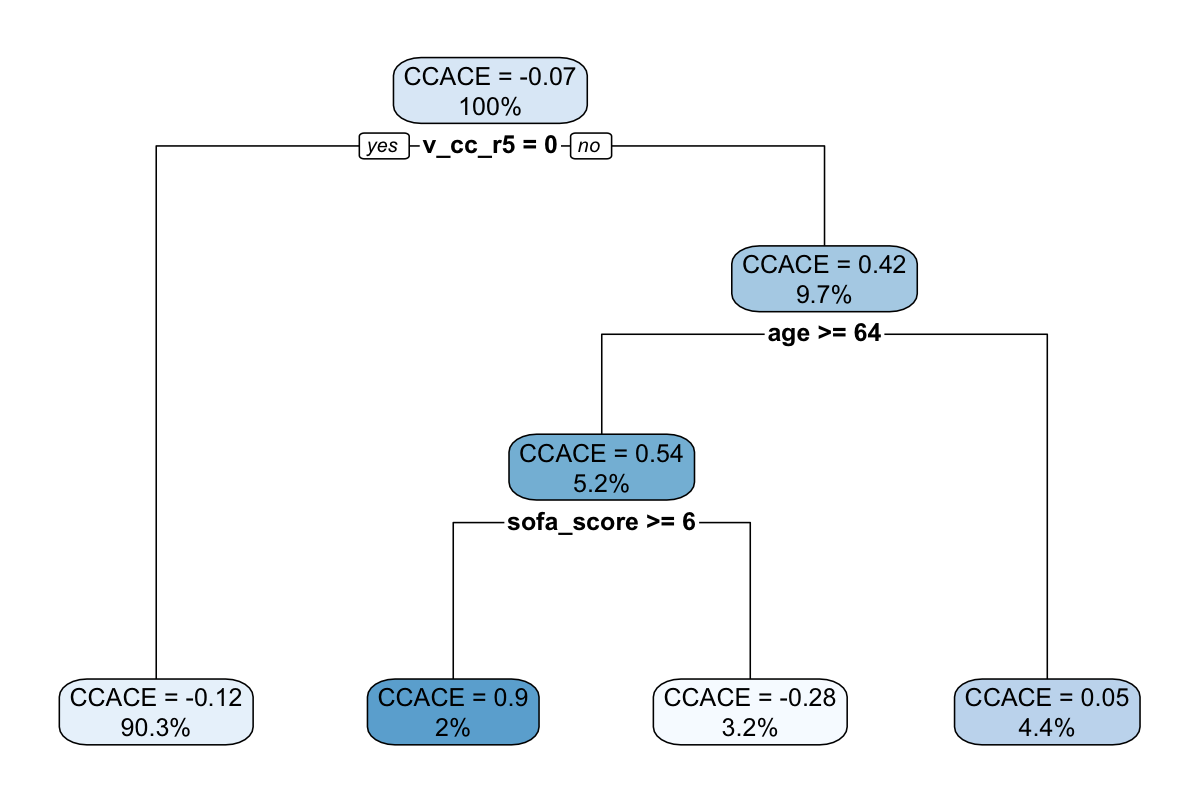}

        \smallskip
        (b) BCF-IV
    \end{minipage}

    \vspace{0.4cm}

    \begin{notes}
    Terminal nodes correspond to discovered patient subgroups, with subgroup-specific CCACE estimates reported in the leaves. GRF-IV identifies 22 terminal subgroups, of which 14 are flagged as weakly identified. BCF-IV identifies 4 terminal subgroups; diagnostics were available for three of these leaves, and only one was not flagged by the weak-IV diagnostic.
    \end{notes}
\end{figure}

\FloatBarrier

As a robustness check, we also examined which variables are most frequently used for constructing subgroups. For both BCF-IV and GRF-IV, the top three splitting variables were age and medical scores, which appears reasonable for defining subgroups. To assess this, we inspected the splitting variables across all leaves of the trees and counted their frequencies. Because BCF-IV partitions the sample into fewer and larger subgroups, it produces fewer splitting rules, so splitting variables appear less frequently by construction. The overall patterns of variable frequencies across the three methods are summarized in Figure~\ref{fig:vip_comparison} in Appendix \ref{sec:Figures}.


\section{Discussion} 
\label{ch:Discussion}

This paper examined heterogenous causal effects under imperfect compliance by developing a general two-step framework for CCACE estimation and subgroup discovery in IV settings. Building on the BCF-IV framework while relaxing its Bayesian constraints, we proposed two tree-based implementations, DRRF-IV and GRF-IV, that aim to provide more interpretable subgroup structures than existing forest-based IV estimators. In addition, we refined evaluation metrics to accommodate continuous covariates.

Our Monte Carlo comparison of DRRF-IV, GRF-IV, and BCF-IV shows that the non-Bayesian estimators perform on par with BCF-IV when treatment heterogeneity is sufficiently strong. By contrast, in low-signal settings, BCF-IV remains superior. Together, these findings suggest that the choice of estimator systematically interacts with the strength and structure of heterogeneous complier effects, which can be relevant in applied IV contexts where heterogeneity is weak or unevenly distributed across the covariate space.

The empirical application with pseudo-data from the SPOTlight study illustrates the limitations of all proposed methods. When the overall CACE is small and statistically insignificant, targeted subgroup discovery also struggled to identify subgroups with significant complier effects. Although the discovered subgroups varied across methods and showed substantial uncertainty, some patterns, such as the role of age and physiological severity, remain consistent with established clinical knowledge. These findings suggest that, in practice, hospitals may not be systematically misallocating ICU beds: prioritizing emergency admissions uniformly appears broadly consistent with the low and diffuse treatment effect signal found in the data. 

There are several directions for future research. A straightforward extension is to broaden the two-step framework to high-dimensional covariate settings or to integrate alternative machine-learning approaches, such as metalearners \citep{kunzel_metalearners_2019}, either for estimating $\operatorname{ITT}(x)$ or for directly targeting CCACE estimation in the subgroup discovery step. Another limitation of the current simulation design is that instrument is generated randomly. Future work should therefore consider settings with covariate-dependent instruments, weaker overlap, and propensity scores closer to one or zero. 
So far, we have just focused on modifying the discovery step. However, refining the inference step might also be promising. While our implementation relies on a standard IV regression, more advanced techniques, such as kernel IV regression \citep{kernelinstrumentalvariableregression}, may offer advantages in settings with non-binary instruments or treatments or general nonlinearities.  Finally, method-specific parameter tuning could also improve the performance of our methods. In the present study, all parameters were run using identical default settings to ensure comparability. A systematic investigation of tuning strategies may yield further gains in accuracy. 

\acks{I would like to thank Christoph Hanck for valuable feedback on earlier drafts of this manuscript and I am grateful to seminar participants at the RuhrMetrics Research Seminar, Statistische Woche 2025 and the EuroCIM 2025 for helpful comments and discussions. I also acknowledge partial financial support from TRR 391 Spatio-temporal Statistics for the Transition of Energy and Transport (520388526) by the Deutsche Forschungsgemeinschaft (DFG, German Research Foundation) and from the Rhine-Ruhr Center for Scientific Data Literacy (DKZ.2R) by the German Federal Ministry of Education and Research (BMBF).  During the preparation of this manuscript, I used OpenAI’s large language models (up to and including GPT-5.5) to assist with proofreading and language editing. I reviewed and verified all generated content and am solely responsible for the accuracy of the final manuscript and any remaining errors.
}

\bibliography{overall_citation}

\newpage 
\appendix
\section{Additional Monte Carlo Results}  
\label{app:Additional_MC}
\subsection{Evaluation Criteria for the Simulation Study}
\label{app:evaluation_criteria}

We evaluate subgroup discovery by comparing true and predicted subgroup labels for each unit $i$. In settings with continuous covariates, exact split points and node numbering vary across methods and simulation replications. We therefore proceed in two steps. Each $i \in \mathcal{I}_{\mathrm{inf}}$ is assigned to both a true subgroup label, which is determined by the data-generating process, and a predicted subgroup label, determined by the estimated terminal-node rules. In particular, the true subgroup structure is determined by the partition of the covariate space induced by \(X_1\) and \(X_2\), depending on whether the heterogeneity setting is slight or strong. This yields the true labels \(\mathcal{L}_0, \mathcal{L}_1, \ldots, \mathcal{L}_4\). 
 
Second, each observation is assigned a predicted terminal node label according to the estimated decision tree fitted during the discovery step. Since the numbering of estimated leaves is arbitrary, predicted terminal nodes are aligned with the true subgroup labels by comparing the sign pattern of the estimated split rules with that of the true subgroup definition. For example, a predicted node defined by conditions of the form \(X_1 < c_1\) and \(X_2 < c_2\) is matched to the true subgroup whose definition also corresponds to the lower-threshold region in both covariates. Analogously, a predicted node with \(X_1 > c_1\) and \(X_2 > c_2\) is matched to the corresponding upper-threshold subgroup. Thus, the alignment is based on the directional structure of the terminal-node rules rather than on the exact numerical split values. After the matching step, each unit has both a true subgroup label and a predicted subgroup label. Based on these labels, we evaluate subgroup discovery with the confusion matrix in Table~\ref{tab:confusion_metrics_multiclass_main}, treating each subgroup \(\mathcal{L}_j\) in a one-versus-rest classification setting. For completness we also define the significance rate for the null group in the strong heterogeneity  setting as
\begin{equation}
\text{SigRate}_{\mathcal{L}_0}
=
\mathbb{P}(\text{detected significant}\mid \text{true label}=\mathcal{L}_0).
\end{equation}
Reported values are averages over the \(MC=500\) simulation runs.

\subsection{Additional Results and Figures} 

\FloatBarrier
\begin{table}[!htbp]
\caption{Average subgroup discovery results under strong and slight heterogeneity.}
\label{tab:subgroup_results}
\centering
\setlength{\tabcolsep}{7pt}
\resizebox{.99\linewidth}{!}{
\begin{tabular}{l r rrrrrrr @{\hskip 1.2cm} rrrrrr}
\\[-1.8ex]\hline 
\hline \\[-1.8ex] 
\multicolumn{2}{c}{ } & \multicolumn{13}{c}{heterogeneity scenario} \\
\cmidrule(l{3pt}r{3pt}){3-15}
\multicolumn{2}{c}{ } & \multicolumn{7}{c}{strong} & \multicolumn{6}{c}{slight} \\
\cmidrule(l{3pt}r{3pt}){3-9} \cmidrule(l{3pt}r{3pt}){10-15}
Method & $k$ & TPR & TDR & FPR & FDR & ODR & SigRate & $\text{SigRate}_{0}$ & TPR & TDR & FPR & FDR & ODR & SigRate\\
\midrule
\multirow{8}{*}{BCF-IV}
 & 0.5 & 0.594 & 0.786 & 0.088 & 0.150 & 0.709 & 0.519 & 0.184 & 0.825 & 0.875 & 0.015 & 0.016 & 0.876 & 0.945\\
 & 1.0 & 0.882 & 0.898 & 0.064 & 0.073 & 0.848 & 0.956 & 0.215 & 0.982 & 0.982 & 0.006 & 0.006 & 0.982 & 1.000\\
 & 1.5 & 0.914 & 0.924 & 0.057 & 0.059 & 0.868 & 0.967 & 0.184 & 0.988 & 0.988 & 0.004 & 0.004 & 0.988 & 1.000\\
 & 2.0 & 0.923 & 0.931 & 0.054 & 0.054 & 0.871 & 0.976 & 0.177 & 0.990 & 0.990 & 0.003 & 0.003 & 0.990 & 1.000\\
 & 2.5 & 0.923 & 0.934 & 0.052 & 0.053 & 0.874 & 0.979 & 0.181 & 0.991 & 0.991 & 0.003 & 0.003 & 0.991 & 1.000\\
 & 3.0 & 0.923 & 0.936 & 0.051 & 0.052 & 0.875 & 0.981 & 0.190 & 0.992 & 0.992 & 0.003 & 0.003 & 0.992 & 1.000\\
 & 3.5 & 0.930 & 0.939 & 0.050 & 0.050 & 0.876 & 0.984 & 0.174 & 0.992 & 0.992 & 0.003 & 0.003 & 0.992 & 1.000\\
 & 4.0 & 0.927 & 0.939 & 0.050 & 0.050 & 0.876 & 0.984 & 0.183 & 0.992 & 0.992 & 0.003 & 0.003 & 0.992 & 1.000\\
\addlinespace
\multirow{8}{*}{DRRF-IV}
 & 0.5 & 0.518 & 0.632 & 0.155 & 0.202 & 0.793 & 0.427 & 0.141 & 0.356 & 0.381 & 0.012 & 0.012 & 0.381 & 0.958\\
 & 1.0 & 0.848 & 0.871 & 0.057 & 0.084 & 0.857 & 0.958 & 0.215 & 0.889 & 0.889 & 0.013 & 0.013 & 0.889 & 0.999\\
 & 1.5 & 0.908 & 0.920 & 0.051 & 0.059 & 0.875 & 0.976 & 0.178 & 0.981 & 0.981 & 0.006 & 0.006 & 0.981 & 1.000\\
 & 2.0 & 0.924 & 0.933 & 0.049 & 0.052 & 0.878 & 0.981 & 0.171 & 0.988 & 0.988 & 0.004 & 0.004 & 0.988 & 1.000\\
 & 2.5 & 0.928 & 0.939 & 0.049 & 0.049 & 0.879 & 0.983 & 0.178 & 0.991 & 0.991 & 0.003 & 0.003 & 0.991 & 1.000\\
 & 3.0 & 0.933 & 0.941 & 0.048 & 0.049 & 0.880 & 0.985 & 0.167 & 0.993 & 0.993 & 0.002 & 0.002 & 0.993 & 1.000\\
 & 3.5 & 0.934 & 0.942 & 0.048 & 0.048 & 0.880 & 0.985 & 0.167 & 0.994 & 0.994 & 0.002 & 0.002 & 0.994 & 1.000\\
 & 4.0 & 0.931 & 0.942 & 0.048 & 0.048 & 0.880 & 0.986 & 0.178 & 0.994 & 0.994 & 0.002 & 0.002 & 0.994 & 1.000\\
\addlinespace
\multirow{8}{*}{GRF-IV}
 & 0.5 & 0.499 & 0.643 & 0.144 & 0.201 & 0.779 & 0.415 & 0.149 & 0.245 & 0.262 & 0.007 & 0.007 & 0.262 & 0.956\\
 & 1.0 & 0.833 & 0.859 & 0.055 & 0.090 & 0.855 & 0.965 & 0.222 & 0.866 & 0.866 & 0.010 & 0.010 & 0.866 & 0.999\\
 & 1.5 & 0.905 & 0.919 & 0.048 & 0.059 & 0.876 & 0.981 & 0.181 & 0.982 & 0.982 & 0.006 & 0.006 & 0.982 & 1.000\\
 & 2.0 & 0.928 & 0.936 & 0.048 & 0.051 & 0.879 & 0.984 & 0.167 & 0.990 & 0.990 & 0.003 & 0.003 & 0.990 & 1.000\\
 & 2.5 & 0.933 & 0.942 & 0.048 & 0.048 & 0.879 & 0.988 & 0.172 & 0.993 & 0.993 & 0.002 & 0.002 & 0.993 & 1.000\\
 & 3.0 & 0.936 & 0.943 & 0.047 & 0.048 & 0.880 & 0.987 & 0.165 & 0.995 & 0.995 & 0.002 & 0.002 & 0.995 & 1.000\\
 & 3.5 & 0.933 & 0.943 & 0.047 & 0.048 & 0.880 & 0.988 & 0.174 & 0.995 & 0.995 & 0.002 & 0.002 & 0.995 & 1.000\\
 & 4.0 & 0.933 & 0.943 & 0.047 & 0.047 & 0.881 & 0.988 & 0.178 & 0.996 & 0.996 & 0.001 & 0.001 & 0.996 & 1.000\\
\bottomrule
\end{tabular}}
\caption*{\footnotesize\textit{Notes:} Entries report average subgroup discovery measures across simulation replications for \(N=10000\). Higher values of TPR, TDR, ODR, and SigRate indicate better performance, while lower values of FPR, FDR, and $\text{SigRate}_{0}$ are preferred.}
\end{table}
\FloatBarrier

\subsubsection*{Small Sample Size}
In addition to the large-sample design discussed in Section \ref{ch:sim_study}, we also consider a smaller sample size of \(N=2000\). Figures~\ref{fig:combined_TDR_2000}--\ref{fig:combined_FDR_2000} and Table~\ref{tab:subgroup_results_2000} report the corresponding subgroup discovery results.

In the slight heterogeneity scenario, BCF-IV dominates even at a low effect size ($k = 1$), achieving $\operatorname{TDR} = 0.790$ and $\operatorname{TPR} = 0.713$ with a $\operatorname{FDR}$ and $\operatorname{FPR}$ below 0.02. As $k$ increases to $4$, both $\operatorname{TDR}$ and TPR improve by approximately 25\% (to 0.982), while false rates decrease by 65\% (to 0.006). DRRF-IV and GRF-IV start weaker ($\operatorname{TDR}=0.233$ and 0.173 at $k=1$), but improve rapidly, with TDR and TPR increasing by over 300\% and converging to the BCF-IV benchmark at higher effect sizes. Significance rates are already high (above 0.90) and reach 1.0 for all methods, indicating that each approach correctly identifies the subgroups with significant treatment effects.

Under strong heterogeneity, performance gaps are more pronounced. At $k=1$, BCF-IV achieves the best performance, with $\operatorname{TDR}=0.737$ and $\operatorname{TPR}=0.501$, and $\operatorname{FDR}=0.169$, outperforming DRRF-IV and GRF-IV by roughly 25\%. With larger effect sizes ($k=4$), all methods perform better: $\operatorname{TPR}$ increases by about 78\% for BCF-IV (to 0.890) and by over 110\% for DRRF-IV and GRF-IV (to 0.853 and 0.873, respectively). False rates decrease by more than 60\%, while significance rates increase from  0.256 and 0.274 to above 0.93.
\FloatBarrier
\begin{figure}[!htbp]
  \centering
    \caption{TDR and TPR under strong and slight heterogeneity}
  \includegraphics[width=0.85\textwidth]{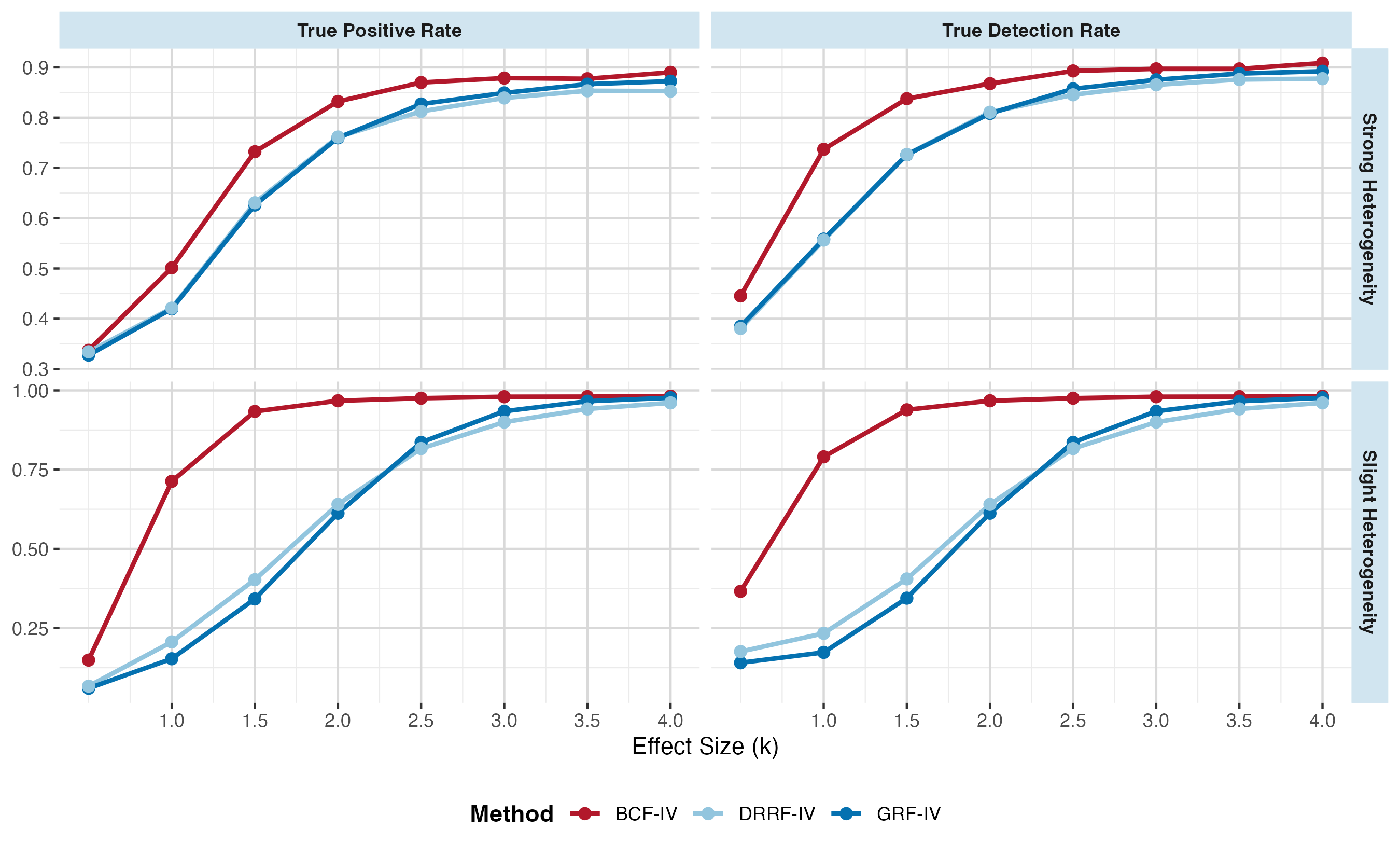} 
  \label{fig:combined_TDR_2000}
  \caption*{\scriptsize\textit{Notes:} Results are averaged over 500 Monte Carlo replications with \(N=2000\). Rows distinguish the strong and slight heterogeneity scenarios, and the horizontal axis shows the treatment effect size \(k\). Higher values indicate better performance.}
\end{figure}
\FloatBarrier

Table \ref{tab:mse_bias_mae_2000} reports the estimation performance of the local CCACE. BCF-IV achieves the lowest MSE at small effect sizes (for $k=0.5$, $0.041$ under slight heterogeneity vs.\ $0.043$ and $0.045$ for DRRF-IV and GRF-IV). However, as $k$ increases, DRRF-IV and GRF-IV close the gap and achieve similar accuracy. The bias remains small across methods (ranging from $-0.006$ and $0.002$), with BCF-IV showing a slightly larger bias at small $k$ under strong heterogeneity.
\FloatBarrier
\begin{figure}[!htbp]
  \caption{FDR and FPR under strong and slight heterogeneity}
  \label{fig:combined_FDR_2000}
  \centering
  \includegraphics[width=0.8\textwidth]{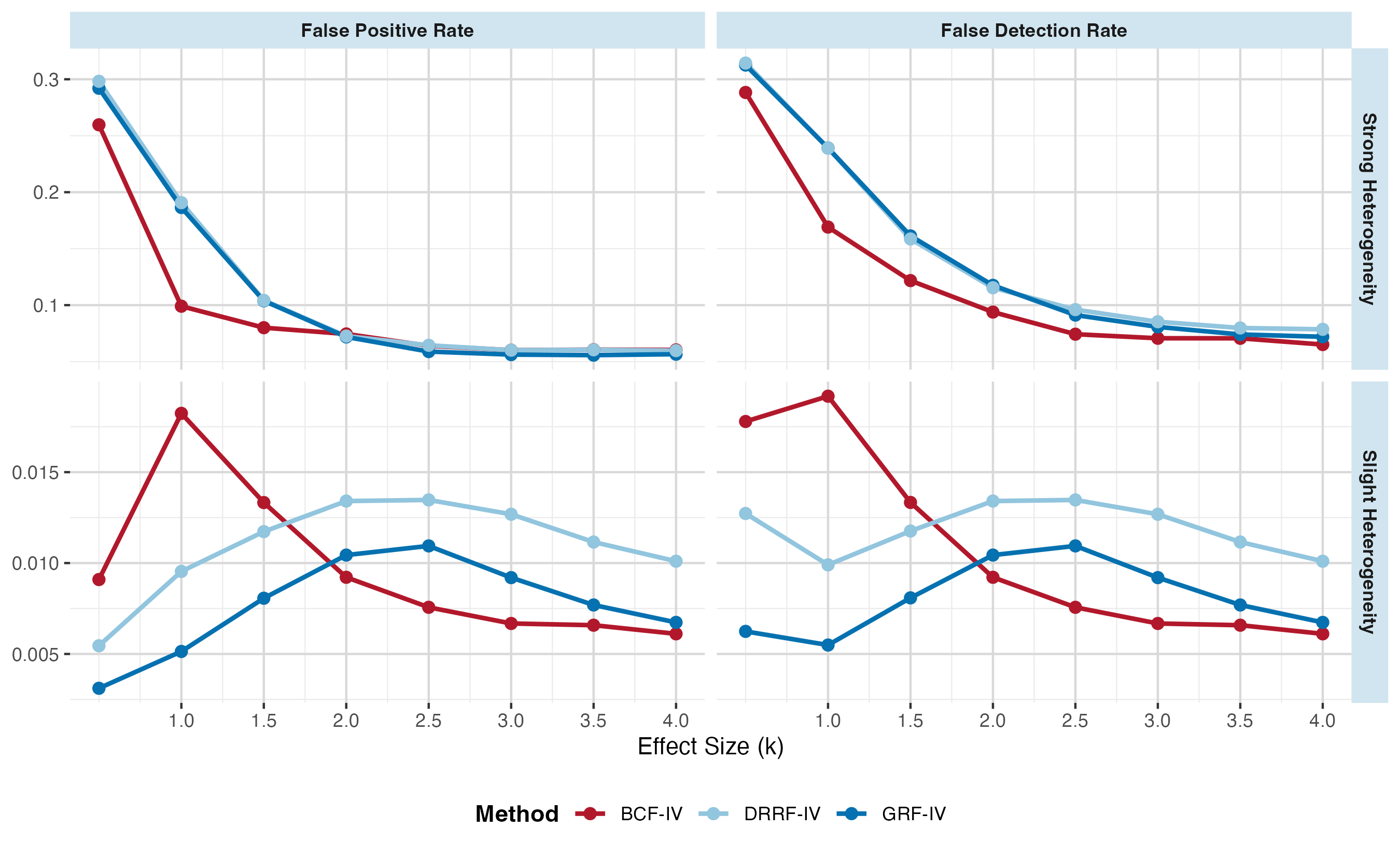} 
  \caption*{\scriptsize\textit{Notes:} Results are averaged over 500 Monte Carlo replications with \(N=2000\). Rows distinguish the strong and slight heterogeneity scenarios, and the horizontal axis shows the treatment effect size \(k\). Lower values indicate better performance.}
\end{figure}
\FloatBarrier
\FloatBarrier

\begin{figure}[!htbp]
  \centering
    \caption{ODR and significance rates for strong and slight heterogeneity}
    \label{fig:combined_ODR_2000}
  \includegraphics[width=0.8\textwidth]{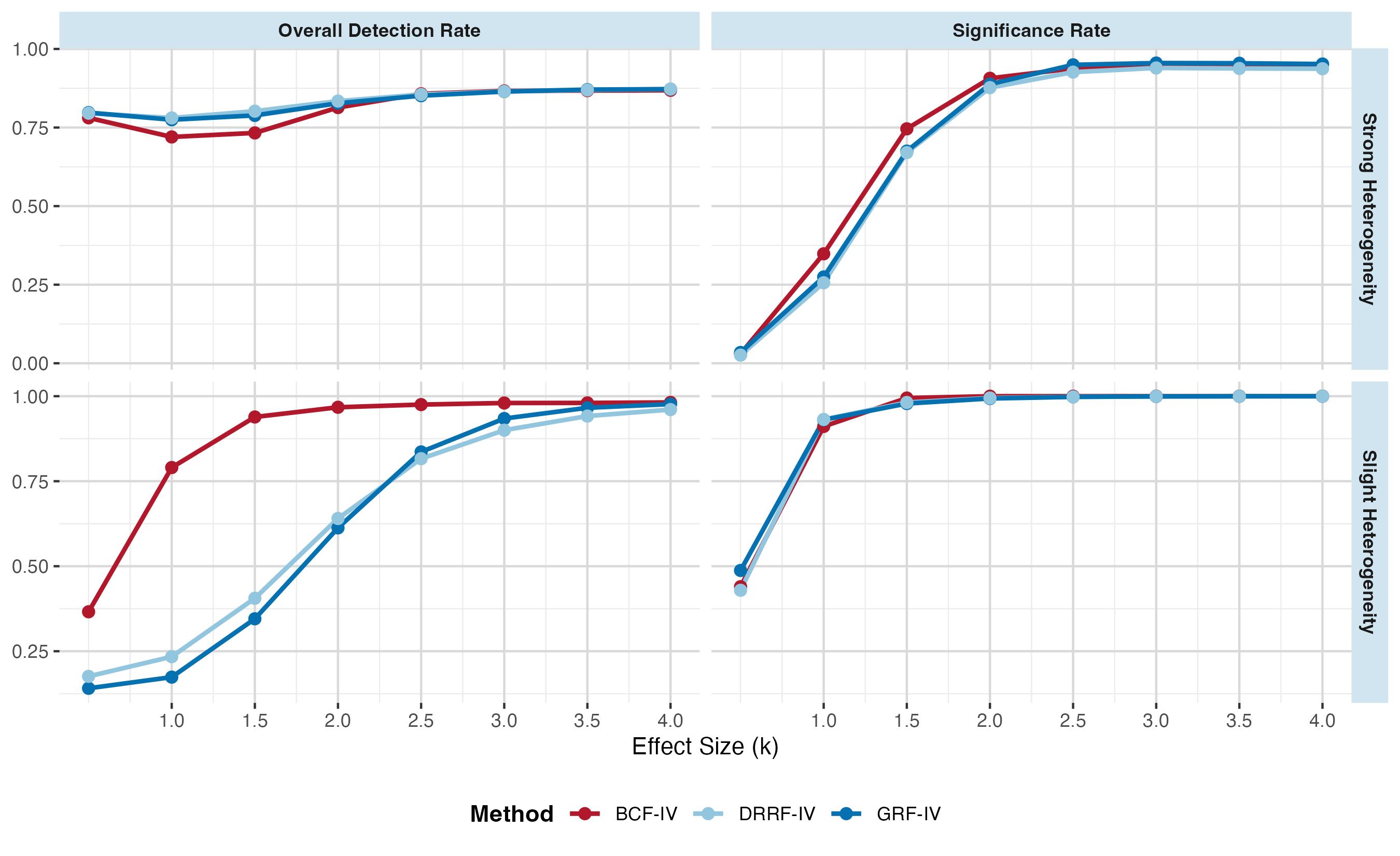} 
  \caption*{\scriptsize\textit{Notes:} Results are averaged over 500 Monte Carlo replications with \(N=2000\).  Rows distinguish the strong and slight heterogeneity scenarios, and the horizontal axis shows the treatment effect size \(k\). Higher values indicate better performance; values close to one indicate better detection of significant subgroup effects.}
\end{figure}

Overall, consistent with the large-sample results, DRRF-IV and GRF-IV perform worse than BCF-IV when the effect size is small. However, the effect size required for the methods to perform on par depends on the sample size. While for $N=10000$ a moderate effect size is sufficient, with the methods becoming comparable from about $k=1.5$ onward, for $N=2000$ this only occurs at substantially larger effect sizes, around $k=3.5$. One possible explanation is the use of honesty in GRF-IV and DRRF-IV. Due to the sample splitting, only about one quarter of the sample is available for tree construction. The data is first divided into discovery and inference samples, and then the discovery sample is split again into a training and an estimation sample. In smaller sets, this may not be enough information to detect reliable subgroups. We refer the reader to \citet{hou2025honesty} for a detailed comparison of adaptive and honest splitting.

\FloatBarrier
\begin{table}[!htbp]
\caption{Subgroup discovery results under strong and slight heterogeneity.}
\label{tab:subgroup_results_2000}
\centering
\setlength{\tabcolsep}{7pt}
\resizebox{.99\linewidth}{!}{
\begin{tabular}{l r rrrrrrr @{\hskip 1.2cm} rrrrrr}
\\[-1.8ex]\hline 
\hline \\[-1.8ex] 
\multicolumn{2}{c}{ } & \multicolumn{13}{c}{heterogeneity scenario} \\
\cmidrule(l{3pt}r{3pt}){3-15}
\multicolumn{2}{c}{ } & \multicolumn{7}{c}{strong} & \multicolumn{6}{c}{slight} \\
\cmidrule(l{3pt}r{3pt}){3-9} \cmidrule(l{3pt}r{3pt}){10-15}
Method & $k$ & TPR & TDR & FPR & FDR & ODR & SigRate & $\text{SigRate}_{0}$ & TPR & TDR & FPR & FDR & ODR & SigRate\\
\midrule
\multirow{8}{*}{BCF-IV}
 & 0.5 & 0.337 & 0.445 & 0.260 & 0.288 & 0.781 & 0.030 & 0.019 & 0.149 & 0.366 & 0.009 & 0.018 & 0.366 & 0.440\\
 & 1.0 & 0.501 & 0.737 & 0.099 & 0.169 & 0.720 & 0.348 & 0.117 & 0.713 & 0.790 & 0.018 & 0.019 & 0.790 & 0.911\\
 & 1.5 & 0.732 & 0.838 & 0.080 & 0.122 & 0.733 & 0.746 & 0.232 & 0.934 & 0.939 & 0.013 & 0.013 & 0.939 & 0.995\\
 & 2.0 & 0.832 & 0.868 & 0.074 & 0.094 & 0.814 & 0.907 & 0.236 & 0.967 & 0.967 & 0.009 & 0.009 & 0.967 & 1.000\\
 & 2.5 & 0.870 & 0.893 & 0.064 & 0.074 & 0.857 & 0.938 & 0.209 & 0.975 & 0.975 & 0.008 & 0.008 & 0.975 & 1.000\\
 & 3.0 & 0.879 & 0.897 & 0.060 & 0.071 & 0.867 & 0.948 & 0.198 & 0.980 & 0.980 & 0.007 & 0.007 & 0.980 & 1.000\\
 & 3.5 & 0.877 & 0.897 & 0.061 & 0.071 & 0.867 & 0.947 & 0.203 & 0.980 & 0.980 & 0.007 & 0.007 & 0.980 & 1.000\\
 & 4.0 & 0.890 & 0.909 & 0.060 & 0.065 & 0.868 & 0.951 & 0.203 & 0.982 & 0.982 & 0.006 & 0.006 & 0.982 & 1.000\\
\addlinespace
\multirow{8}{*}{DRRF-IV}
 & 0.5 & 0.334 & 0.381 & 0.298 & 0.314 & 0.796 & 0.026 & 0.016 & 0.067 & 0.176 & 0.005 & 0.013 & 0.176 & 0.429\\
 & 1.0 & 0.421 & 0.556 & 0.191 & 0.239 & 0.781 & 0.256 & 0.106 & 0.207 & 0.233 & 0.010 & 0.010 & 0.234 & 0.932\\
 & 1.5 & 0.631 & 0.727 & 0.104 & 0.159 & 0.802 & 0.670 & 0.216 & 0.403 & 0.405 & 0.012 & 0.012 & 0.406 & 0.981\\
 & 2.0 & 0.762 & 0.811 & 0.073 & 0.115 & 0.834 & 0.877 & 0.240 & 0.640 & 0.640 & 0.013 & 0.013 & 0.640 & 0.995\\
 & 2.5 & 0.813 & 0.845 & 0.064 & 0.096 & 0.855 & 0.926 & 0.230 & 0.816 & 0.816 & 0.013 & 0.013 & 0.816 & 0.998\\
 & 3.0 & 0.839 & 0.865 & 0.060 & 0.085 & 0.865 & 0.939 & 0.212 & 0.900 & 0.900 & 0.013 & 0.013 & 0.900 & 0.999\\
 & 3.5 & 0.853 & 0.876 & 0.060 & 0.080 & 0.869 & 0.938 & 0.201 & 0.941 & 0.941 & 0.011 & 0.011 & 0.942 & 1.000\\
 & 4.0 & 0.853 & 0.878 & 0.060 & 0.079 & 0.872 & 0.937 & 0.204 & 0.961 & 0.961 & 0.010 & 0.010 & 0.961 & 1.000\\
\addlinespace
\multirow{8}{*}{GRF-IV}
 & 0.5 & 0.328 & 0.385 & 0.292 & 0.313 & 0.797 & 0.034 & 0.027 & 0.060 & 0.141 & 0.003 & 0.006 & 0.141 & 0.487\\
 & 1.0 & 0.420 & 0.558 & 0.186 & 0.239 & 0.775 & 0.274 & 0.121 & 0.153 & 0.173 & 0.005 & 0.005 & 0.174 & 0.931\\
 & 1.5 & 0.626 & 0.726 & 0.104 & 0.161 & 0.789 & 0.675 & 0.227 & 0.342 & 0.344 & 0.008 & 0.008 & 0.345 & 0.978\\
 & 2.0 & 0.760 & 0.808 & 0.072 & 0.118 & 0.827 & 0.888 & 0.249 & 0.613 & 0.613 & 0.010 & 0.010 & 0.613 & 0.993\\
 & 2.5 & 0.827 & 0.857 & 0.059 & 0.091 & 0.851 & 0.949 & 0.229 & 0.836 & 0.836 & 0.011 & 0.011 & 0.836 & 0.998\\
 & 3.0 & 0.849 & 0.875 & 0.056 & 0.081 & 0.864 & 0.955 & 0.216 & 0.934 & 0.934 & 0.009 & 0.009 & 0.934 & 0.999\\
 & 3.5 & 0.867 & 0.888 & 0.056 & 0.074 & 0.871 & 0.954 & 0.199 & 0.966 & 0.966 & 0.008 & 0.008 & 0.966 & 1.000\\
 & 4.0 & 0.873 & 0.892 & 0.057 & 0.072 & 0.872 & 0.952 & 0.195 & 0.977 & 0.977 & 0.007 & 0.007 & 0.977 & 1.000\\
\bottomrule
\end{tabular}}
\caption*{\footnotesize\textit{Notes:} Entries report average subgroup discovery measures across simulation replications for \(N=2000\). Higher values of TPR, TDR, ODR, and SigRate indicate better performance, while lower values of FPR, FDR, and $\text{SigRate}_{0}$ are preferred.}
\end{table}
\FloatBarrier

\FloatBarrier
\begin{table}[!htbp]
\caption{MSE, Bias, and MAE under strong and slight heterogeneity.}
\label{tab:mse_bias_mae_2000}
\centering
\setlength{\tabcolsep}{10pt}
\resizebox{.95\linewidth}{!}{
\begin{tabular}{l r rrr @{\hskip 1.8cm} rrr}
\\[-1.8ex]\hline 
\hline \\[-1.8ex] 
\multicolumn{2}{c}{ } & \multicolumn{6}{c}{heterogeneity scenario} \\
\cmidrule(l{3pt}r{3pt}){3-8}
\multicolumn{2}{c}{ } & \multicolumn{3}{c}{strong} & \multicolumn{3}{c}{slight} \\
\cmidrule(l{3pt}r{3pt}){3-5} \cmidrule(l{3pt}r{3pt}){6-8}
Method & $k$ & MSE & Bias & MAE & MSE & Bias & MAE\\
\midrule
\multirow{8}{*}{BCF-IV}
 & 0.5 & 0.046 & -0.006 & 0.159 & 0.041 & -0.002 & 0.149\\
 & 1.0 & 0.041 & -0.006 & 0.153 & 0.038 & -0.003 & 0.146\\
 & 1.5 & 0.041 & -0.002 & 0.157 & 0.030 & -0.001 & 0.138\\
 & 2.0 & 0.042 &  0.001 & 0.158 & 0.030 & -0.006 & 0.138\\
 & 2.5 & 0.049 & -0.002 & 0.168 & 0.033 &  0.002 & 0.139\\
 & 3.0 & 0.049 & -0.003 & 0.171 & 0.031 & -0.003 & 0.141\\
 & 3.5 & 0.054 & -0.000 & 0.176 & 0.029 & -0.004 & 0.134\\
 & 4.0 & 0.059 & -0.008 & 0.183 & 0.031 & -0.002 & 0.141\\
\addlinespace
\multirow{8}{*}{DRRF-IV}
 & 0.5 & 0.064 & -0.003 & 0.176 & 0.043 & -0.005 & 0.153\\
 & 1.0 & 0.049 &  0.002 & 0.162 & 0.068 & -0.002 & 0.182\\
 & 1.5 & 0.042 & -0.002 & 0.159 & 0.122 & -0.015 & 0.210\\
 & 2.0 & 0.046 & -0.002 & 0.166 & 0.103 & -0.012 & 0.191\\
 & 2.5 & 0.048 & -0.003 & 0.170 & 0.082 & -0.011 & 0.168\\
 & 3.0 & 0.051 & -0.003 & 0.174 & 0.054 & -0.004 & 0.157\\
 & 3.5 & 0.054 &  0.001 & 0.179 & 0.054 & -0.003 & 0.155\\
 & 4.0 & 0.057 & -0.000 & 0.183 & 0.037 & -0.004 & 0.147\\
\addlinespace
\multirow{8}{*}{GRF-IV}
 & 0.5 & 0.037 &  0.000 & 0.148 & 0.045 &  0.002 & 0.155\\
 & 1.0 & 0.045 &  0.002 & 0.159 & 0.084 & -0.005 & 0.193\\
 & 1.5 & 0.040 & -0.000 & 0.156 & 0.137 & -0.012 & 0.223\\
 & 2.0 & 0.043 & -0.002 & 0.161 & 0.129 & -0.005 & 0.207\\
 & 2.5 & 0.047 & -0.002 & 0.168 & 0.082 & -0.003 & 0.172\\
 & 3.0 & 0.050 & -0.002 & 0.173 & 0.047 & -0.004 & 0.152\\
 & 3.5 & 0.053 &  0.001 & 0.178 & 0.035 & -0.004 & 0.144\\
 & 4.0 & 0.058 &  0.000 & 0.183 & 0.033 & -0.001 & 0.143\\
\bottomrule
\end{tabular}}
\caption*{\footnotesize\textit{Notes:} Entries report average estimation accuracy measures across simulation replications for \(N=2000\). For MSE and MAE, lower values indicate better performance; for Bias, values closer to zero are preferred.}
\end{table}
\FloatBarrier

\subsection{Additional Application Results}  
\label{app:Additional_Application_Results}
\begin{landscape}
\begin{table}[!htbp]
\centering
\scriptsize
\caption{Leaf-level DRRF-IV estimates and weak-IV diagnostics}
\label{tab:drrf_leaf_diagnostics}
\begin{tabular}{>{\raggedright\arraybackslash}p{9.2cm}rrrrr}
\\[-1.8ex]\hline 
\hline \\[-1.8ex] Subgroup rule & CCACE & \(p\)-value & Weak-IV \(p\)-value & \(\widehat{\pi}_C\) & SE \\
\midrule

\texttt{icnarc\_score} \(\geq 24.5\), \texttt{v\_cc\_r4} \(\geq 0.5\)
& 0.0418 & 0.8918 & 0.0002 & 0.1898 & 0.3068 \\

\texttt{icnarc\_score} \(\geq 24.5\), \texttt{v\_cc\_r4} \(< 0.5\), \texttt{sofa\_score} \(\geq 6.5\)
& -5.2000 & 0.8311 & 0.8271 & 0.0162 & 24.3376 \\

\texttt{icnarc\_score} \(\geq 24.5\), \texttt{v\_cc\_r4} \(< 0.5\), \texttt{sofa\_score} \(< 6.5\)
& 0.5648 & 0.1310 & 0.0019 & 0.1847 & 0.3729 \\

\texttt{icnarc\_score} \(< 24.5\), \texttt{age} \(\geq 76.5\), \texttt{icnarc\_score} \(\geq 20.5\), \texttt{v\_cc\_r2} \(\geq 0.5\)
& -0.9256 & 0.3548 & 0.1043 & -0.1048 & 0.9966 \\

\texttt{icnarc\_score} \(< 24.5\), \texttt{age} \(\geq 76.5\), \texttt{icnarc\_score} \(\geq 20.5\), \texttt{v\_cc\_r2} \(< 0.5\)
& -57.0000 & 0.9971 & 0.9971 & -0.0003 & 15826.3335 \\

\texttt{icnarc\_score} \(< 24.5\), \texttt{age} \(< 76.5\), \texttt{icnarc\_score} \(\geq 15.5\), \texttt{v\_cc\_r4} \(\geq 0.5\)
& 0.1045 & 0.5480 & 0.0000 & 0.2328 & 0.1738 \\

\texttt{icnarc\_score} \(< 24.5\), \texttt{age} \(\geq 76.5\), \texttt{icnarc\_score} \(< 20.5\), \texttt{news\_score} \(< 3.5\), \texttt{v\_cc\_r2} \(\geq 0.5\)
& 8.6044 & 0.6367 & 0.6499 & -0.0177 & 18.1969 \\

\texttt{icnarc\_score} \(< 24.5\), \texttt{age} \(\geq 76.5\), \texttt{icnarc\_score} \(< 20.5\), \texttt{news\_score} \(< 3.5\), \texttt{v\_cc\_r2} \(< 0.5\)
& -0.1887 & 0.3221 & 0.0000 & 0.4539 & 0.1897 \\

\texttt{icnarc\_score} \(< 24.5\), \texttt{age} \(\geq 76.5\), \texttt{icnarc\_score} \(< 20.5\), \texttt{news\_score} \(\geq 3.5\), \texttt{sepsis\_dx} \(< 0.5\)
& -2.0919 & 0.1714 & 0.1204 & 0.0594 & 1.5272 \\

\texttt{icnarc\_score} \(< 24.5\), \texttt{age} \(\geq 76.5\), \texttt{icnarc\_score} \(< 20.5\), \texttt{news\_score} \(\geq 3.5\), \texttt{sepsis\_dx} \(\geq 0.5\)
& -0.0848 & 0.8833 & 0.0644 & 0.0660 & 0.5774 \\

\texttt{icnarc\_score} \(< 24.5\), \texttt{age} \(< 76.5\), \texttt{icnarc\_score} \(< 15.5\), \texttt{age} \(< 62.5\), \texttt{age} \(\geq 53.5\)
& -1.8443 & 0.5028 & 0.4424 & 0.0323 & 2.7505 \\

\texttt{icnarc\_score} \(< 24.5\), \texttt{age} \(< 76.5\), \texttt{icnarc\_score} \(< 15.5\), \texttt{age} \(< 62.5\), \texttt{age} \(< 53.5\)
& -0.0341 & 0.8655 & 0.0052 & 0.0806 & 0.2013 \\

\texttt{icnarc\_score} \(< 24.5\), \texttt{age} \(< 76.5\), \texttt{icnarc\_score} \(< 15.5\), \texttt{age} \(\geq 62.5\), \texttt{age} \(\geq 70.5\)
& 0.2276 & 0.4439 & 0.0032 & 0.1405 & 0.2969 \\

\texttt{icnarc\_score} \(< 24.5\), \texttt{age} \(< 76.5\), \texttt{icnarc\_score} \(< 15.5\), \texttt{age} \(\geq 62.5\), \texttt{age} \(< 70.5\)
& 0.8526 & 0.5674 & 0.4779 & 0.0272 & 1.4900 \\

\texttt{icnarc\_score} \(< 24.5\), \texttt{age} \(< 76.5\), \texttt{icnarc\_score} \(\geq 15.5\), \texttt{v\_cc\_r4} \(< 0.5\), \texttt{news\_score} \(\geq 7.5\)
& -0.2104 & 0.5746 & 0.0109 & 0.1254 & 0.3746 \\

\texttt{icnarc\_score} \(< 24.5\), \texttt{age} \(< 76.5\), \texttt{icnarc\_score} \(\geq 15.5\), \texttt{v\_cc\_r4} \(< 0.5\), \texttt{news\_score} \(< 7.5\)
& 3.6916 & 0.8754 & 0.8750 & 0.0067 & 23.5262 \\
\bottomrule
\end{tabular}
\caption*{\scriptsize\textit{Notes:} The table reports leaf-level DRRF-IV estimates for the empirical application. The weak-IV \(p\)-value is the first-stage relevance diagnostic from the leaf-level IV regression. Large weak-IV \(p\)-values indicate weak evidence of first-stage relevance. \(\widehat{\pi}_C\) denotes the estimated compliance rate within the leaf. Adjusted \(p\)-values are omitted because they are equal to 1 for all leaves.}
\end{table}
\end{landscape}

\begin{landscape}
\begin{table}[!htbp]
\centering
\scriptsize
\caption{Leaf-level GRF-IV estimates and weak-IV diagnostics}
\label{tab:grf_leaf_diagnostics}
\begin{tabular}{>{\raggedright\arraybackslash}p{9.2cm}rrrrr}
\\[-1.8ex]\hline 
\hline \\[-1.8ex] 
Subgroup rule & CCACE & \(p\)-value & Weak-IV \(p\)-value & \(\widehat{\pi}_C\) & SE \\
\midrule

\(23.5 \leq \texttt{icnarc\_score} < 24.5\), \texttt{age} \(\geq 71.5\)
& -0.2749 & 0.6598 & 0.1040 & 0.6071 & 0.6209 \\

\texttt{icnarc\_score} \(\geq 23.5\), \texttt{age} \(< 71.5\), \texttt{v\_cc\_r4} \(\geq 0.5\)
& 1.2678 & 0.3433 & 0.2134 & 0.5685 & 1.3335 \\

\texttt{icnarc\_score} \(\geq 24.5\), \texttt{age} \(\geq 71.5\), \texttt{news\_score} \(< 10.5\)
& 0.3462 & 0.2210 & 0.0002 & 0.6173 & 0.2823 \\

\texttt{icnarc\_score} \(\geq 24.5\), \texttt{age} \(\geq 71.5\), \texttt{news\_score} \(\geq 10.5\)
& -0.4633 & 0.3665 & 0.0253 & 0.6091 & 0.5109 \\

\texttt{icnarc\_score} \(\geq 23.5\), \texttt{age} \(< 71.5\), \texttt{v\_cc\_r4} \(< 0.5\), \texttt{news\_score} \(\geq 7.5\)
& -0.0577 & 0.9035 & 0.0152 & 0.5915 & 0.4746 \\

\texttt{icnarc\_score} \(\geq 23.5\), \texttt{age} \(< 71.5\), \texttt{v\_cc\_r4} \(< 0.5\), \texttt{news\_score} \(< 7.5\)
& 0.6420 & 0.6968 & 0.5197 & 0.5300 & 1.6427 \\

\texttt{icnarc\_score} \(< 13.5\), \texttt{sepsis\_dx} \(< 0.5\), \texttt{age} \(\geq 77.5\), \texttt{news\_score} \(< 5.5\)
& -1.6189 & 0.1499 & 0.0800 & 0.5451 & 1.1206 \\

\texttt{icnarc\_score} \(< 13.5\), \texttt{sepsis\_dx} \(< 0.5\), \texttt{age} \(\geq 77.5\), \texttt{news\_score} \(\geq 5.5\)
& -1.4668 & 0.2855 & 0.2091 & 0.5294 & 1.3678 \\

\texttt{icnarc\_score} \(< 13.5\), \texttt{sepsis\_dx} \(< 0.5\), \texttt{age} \(< 59.5\)
& -0.1607 & 0.5878 & 0.0477 & 0.5184 & 0.2963 \\

\texttt{icnarc\_score} \(< 13.5\), \texttt{sepsis\_dx} \(< 0.5\), \(59.5 \leq \texttt{age} < 77.5\)
& -1.4095 & 0.2804 & 0.1587 & 0.5320 & 1.3042 \\

\texttt{icnarc\_score} \(< 13.5\), \texttt{sepsis\_dx} \(\geq 0.5\), \texttt{v\_cc1} \(\geq 0.5\), \texttt{v\_cc\_r2} \(\geq 0.5\)
& -1.4384 & 0.4611 & 0.3909 & 0.5581 & 1.9427 \\

\texttt{icnarc\_score} \(< 13.5\), \texttt{sepsis\_dx} \(\geq 0.5\), \texttt{v\_cc1} \(\geq 0.5\), \texttt{v\_cc\_r2} \(< 0.5\)
& 0.1352 & 0.5547 & 0.0004 & 0.6641 & 0.2282 \\

\texttt{icnarc\_score} \(< 13.5\), \texttt{sepsis\_dx} \(\geq 0.5\), \texttt{v\_cc1} \(< 0.5\), \texttt{sofa\_score} \(< 1.5\)
& 1.0277 & 0.5413 & 0.4674 & 0.4710 & 1.6812 \\

\texttt{icnarc\_score} \(< 13.5\), \texttt{sepsis\_dx} \(\geq 0.5\), \texttt{v\_cc1} \(< 0.5\), \texttt{sofa\_score} \(\geq 1.5\)
& 0.2031 & 0.7952 & 0.2826 & 0.5042 & 0.7824 \\

\(13.5 \leq \texttt{icnarc\_score} < 23.5\), \texttt{age} \(\geq 82.5\), \texttt{sofa\_score} \(< 2.5\)
& 6.6055 & 0.5493 & 0.5528 & 0.4214 & 11.0072 \\

\(13.5 \leq \texttt{icnarc\_score} < 23.5\), \(73.5 \leq \texttt{age} < 82.5\), \texttt{sofa\_score} \(< 2.5\)
& -0.0612 & 0.9305 & 0.1511 & 0.5204 & 0.7010 \\

\(13.5 \leq \texttt{icnarc\_score} < 23.5\), \texttt{age} \(\geq 73.5\), \texttt{sofa\_score} \(\geq 2.5\), \texttt{v\_cc1} \(\geq 0.5\)
& -0.5286 & 0.7791 & 0.4941 & 0.4787 & 1.8785 \\

\(13.5 \leq \texttt{icnarc\_score} < 23.5\), \texttt{age} \(\geq 73.5\), \texttt{sofa\_score} \(\geq 2.5\), \texttt{v\_cc1} \(< 0.5\)
& 0.3220 & 0.1779 & 0.0000 & 0.5714 & 0.2388 \\

\(13.5 \leq \texttt{icnarc\_score} < 23.5\), \texttt{age} \(< 73.5\), \texttt{v\_cc\_r2} \(< 0.5\), \texttt{v\_cc\_r5} \(< 0.5\)
& 0.1677 & 0.3031 & 0.0000 & 0.6010 & 0.1627 \\

\(13.5 \leq \texttt{icnarc\_score} < 23.5\), \texttt{age} \(< 73.5\), \texttt{v\_cc\_r2} \(< 0.5\), \texttt{v\_cc\_r5} \(\geq 0.5\)
& -0.9552 & 0.8977 & 0.6476 & 0.4706 & 7.4198 \\

\(13.5 \leq \texttt{icnarc\_score} < 23.5\), \texttt{age} \(< 73.5\), \texttt{v\_cc\_r2} \(\geq 0.5\), \texttt{news\_score} \(\geq 7.5\)
& -0.4871 & 0.2797 & 0.0167 & 0.5532 & 0.4497 \\

\(13.5 \leq \texttt{icnarc\_score} < 23.5\), \texttt{age} \(< 73.5\), \texttt{v\_cc\_r2} \(\geq 0.5\), \texttt{news\_score} \(< 7.5\)
& 0.1922 & 0.8072 & 0.2735 & 0.5054 & 0.7869 \\

\bottomrule
\end{tabular}
\caption*{\scriptsize\textit{Notes:} The table reports leaf-level GRF-IV estimates for the empirical application. Redundant split conditions are omitted from the subgroup rules for readability. Adjusted \(p\)-values are omitted because they are equal to 1 for all leaves.}
\end{table}
\end{landscape}
\begin{table}[!htbp]
\centering
\scriptsize
\caption{Leaf-level BCF-IV estimates and weak-IV diagnostics}
\label{tab:bcf_leaf_diagnostics}
\begin{tabular}{>{\raggedright\arraybackslash}p{6cm}rrrrr}
\\[-1.8ex]\hline 
\hline \\[-1.8ex] 
Subgroup rule & CCACE & \(p\)-value & Weak-IV \(p\)-value & \(\widehat{\pi}_C\) & SE \\
\midrule

\texttt{v\_cc\_r5} \(\geq 0.5\), \texttt{age} \(< 63.5\)
& 0.0453 & 0.9789 & 0.1363 & 0.0325 & 1.7103 \\

\texttt{v\_cc\_r5} \(\geq 0.5\), \texttt{age} \(\geq 63.5\), \texttt{sofa\_score} \(\geq 5.5\)
& 0.9044 & 0.3226 & 0.0088 & 0.1051 & 0.9109 \\

\texttt{v\_cc\_r5} \(\geq 0.5\), \texttt{age} \(\geq 63.5\), \texttt{sofa\_score} \(< 5.5\)
& -0.2830 & 0.9359 & 0.4867 & 0.0199 & 3.5129 \\

\bottomrule
\end{tabular}
\caption*{\scriptsize\textit{Notes:} The table reports leaf-level BCF-IV estimates for the empirical application. Diagnostics were available for three of the four terminal leaves in the displayed BCF-IV tree. Adjusted \(p\)-values are omitted because they are equal to 1 for all evaluated leaves.}
\end{table}
\subsection*{Additional Figures and Statistics}  \label{sec:Figures}

\begin{table}[!htbp] \centering 
  \caption{Summary statistics for continuous variables for the ICU dataset} 
  \label{tab:summary_cont_ICU} 
\begin{tabular}{@{\extracolsep{5pt}}lccccc} 
\\[-1.8ex]\hline 
\hline \\[-1.8ex] 
Statistic & \multicolumn{1}{c}{N} & \multicolumn{1}{c}{Mean} & \multicolumn{1}{c}{St. Dev.} & \multicolumn{1}{c}{Min} & \multicolumn{1}{c}{Max} \\ 
\hline \\[-1.8ex] 
Age & 12,903 & 64.96 & 17.71 & 18 & 93 \\ 
ICNARC physiological score & 12,903 & 15.10 & 7.34 & 0 & 53 \\ 
National Early Warning Score (NEWS) & 12,903 & 6.24 & 3.10 & 0 & 20 \\ 
SOFA score & 12,903 & 3.14 & 2.18 & 0 & 14 \\ 
\hline \\[-1.8ex] 
\end{tabular} 
\caption*{\footnotesize\textit{Notes:} 
The ICNARC physiological score measures acute illness severity using criteria from the Intensive Care National Audit and Research Centre. The National Early Warning Score (NEWS) summarizes early signs of clinical deterioration based on routinely measured physiological variables. The SOFA score, or Sequential Organ Failure Assessment score, measures the extent of organ dysfunction.}
\end{table}
\begin{table}[htbp]
\centering
\caption{Summary statistics for categorical variables fot the ICU dataset.}
\label{tab:summary_cat_ICU}
\begin{tabular}{@{\extracolsep{5pt}}lcc} 
\\[-1.8ex]\hline 
\hline \\[-1.8ex] 
Variable & No & Yes \\
\hline
Less than 4 ICU beds available & 6026 & 6877 \\
Admitted to an ICU bed & 7890 & 5013 \\
Male & 5894 & 7009 \\
Sepsis diagnosis & 5018 & 7885 \\
Peri-arrest diagnosis & 12256 & 647 \\
Level of care at assessment: Level 0 (normal ward care) & 11210 & 1693 \\
Level of care at assessment: Level 1 (ward care with support) & 4061 & 8842 \\
Level of care at assessment: Level 2 (high dependency unit) & 10738 & 2165 \\
Level of care at assessment: Level 3 (ICU care) & 12774 & 129 \\
Recommended level of care: Level 0 (normal ward care) & 12024 & 879 \\
Recommended level of care: Level 1 (ward care with support) & 5909 & 6994 \\
Recommended level of care: Level 2 (high dependency unit) & 9194 & 3709 \\
Recommended level of care: Level 3 (ICU care) & 11689 & 1214 \\
\hline
\end{tabular}
\begin{flushleft}
\footnotesize
\textit{Notes:} Values report counts by category. For binary indicators, “Yes” denotes category 1 and “No” denotes category 0.
\end{flushleft}
\end{table}

\FloatBarrier

\begin{figure}[htbp]
    \centering
    \caption{Variable importance plots for the DRRF-IV, GRF-IV and BCF-IV.}
    \label{fig:vip_comparison}

    \begin{minipage}[t]{0.55\textwidth}
        \centering
        \includegraphics[width=\linewidth]{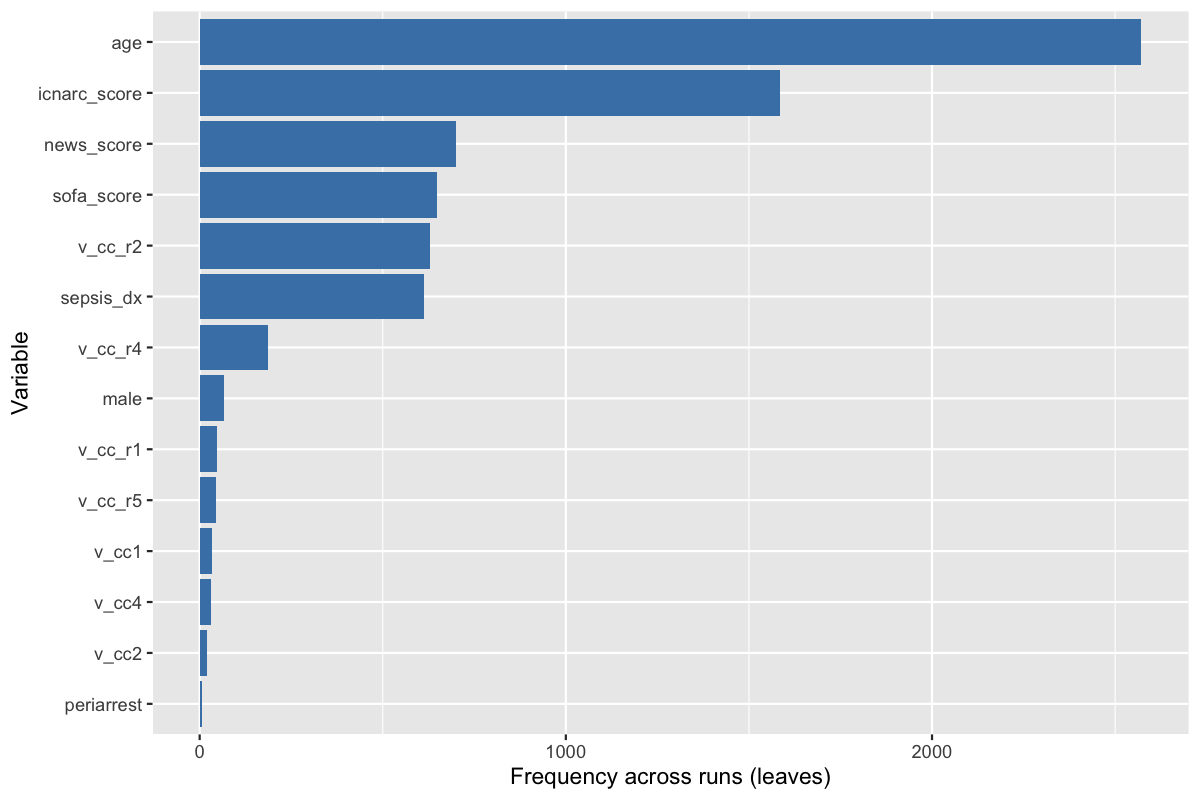}

        \smallskip
        (a) DRRF-IV
    \end{minipage}

    \vspace{0.75em}

    \begin{minipage}[t]{0.55\textwidth}
        \centering
        \includegraphics[width=\linewidth]{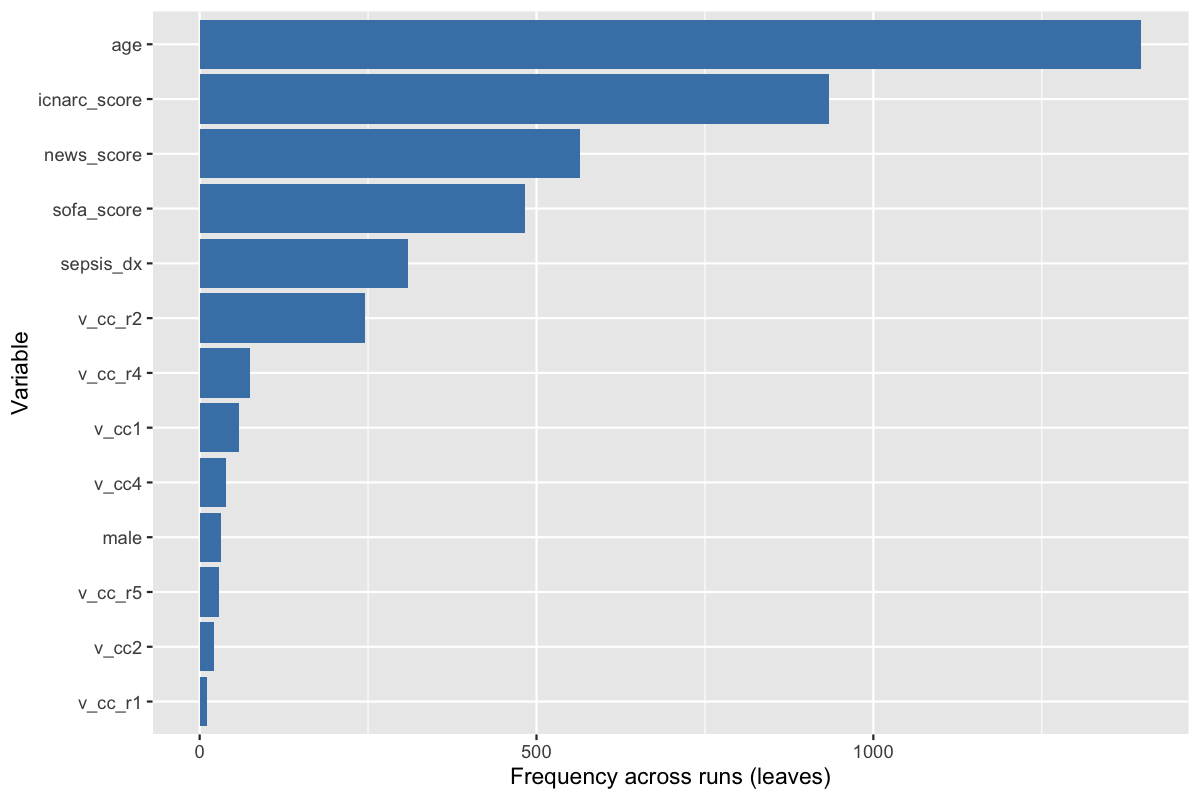}

        \smallskip
        (b) GRF-IV
    \end{minipage}

    \vspace{0.75em}

    \begin{minipage}[t]{0.55\textwidth}
        \centering
        \includegraphics[width=\linewidth]{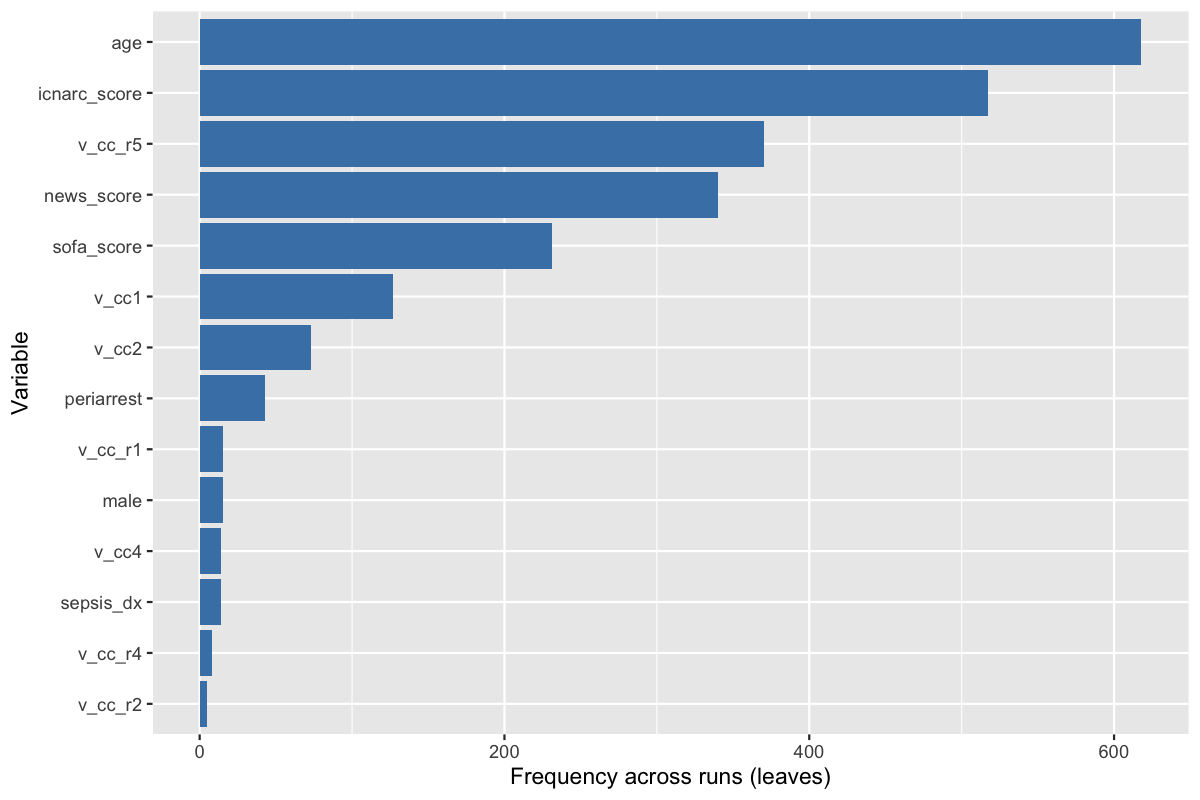}

        \smallskip
        (c) BCF-IV
    \end{minipage}

    \vspace{0.4cm}

    \begin{notes}
    The figures display the variables that appear most frequently in subgroup construction. Because the number of identified subgroups varies across methods, the total number of variable occurrences over 100 simulation runs differs accordingly.
    \end{notes}
\end{figure}

\FloatBarrier
\clearpage
\subsection*{Analysis with \texttt{site} included} 
\label{app:additional_application}
In the dataset of \citet{Keele_IV}, \texttt{site} indicates the hospital the patient attended. The authors addressed hospital-level differences at the design stage by imposing balance on hospital across treatment groups through their matching procedure. Since matching is not the focus of our study, we instead account for site-level heterogeneity by including one-hot-encoded hospital indicators as additional covariates. Thus, in addition to the 18 variables from the main analysis, we include 47 site dummy variables, yielding a total of 65 variables. As in the main analysis, we restrict to observations with $\hat{p}_Z(X_i) \in [0.1, 0.9]$. This avoids observations with very limited overlap in the probability of prompt ICU admission. We conduct the same analysis as in Section \ref{ch:emp_app} using a CART tree with \texttt{cp = 0.001} in the discovery step. Unlike in the Monte Carlo simulation study, however, we set \texttt{minbucket = 65} and \texttt{maxdepth = 5} to avoid subgroups that are too small. Since the discovery sample contains half of the full data, this implies that each terminal node contains at least about 1\% of the observations in the discovery set. 

The overall $\operatorname{CCACE}$ is $-0.1$, which is similar to the result reported by \citet{Keele_IV} for the strong IV setting with refined covariate balance $(-0.189)$. More generally, including \texttt{site} appears to produce smaller trees with fewer terminal nodes. Possible explanations are that the same tuning parameters as in the main analysis may be too restrictive once site indicators are added or that hospital effects may also absorb variation that would otherwise be captured by additional subgroup splits.
All methods use the site indicators as additional covariates for subgroup construction, although their relative importance differs across methods, as illustrated in Figures~\ref{fig:vip_comparison_site_included_1} and~\ref{fig:vip_comparison_site_included_2}.  In DRRF-IV, the first site indicator appears only in sixth position in the variable-frequency ranking, while in BCF-IV it appears in fourth position. In GRF-IV, the first site indicator appears only in tenth position, but site indicators are used most frequently overall. Thus, hospital-level variation appears to contribute to subgroup construction in all three methods, although the extent and form of this contribution differ across approaches.
\begin{figure}[!htbp]
    \centering
    \caption{Decision trees generated by GRF-IV and BCF-IV with hospital indicators included}
    \label{fig:tree_grf_bcf_site_incl}

    \begin{minipage}[t]{0.95\textwidth}
        \centering
        \includegraphics[
            width=\textwidth,
            height=0.36\textheight,
            keepaspectratio
        ]{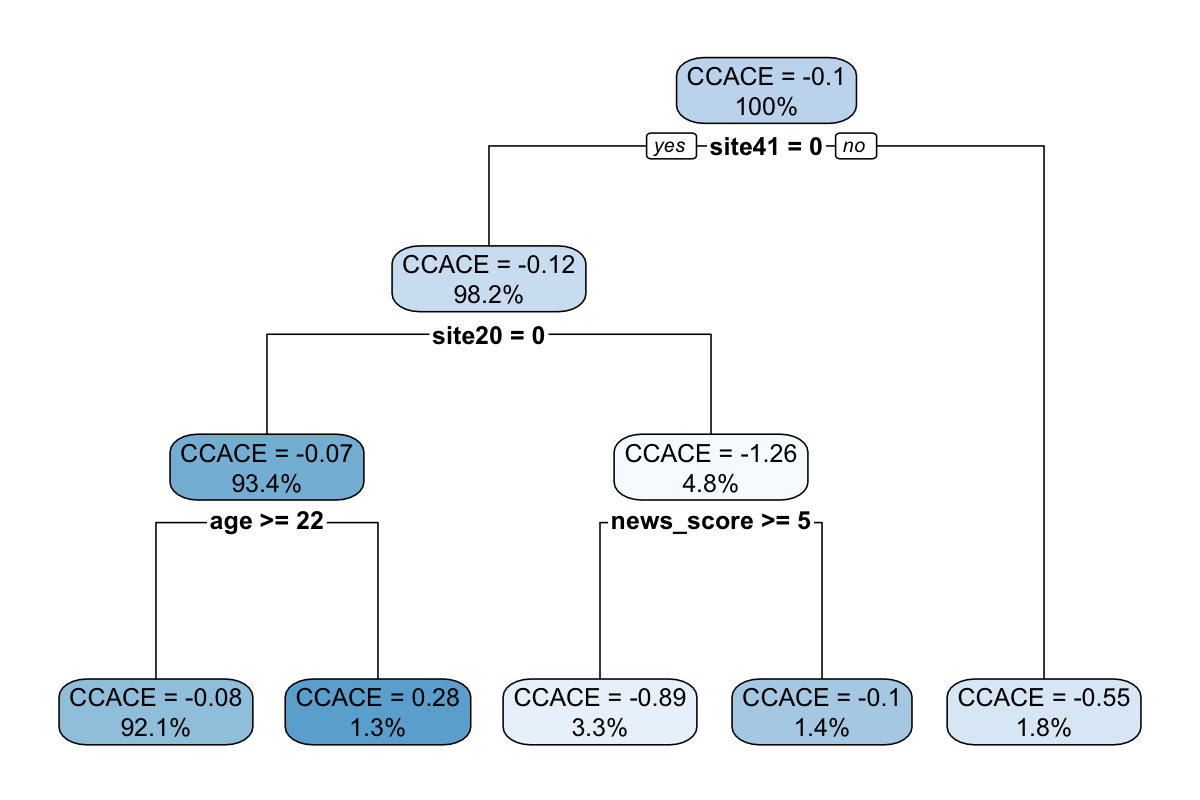}

        \smallskip
        (a) GRF-IV
    \end{minipage}

    \vspace{0.4cm}

    \begin{minipage}[t]{0.95\textwidth}
        \centering
        \includegraphics[
            width=\textwidth,
            height=0.36\textheight,
            keepaspectratio
        ]{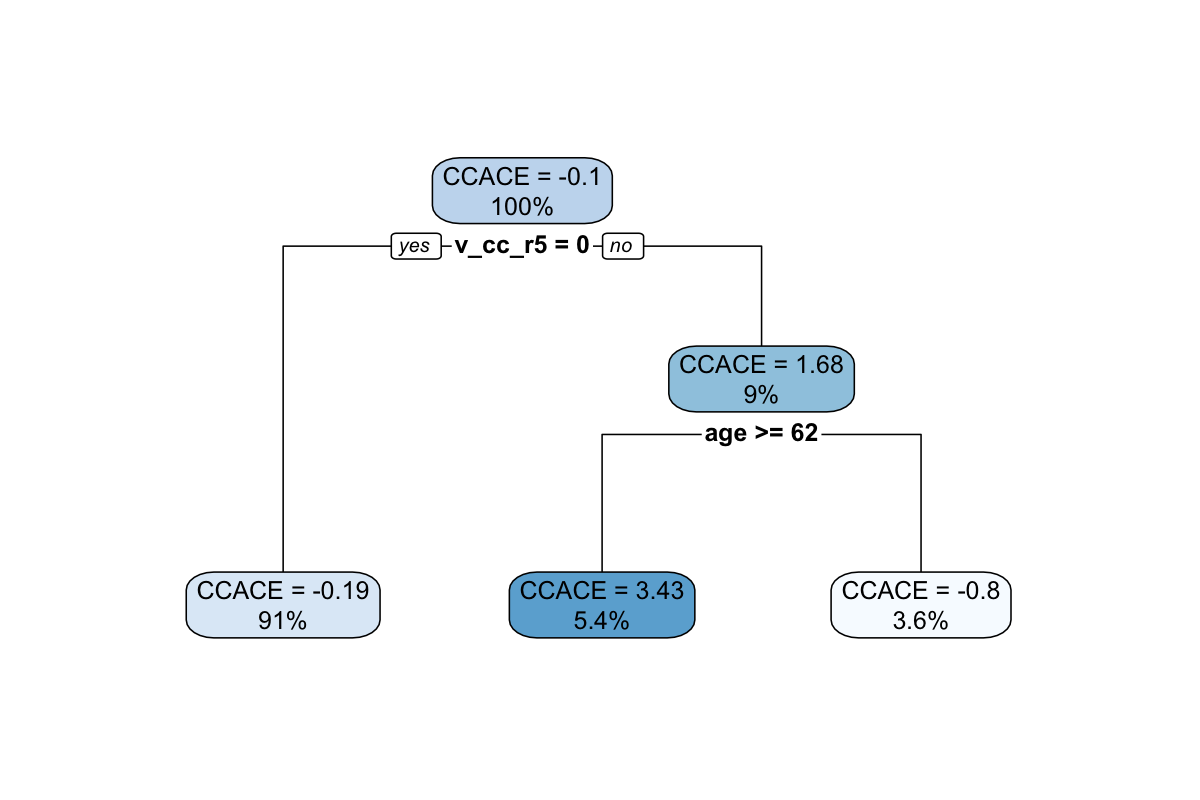}

        \smallskip
        (b) BCF-IV
    \end{minipage}

    \vspace{0.4cm}

    \begin{notes}
    The figure compares the subgroup trees obtained from GRF-IV and BCF-IV when hospital indicators are included in the analysis. GRF-IV identifies 5 terminal subgroups, while BCF-IV identifies 3 terminal subgroups.
    \end{notes}
\end{figure}

\FloatBarrier
\FloatBarrier

\begin{figure}[!htbp]
    \centering
    \caption{Decision tree generated by DRRF-IV with hospital indicators included}
    \label{fig:tree_drrf_site_incl}
    \includegraphics[
        width=\textwidth,
        height=0.85\textheight,
        keepaspectratio
    ]{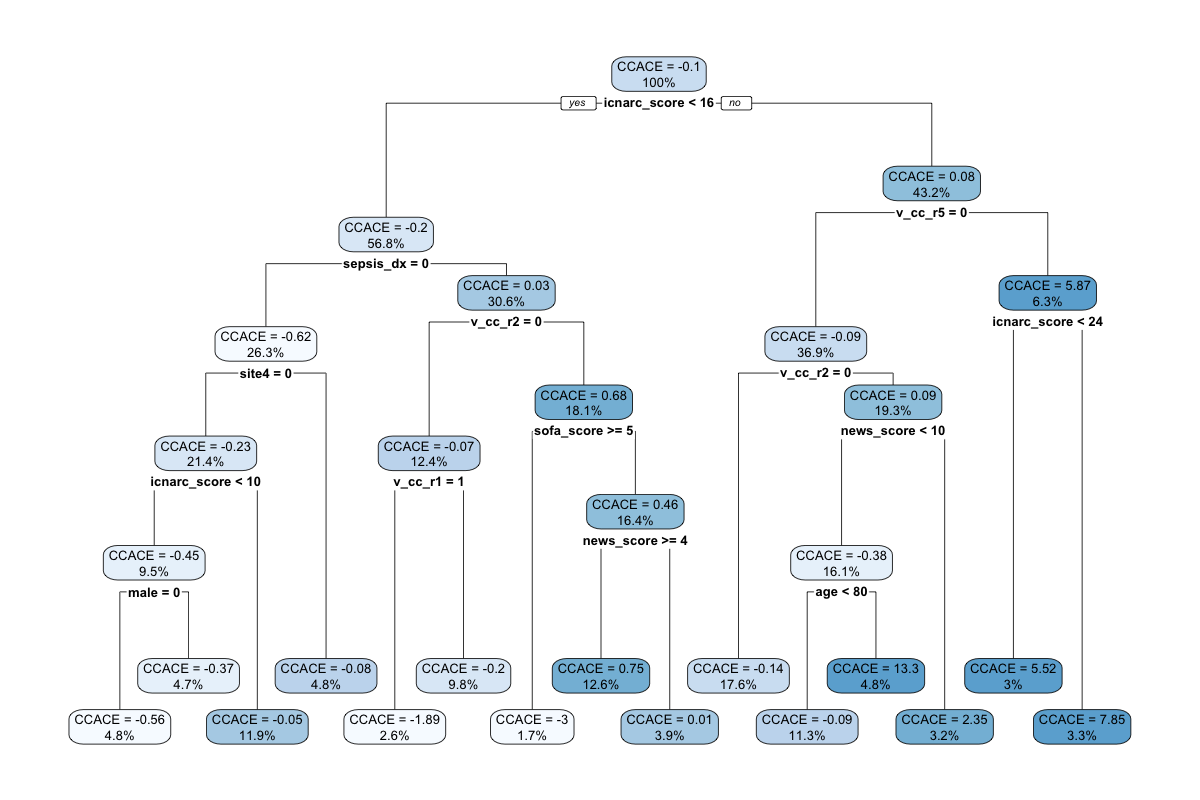}
    \caption*{\scriptsize\textit{Notes:} The figure shows the subgroup tree obtained from DRRF-IV when hospital indicators are included in the analysis. Terminal nodes correspond to discovered patient subgroups, with subgroup-specific CCACE estimates reported in the leaves. DRRF-IV identifies 5 terminal subgroups.}
\end{figure}

\FloatBarrier
\begin{figure}[htbp]
    \centering
    \caption{Variable importance plots for DRRF-IV and GRF-IV with hospital indicators included.}
    \label{fig:vip_comparison_site_included_1}

    \begin{minipage}[t]{0.8\textwidth}
        \centering
        \includegraphics[width=\linewidth]{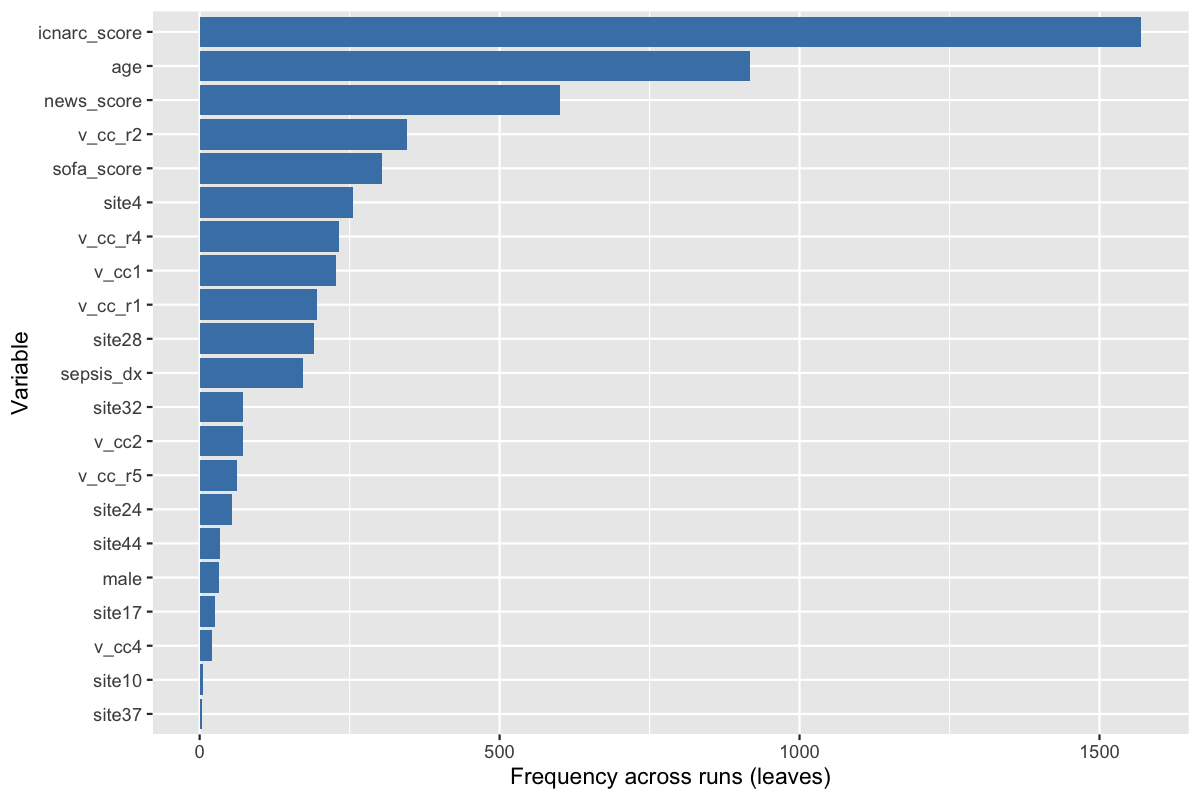}

        \smallskip
        (a) DRRF-IV
    \end{minipage}

    \vspace{0.75em}

    \begin{minipage}[t]{0.8\textwidth}
        \centering
        \includegraphics[width=\linewidth]{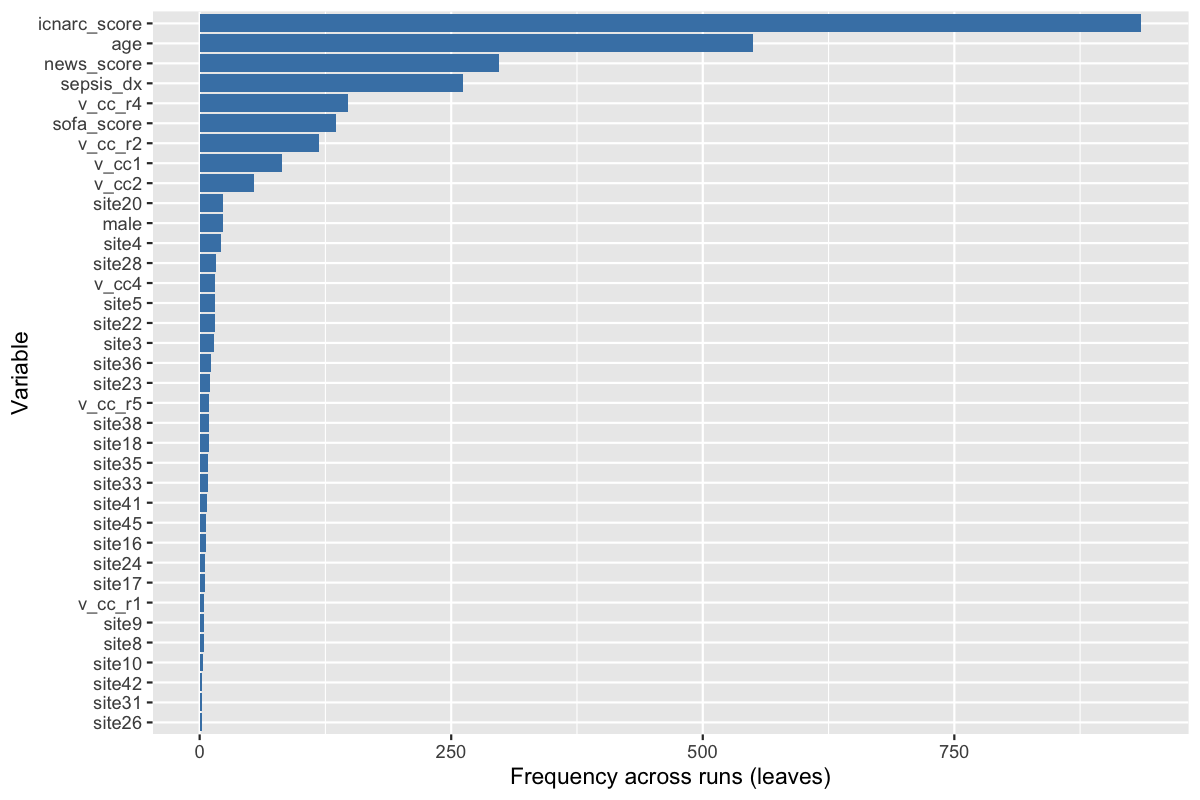}

        \smallskip
        (b) GRF-IV
    \end{minipage}

    \vspace{0.4cm}

    \caption*{\scriptsize\textit{Notes:} 
    The figures display the variables that appear most frequently in subgroup construction.}
\end{figure}

\begin{figure}[htbp]
    \centering
    \caption{Variable importance plot for BCF-IV with hospital indicators included}
    \label{fig:vip_comparison_site_included_2}

    \begin{minipage}[t]{0.8\textwidth}
        \centering
        \includegraphics[width=\linewidth]{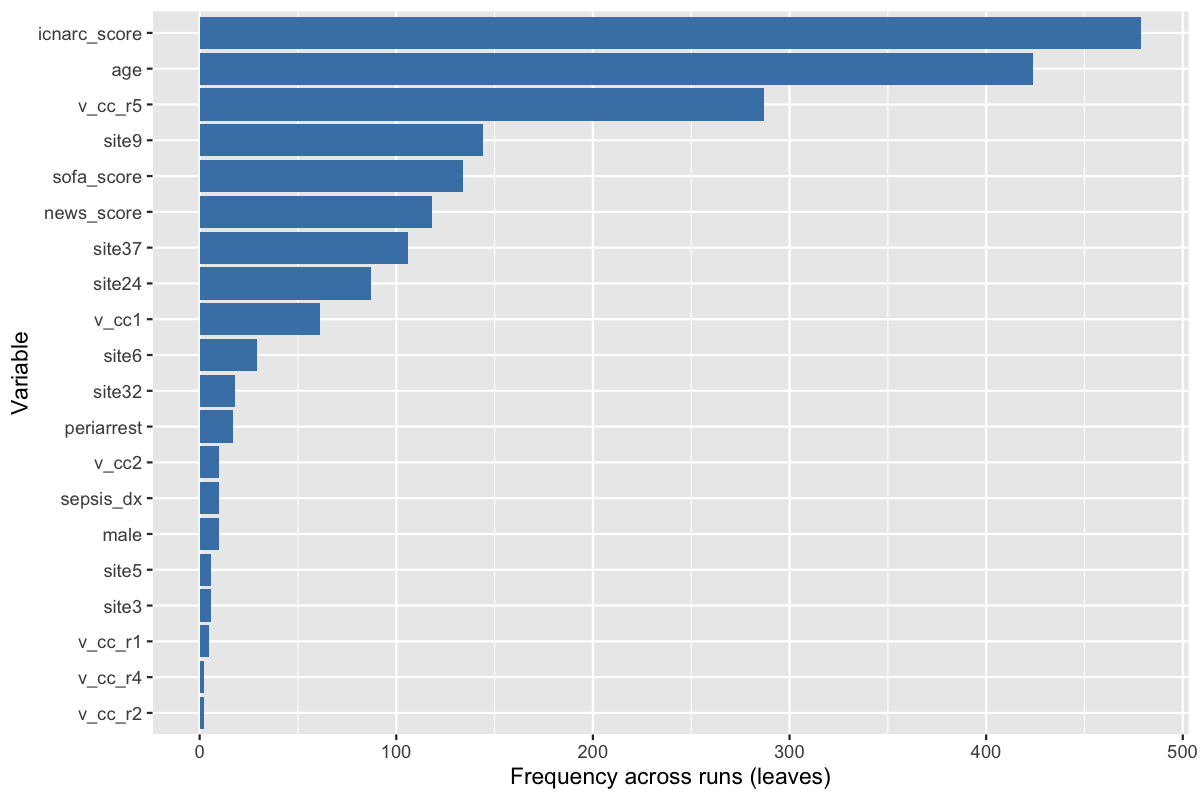}
        \smallskip
        BCF-IV
    \end{minipage}
    \vspace{0.4cm}
    \caption*{\scriptsize\textit{Notes:} 
    The figure displays the variables that appear most frequently in subgroup construction.}
\end{figure}
\FloatBarrier

\FloatBarrier

\subsection{Additional Theory}   \label{sec:Theory}
\subsection*{Doubly Robust Estimation}   \label{sec:Appendix_DR}

The DRRF-IV estimator proposed in Section~\ref{sec:RF with IV} is based on a transformed outcome $Y_i^{\star\text{IV}}$, which satisfies Neyman orthogonality with respect to $m(\cdot)$. It is, however, not doubly robust in the usual augmented inverse-probability weighted (AIPW) sense. Doubly robust estimators \citep{robins1995semiparametric,robins1994estimation,Bang_Robins_2006, funk2011doubly} combine outcome modeling and propensity score weighting \citep{Rosenbaum_Rubin_propensity} and remain consistent if at least one of the two nuisance components is correctly specified. In this section, we quickly outline a doubly robust extension as an alternative to DRRF-IV for estimating $\mathrm{ITT}(x)$. We adapt the notation of \citet{econml} to our IV setting, and define the nuisance functions
\[
m_z(x):=\mathbb E[Y_i\mid Z_i=z,X_i=x], \qquad z\in\{0,1\}. 
\]
With these nuisance functions and $p_Z(X_i)$, we construct the instrumental augmented pseudo-outcome, which is an IV version of the doubly robust estimators from \citep{robins1995semiparametric,robins1994estimation} 
\begin{align*}
Y_i^{\mathrm{AIPW\text{-}IV}}
:=
m_1(X_i)-m_0(X_i)
+
\frac{Z_i}{p_Z(X_i)}\bigl(Y_i-m_1(X_i)\bigr)
-
\frac{1-Z_i}{1-p_Z(X_i)}\bigl(Y_i-m_0(X_i)\bigr).
\end{align*}

\begin{theorem}
\label{thm:aipw_itt}
Under Assumption \ref{assump:overlap_iv}, 
$Y_i^{\mathrm{AIPW\text{-}IV}}$ satisfies
$$\mathbb E\!\left[Y_i^{\mathrm{AIPW\text{-}IV}}\mid X_i=x\right]
=
\mathrm{ITT}(x)$$
\end{theorem}

\begin{proof}
See Appendix \ref{sec:proof}
\end{proof}
As indicated by \citet{econml}, there is a trade-off between double/debiased and doubly robust estimators. In general, doubly robust methods tend to have higher variance, especially when $p_Z(X_i)$ is close to 0 or 1. By contrast, double/debiased estimators rely on orthogonalization with respect to nuisance functions and may be more sensitive when these nuisance components are poorly estimated.

\subsection*{Two-step pragmatic subgroup discovery} 
\label{app:2step_explaination}
The pragmatic two-step subgroup discovery approach \citep{Two_step_pragmatic_subgroup_discovery} separates treatment effect estimation from subgroup identification to improve interpretability. In the first step, any flexible CATE estimator is fitted on the full covariate set to capture potentially complex heterogeneity. In the second step, the estimated unit-level treatment effects are used as the outcome in a CART model based on a reduced set of interpretable covariates, yielding subgroups described by simple decision rules. Optionally, subgroup-specific treatment effects can be summarized within the resulting leaves to facilitate interpretation and comparison across groups. Algorithm \ref{algo:komura_two_step} summarizes the procedure. Note that, the choice $X_i^{\text{interp}}$ is not formally determined, but instead relies on researcher judgment and considerations. In particular, \citet{Two_step_pragmatic_subgroup_discovery} motivates this choice in terms of practical relevance and interpretability, rather than providing a quantified selection criterion.

\FloatBarrier
\vspace{0.5cm}
\begin{algorithm}[h]
\SetAlgoLined
\caption{Pragmatic Two-Step Subgroup Discovery}
\label{algo:komura_two_step}

\KwIn{$N$ training observations $(X_i, D_i, Z_i, Y_i)$}
\KwOut{Interpretable partition $\{\mathcal{L}_j\}_{j=1}^J$ and subgroup summaries}

\textbf{Step 1:} Fit any CATE estimator on the full covariate vector $X_i$ to obtain unit-level treatment effect predictions $\hat{\tau}(X_i)$

\vspace{0.3cm}

\textbf{Step 2:} Using the interpretable covariates $X_i^{\text{interp}}$, fit a CART model, where $\hat{\tau}(X_i)$ are used as the response variable to obtain a partition
\[
\widehat{\Pi}
=
\{\mathcal{L}_1,\dots,\mathcal{L}_J\}
\] of the covariate space.

\vspace{0.3cm}

\textbf{Post-estimation (Optional):}

For each subgroup $\mathcal{L}_j \in \widehat{\Pi}$ compute
\begin{align*}
\hat{\tau}_{\mathcal{L}_j}
&=
\frac{1}{|\{i:D_i=1,X_i\in\mathcal{L}_j\}|}
\sum_{\{i:D_i=1,X_i\in\mathcal{L}_j\}} Y_i^{obs} \\
&-
\frac{1}{|\{i:D_i=0,X_i\in\mathcal{L}_j\}|}
\sum_{\{i:D_i=0,X_i\in\mathcal{L}_j\}} Y_i^{obs}.
\end{align*}

\end{algorithm}
\FloatBarrier
\subsection*{GRF-IV} \label{app:GRF-IV}
Algorithm \ref{algo:Algorithm_GRF_IV} summarizes the implementation of the GRF-IV procedure. 
The method extends the generalized random forest framework to instrumental variable settings by separating tree construction from estimation under an honest sample-splitting scheme. Tree splits are guided by pseudo-outcomes based on local IV effects, and the resulting forest weights are used to solve local IV moment conditions for a target point $x$, yielding an estimate of $\hat{\tau}_{\text{CCACE}}(x)$.
For simplicity, we skip the local-centering step in the algorithmic presentation, but \texttt{grf::instrumental\_forest()} estimates these quantities internally and uses them for local centering. 

\begin{algorithm}[h]
\small
\caption{GRF-IV}
\label{algo:Algorithm_GRF_IV}
\KwIn{$N$ training observations $(X_i, D_i, Z_i,  Y_i)$}
\KwOut{Estimated treatment effect $\hat{\tau}_{\text{CCACE}}(x)$}
\For{$b = 1$ \KwTo $B$}{
    Draw a random subsample of size $s$ without replacement, and partition it into two disjoint sets:  a training set $\mathcal{I}$ of size $\lfloor s/2 \rfloor$ and an estimation set $\mathcal{J}$ of size $\lceil s/2 \rceil$. \\
    
    \textbf{1. Step: Tree Growing (on $\mathcal{I}$)}
    \begin{itemize}
        \item At each parent node $P$, estimate the local plug-in IV effect:
        \[
        \hat{\tau}_P = \frac{\widehat{\operatorname{Cov}}(Y_i, Z_i \mid X_i \in P)}{\widehat{\operatorname{Cov}}(D_i, Z_i \mid X_i \in P)}
        \]
        \item Construct pseudo-outcomes:
       \[
        \rho_i
        =        (Z_i - \bar{Z}_P)
        \left[(Y_i - \bar{Y}_P) - \hat{\tau}_P (D_i - \bar{D}_P)\right]
        \]
        where $\bar{Z}_P$, $\bar{D}_P$, and $\bar{Y}_P$ denote the sample means of $Z_i$, $D_i$, and $Y_i$ within the current node $P$. The splits are then chosen to maximize variation in $\rho_i$.
    \end{itemize}
       \textbf{2. Step: Weighting (on $\mathcal{J}$)}
    \begin{itemize}
        \item For each test point $x$, find the leaf $L_b(x)$ such that $X_i \in L_b(x)$. Then assign weights
        \[
        \alpha_{bi}(x) = \frac{\mathds{1}(X_i \in L_b(x))}{|L_b(x)|}.
        \]
    \end{itemize}
}
\textbf{Step 3: Aggregation and Final Estimation}
\begin{itemize}
    \item Compute aggregated weights: 
    \[
    \alpha_i(x) = \frac{1}{B} \sum_{b=1}^B \alpha_{bi}(x).
    \]
    \item Solve the local moment conditions to obtain $\hat{\tau}_{\text{CCACE}}(x)$ and $\hat{\eta}(x)$:
    \begin{align*}
            \sum_{i=1}^N \alpha_i(x) \cdot Z_i \cdot \left( Y_i - \hat{\eta}(x) - \hat{\tau}_{\text{CCACE}}(x) \cdot D_i \right) = 0 \\
    \sum_{i=1}^N \alpha_i(x) \cdot \left( Y_i - \hat{\eta}(x) - \hat{\tau}_{\text{CCACE}}(x) \cdot D_i \right) = 0
    \end{align*}
\end{itemize}
\end{algorithm}

\clearpage
\section{Proofs} \label{sec:proof}
This section provides the proofs of the theoretical results stated in the main text. We proceed in the order in which the results are introduced. First, we prove Theorem~\ref{thm:dr_itt}, showing that the transformed outcome $Y_i^{\star\mathrm{IV}}$ has conditional expectation equal to the conditional intention-to-treat effect $\operatorname{ITT}(x)$. We then establish the Neyman orthogonality of the corresponding conditional moment with respect to the nuisance function $m(\cdot)$. 

\begin{proof}[\textbf{Proof of Theorem~\ref{thm:dr_itt}}]
First, by definition of \(Y_i^{\star\mathrm{IV}}\) we can rewrite 
\begin{align}
\mathbb{E}[Y_i^{\star\mathrm{IV}} \mid X_i=x]
&=
\mathbb{E}\left[
\frac{(Z_i-p_Z(x))(Y_i-m(x))}{p_Z(x)(1-p_Z(x))}
\,\middle|\, X_i=x
\right]  \notag \\
&=
\frac{1}{p_Z(x)(1-p_Z(x))}
\mathbb{E}\Big[(Z_i-p_Z(x))(Y_i-m(x)) \mid X_i=x\Big]. \label{eq:begin_proof}
\end{align}
Since
\begin{align*}
\mathbb{E}\Big[(Z_i-p_Z(x))(Y_i-m(x)) \mid X_i=x\Big]
&=
\mathbb{E}[Z_i(Y_i-m(x))\mid X_i=x] \\
&\qquad
-\,p_Z(x)\,\underbrace{\mathbb{E}[Y_i-m(x)\mid X_i=x]}_{=\,0},
\end{align*}
it follows that
\begin{align*}
\mathbb{E}\Big[(Z_i-p_Z(x))(Y_i-m(x)) \mid X_i=x\Big]
&=
\mathbb{E}[Z_i(Y_i-m(x))\mid X_i=x] \\
&=
\mathbb{E}[Z_iY_i\mid X_i=x]-m(x)\mathbb{E}[Z_i\mid X_i=x]. \\
\end{align*}
Since $Z_i$ is binary we can rewrite $\mathbb{E}[Z_iY_i\mid X_i=x]$ as follows
\begin{align*}
\mathbb{E}[Z_iY_i\mid X_i=x]&=1\cdot \mathbb{E}[Y_i\mid Z_i=1,X_i=x]\cdot \Pr(Z_i=1\mid X_i=x)\\&+ 0\cdot \mathbb{E}[Y_i\mid Z_i=0,X_i=x]\cdot \Pr(Z_i=0\mid X_i=x) \\
&=\mathbb{E}[Y_i\mid Z_i=1,X_i=x]\cdot p_Z(x).
\end{align*}
And therefore 
\begin{samepage}
\begin{align*}
\mathbb{E}\Big[(Z_i-p_Z(x))(Y_i-m(x)) \mid X_i=x\Big] &=p_Z(x)\,\mathbb{E}[Y_i\mid Z_i=1,X_i=x]-p_Z(x)m(x) \\
&=
p_Z(x)\Big(\mathbb{E}[Y_i\mid Z_i=1,X_i=x]-m(x)\Big).
\end{align*}
\end{samepage}
Substituting this into \eqref{eq:begin_proof}, we obtain
\begin{equation} \label{eq:numerator}
\mathbb{E}[Y_i^{\star\mathrm{IV}}\mid  X_i=x]
=
\frac{\mathbb{E}[Y_i\mid Z_i=1, X_i=x]-m(x)}{1-p_Z(x)}.
\end{equation}

Using
\[
m(x)
=
p_Z(x)\mathbb{E}[Y_i\mid Z_i=1, X_i=x]
+
(1-p_Z(x))\mathbb{E}[Y_i\mid Z_i=0, X_i=x],
\]
we can rewrite the numerator of \eqref{eq:numerator} \[
\mathbb{E}[Y_i\mid Z_i=1, X_i=x]-m(x)
=
(1-p_Z(x))
\Big(
\mathbb{E}[Y_i\mid Z_i=1,X_i=x]-\mathbb{E}[Y_i\mid Z_i=0, X_i=x]
\Big).
\]
Finally,
\[
\mathbb{E}[Y_i^{\star\mathrm{IV}}\mid  X_i=x]
=
\mathbb{E}[Y_i\mid Z_i=1, X_i=x]-\mathbb{E}[Y_i\mid Z_i=0, X_i=x]
=
\operatorname{ITT}(x),
\]
which completes the proof. 
\end{proof}
\begin{proof}[\textbf{Proof of Lemma~\ref{lem:orthogonality}}]
Let \(x\) be a fixed test point and $m_\lambda(\cdot)$ a perturbation given by
\[
m_\lambda(\cdot)=m(\cdot)+\lambda h(\cdot),
\]
where \(h(\cdot)\) is an arbitrary perturbation. Then
\[
\psi_i(\theta(x),m_\lambda,p_Z)
=
\frac{(Z_i-p_Z(X_i))(Y_i-m(X_i)-\lambda h(X_i))}
     {p_Z(X_i)(1-p_Z(X_i))}
-\theta(x).
\]
The corresponding conditional moment is
\[
\Psi(\theta(x),m_\lambda,p_Z)
=
\mathbb{E}\!\left[
\frac{(Z_i-p_Z(X_i))(Y_i-m(X_i)-\lambda h(X_i))}
     {p_Z(X_i)(1-p_Z(X_i))}
\Bigm| X_i=x
\right]
-\theta(x).
\]
Since we condition on \(X_i=x\), the functions \(p_Z(X_i)\), \(m(X_i)\), and \(h(X_i)\) become \(p_Z(x)\), \(m(x)\), and \(h(x)\), respectively. Hence,
\[
\Psi(\theta(x),m_\lambda,p_Z)
=
\mathbb{E}\!\left[
\frac{(Z_i-p_Z(x))(Y_i-m(x)-\lambda h(x))}
     {p_Z(x)(1-p_Z(x))}
\Bigm| X_i=x
\right]
-\theta(x).
\]
Differentiating with respect to \(\lambda\) gives
\begin{equation} \label{eq:differentating}
    \frac{\partial}{\partial\lambda}
\Psi(\theta(x),m_\lambda,p_Z)
=
-
\mathbb{E}\!\left[
\frac{(Z_i-p_Z(x))h(x)}
     {p_Z(x)(1-p_Z(x))}
\Bigm| X_i=x
\right].
\end{equation}
Evaluating at $\lambda=0$ does not change \eqref{eq:differentating}, since it does not depend on \(\lambda\). Because \(h(x)\), \(p_Z(x)\), and \(1-p_Z(x)\) are fixed conditional on \(X_i=x\) and by definition of $p_Z(x)$ it follows that
\[-
\mathbb{E}\!\left[
\frac{(Z_i-p_Z(x))h(x)}
     {p_Z(x)(1-p_Z(x))}
\Bigm| X_i=x
\right]=0.
\]
Therefore,
\[
\left.
\frac{\partial}{\partial\lambda}
\Psi(\theta(x),m+\lambda h,p_Z)
\right|_{\lambda=0}
=
0.
\]
Thus, the conditional moment functional is Neyman--orthogonal with respect to \(m(\cdot)\).

\end{proof}

\begin{proof}[\textbf{Proof of Theorem~\ref{thm:aipw_itt}}]
We obtain
\begin{align} \label{eq:AIPW-IV}
\mathbb E\!\left[Y_i^{\mathrm{AIPW\text{-}IV}}\mid X_i=x\right]
&=
m_1(x)-m_0(x)
+
\mathbb E\!\left[\frac{Z_i}{p_Z(x)}\bigl(Y_i-m_1(x)\bigr)\mid X_i=x\right]  \notag \\
&\qquad
-
\mathbb E\!\left[\frac{1-Z_i}{1-p_Z(x)}\bigl(Y_i-m_0(x)\bigr)\mid X_i=x\right].
\end{align}
If $p_Z(x)$ is correctly specified, then for the third term in Equation \eqref{eq:AIPW-IV} we have 
\begin{align*}
    \mathbb E\!\left[\frac{Z_i}{p_Z(x)}\bigl(Y_i-m_1(x)\bigr)\mid X_i=x\right]&= \mathbb E\!\left[\frac{Z_iY_i}{p_Z(x)}\mid X_i=x\right] 
-
\mathbb E\!\left[\frac{Z_i m_1(x)}{p_Z(x)}\mid X_i=x\right], \\
&= \frac{1}{p_Z(x)}\mathbb E\!\left[ Y_i Z_i\mid  X_i=x\right] -m_1(x)\frac{\mathbb E[Z_i\mid X_i=x]}{p_Z(x)}, \\
&=\mathbb E\!\left[Y_i \mid Z_i=1,X_i=x\right] -m_1(x).
\end{align*}
Analogously, for the last term in \eqref{eq:AIPW-IV} we have 
\begin{align*}
\mathbb E\!\left[\frac{1-Z_i}{1-p_Z(x)}\bigl(Y_i-m_0(x)\bigr)\mid X_i=x\right] =\mathbb E[Y_i\mid Z_i=0,X_i=x]-m_0(x).
\end{align*}
Substituting these into Equation~\eqref{eq:AIPW-IV} yields
 $$\mathbb E\!\left[Y_i^{\mathrm{AIPW\text{-}IV}}\mid X_i=x\right]
= \mathbb E\!\left[Y_i \mid Z_i=1,X_i=x\right] - \mathbb E[Y_i\mid Z_i=0,X_i=x]
=\mathrm{ITT}(x).$$
If instead $m_1(x)$ and $m_0(x)$ are correctly specified, then the residual terms satisfy
\[
\mathbb E[Y_i-m_1(x)\mid Z_i=1,X_i=x]=0,
\qquad
\mathbb E[Y_i-m_0(x)\mid Z_i=0,X_i=x]=0.
\]
Therefore, the third term in \eqref{eq:AIPW-IV} becomes 
\[ \frac{1}{p_Z(x)}
\mathbb E\!\left[Z_i\bigl(Y_i-m_1(x)\bigr)\mid X_i=x\right]
=
\frac{1}{p_Z(x)}p_Z(x)\,
\mathbb E[Y_i-m_1(x)\mid Z_i=1,X_i=x]
=0.\]
Analogously, the last term in \eqref{eq:AIPW-IV} also vanishes. 
Hence, 
$$\mathbb E\!\left[Y_i^{\mathrm{AIPW\text{-}IV}}\mid X_i=x\right]=m_1(x)-m_0(x)=\mathrm{ITT}(x).$$
Therefore, the $Y_i^{\mathrm{AIPW\text{-}IV}}$ is doubly robust in the sense that it identifies $\mathrm{ITT}(x)$ if either the propensity score or the outcome regressions are correctly specified.
\end{proof}
\section{Algorithmic and Implementation Details} 
\label{app:implementation_TSIV}
\subsection{Implementation in \texttt{R}}
All simulations were performed in \texttt{R} using RStudio \citep{r_2025, rstudio_2025}.
The corresponding implementation is available at \url{https://github.com/karo93/TwoStepForestIV}.

For the simulation study we used the packages \texttt{grf} \citep{grf_package}, \texttt{bartCause} \citep{bartCause}, and \texttt{SparseBCF} \citep{caron2022shrinkage}. Across methods, the only difference lies in the estimation of $\hat{\tau}^{\mathrm{dis}}_{\mathrm{CCACE}}(x)$. The benchmark BCF-IV algorithm builds on the implementation of \citet{bargaglibcf}, however, we modified the \texttt{bcf\_iv}-function to improve computational efficiency and add several numerical safeguards, as the original implementation occasionally failed at these steps or returned missing values.

First, observations with propensity scores close to zero and one are excluded to ensure overlap, similar to \citet{stoffi_gnecco_CTIV}. Second, in the discovery step, we exclude observations for which the estimated compliance rate is not finite or is smaller than a threshold \(\varepsilon=0.001\). The discovery-stage CCACE is therefore computed only for observations satisfying
\begin{align*}
   \widehat{\pi}_C(X_i) &> \varepsilon,
 \intertext{for}
 \widehat{\tau}^{\operatorname{dis}}_{\operatorname{CCACE}}(X_i)
&=
\frac{\widehat{\operatorname{ITT}}(X_i)}
{\widehat{\pi}_C(X_i)}.
\end{align*}

This avoids using unstable ratios generated by zero or near-zero estimated compliance rates during subgroup discovery. Third, node-wise IV regressions are conducted only when the corresponding subgroup contains a sufficient number of observations that show variation in both treatment receipt and instrument assignment. In addition, regressions with rank-deficient model matrices or zero residual degrees of freedom are skipped. In each Monte Carlo replication, all methods use a 50/50 split to construct discovery and inference subsamples. The subgroup tree is estimated with \texttt{rpart} at a complexity parameter \texttt{cp = 0.001}, to construct a sufficiently deep tree. For the forest-based procedures, we use \texttt{num.trees = 2000}. In the BCF-IV implementation, Bayesian first-stage estimators are run with 500 burn-in draws and 500 posterior draws for continuous outcomes. All remaining tuning parameters are left at their default values, as a parameter tuning analysis is beyond the scope of this paper. Furthermore, in our BCF-IV implementation, $\text{ITT}(x)$ is estimated using \texttt{SparseBCF::SparseBCF()}. Setting \texttt{sparse = FALSE} produces the same results as \texttt{bartCause::bartc}, but with lower computational cost. A runtime comparison is reported below in Section \ref{app:runtime}.

\subsection{Estimation of Nuisance Functions}
\label{app:estimation_nuisance}
The nuisance functions are estimated on the discovery sample. Specifically, we estimate $$m(x)=\mathbb{E}[Y_i\mid X_i=x]$$
using a regression forest via the $\texttt{randomForest::randomForest()}$ function, and $$p_Z(x)=\operatorname{Pr}(Z_i=1\mid X_i=x)$$
using logistic regression via \texttt{glm()}. Let \(\widehat m(X_i)\) and \(\widehat p_Z(X_i)\) denote the corresponding fitted values. These estimates are then used to construct the transformed outcome
\[
Y_i^{\star\text{IV}}
=
\frac{(Z_i-\widehat p_Z(X_i))(Y_i-\widehat m(X_i))}
{\widehat p_Z(X_i)\bigl(1-\widehat p_Z(X_i)\bigr)},
\qquad i \in \mathcal{I}_{\mathrm{dis}}.
\]
Note that in the implementation, $\widehat m(X_i)$ and $\widehat p_Z(X_i)$ are obtained by in-sample prediction on $\mathcal{I}_{\mathrm{dis}}$. We also considered to estimate $m(x)$ using the $\texttt{grf::regression\_forest()}$ function with honest splitting to mitigate overfitting. This yielded very similar results, suggesting that the effect of this modification is negligible in our setting. Nevertheless, the resulting pseudo-outcome may still be affected by bias from nuisance estimation. Cross-fitting or out-of-bag prediction could further reduce this bias, but incorporating these refinements is beyond the scope of this work.
In addition, we estimate the conditional compliance rate
\[
\pi_C(x)=\mathbb E[D_i\mid Z_i=1,X_i=x]-\mathbb E[D_i\mid Z_i=0,X_i=x],
\]
using an honest causal forest via $\texttt{grf::causal\_forest()}$ with \(D_i\) as the outcome, \(Z_i\) as the treatment indicator, and \(X_i\) as covariates. Let \(\widehat \pi_C(x)\) denote the resulting estimate. The estimated conditional complier average causal effect is then given by
\[
\widehat\tau_{\mathrm{CCACE}}(x)
=
\frac{\widehat{\mathrm{ITT}}(x)}{\widehat\pi_C(x)}.
\]

\subsection{Computational Runtime}
\label{app:runtime}
To assess computational performance, we benchmark runtime in each Monte Carlo simulation by recording elapsed wall-clock time for the full end-to-end procedure using \texttt{proc.time()} in \texttt{R}. All computations were conducted on a Mac mini with 32GB RAM, an Apple M2 Pro chip, and macOS 14.5. For each replicate, we store both the total runtime and the subroutine (sample splitting, propensity score estimation, first-stage estimation, tree construction, and node-wise IV regressions) runtimes. 
Furthermore, runtime is evaluated on a subset of the simulation designs introduced in Section \ref{ch:sim_study}, focusing on sample sizes $N \in \{2000,5000,10000\}$, $MC = 100$ and effect size $k \in \{0.5, 4\}$, for both slight and strong heterogeneity settings.

We first compare the original BCF-IV algorithm with our modified implementation based on \texttt{SparseBCF}. Since runtime optimization is not the main contribution of this paper, this comparison is restricted to a single representative design, the strong heterogeneity setting with $k = 0.5$ and $N = 2000$. Figure~\ref{fig:runtime_bcf_compare} shows that the modified version achieves lower runtime and markedly reduced dispersion, suggesting that it is both computationally faster and more stable across simulation replications. In all subsequent runtime comparisons with the non-Bayesian methods, we therefore use the modified BCF-IV implementation.
\begin{figure}[H]
  \centering
      \caption{Runtime comparison between the original and the modified BCF-IV.}
  \label{fig:runtime_bcf_compare}
    \includegraphics[width=0.6\textwidth]{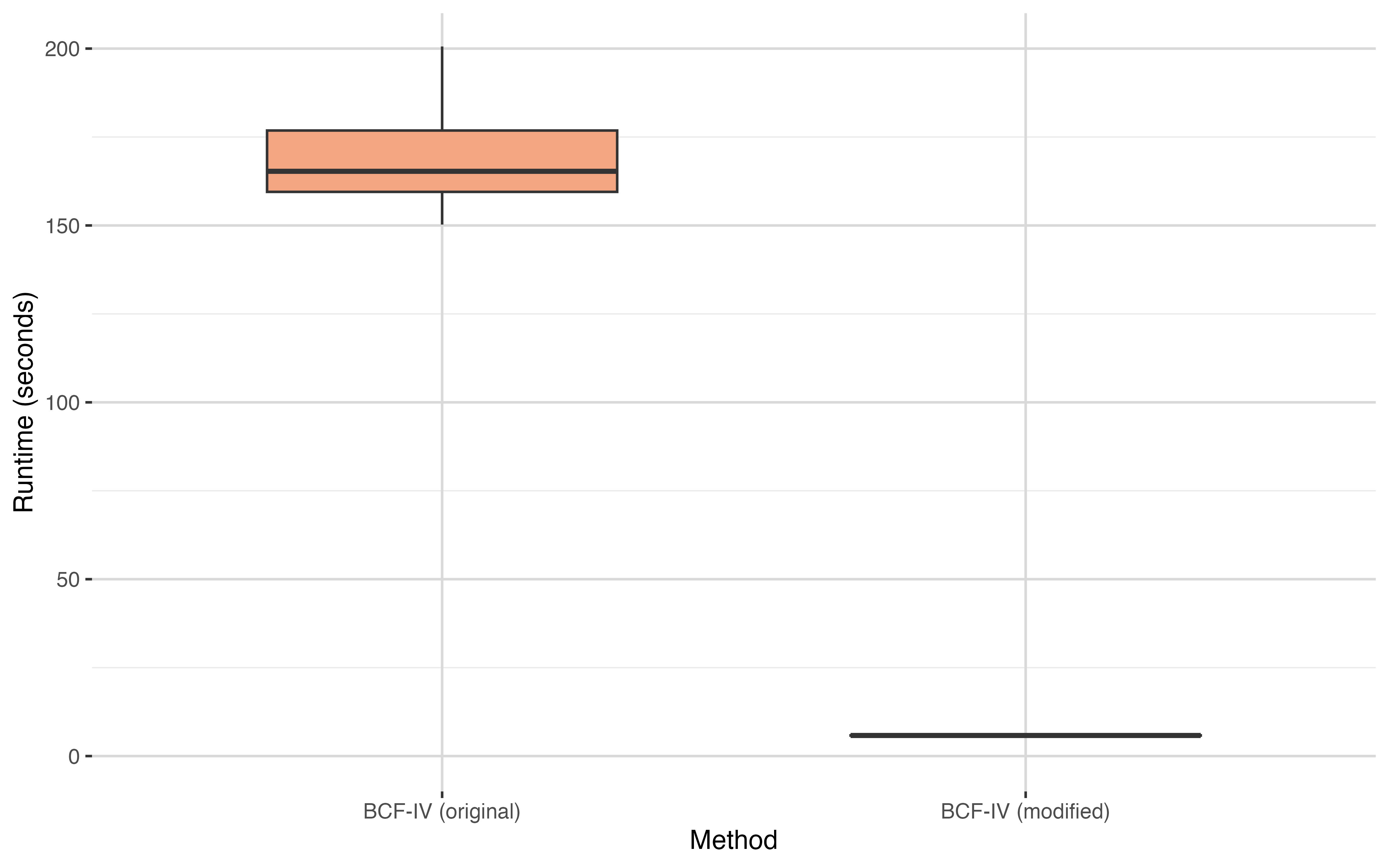}
\caption*{\footnotesize\textit{Notes:} Runtime is measured in seconds. The orange boxplot is the BCF-IV algorithm, whereas the red boxplot shows the modified implementation. Lower values indicate better computational performance.
 }\end{figure}

Figure~\ref{fig:runtime_all_compare} summarizes the runtime comparison results. Across all scenarios, the runtime comparison yields broadly similar results. Despite using the modified BCF-IV, it still remains by far the slowest method, with runtimes approximately ten times higher compared to GRF-IV. In contrast, GRF-IV is not only the fastest approach, but also exhibits the smallest increase in runtime as the sample size grows. For BCF-IV, the main driver of the high computational cost is the Bayesian estimation of  $\text{ITT}(x)$, whereas for DRRF-IV, the higher runtime is mainly driven by the modified splitting rule. 

\begin{figure}[!htbp]
  \centering
      \caption{Runtime comparison between for different dataset sizes.}
  \label{fig:runtime_all_compare}
    \includegraphics[width=0.75\textwidth]{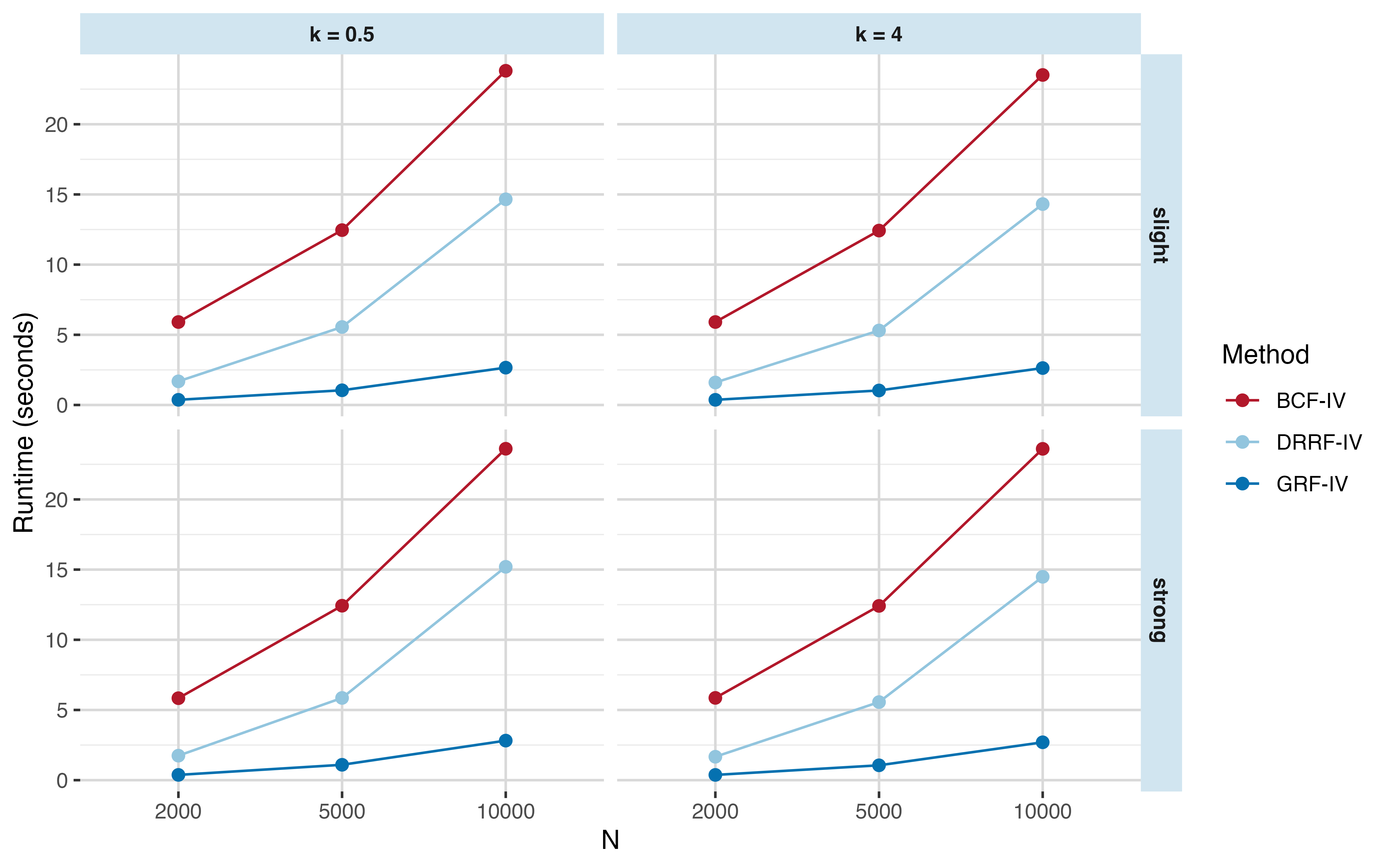}
\caption*{\footnotesize\textit{Notes:} Average runtime is measured in seconds. Lower values indicate better computational performance.
 }\end{figure}

\FloatBarrier
\begin{table}[!htbp]
\caption{Average runtime comparison across methods.}
\label{tab:runtime_comparison}
\centering
\setlength{\tabcolsep}{8pt}
\resizebox{.98\linewidth}{!}{
\begin{tabular}{l r rrr @{\hskip 1.2cm} rrr}
\\[-1.8ex]\hline
\hline \\[-1.8ex]
\multicolumn{2}{c}{ } & \multicolumn{6}{c}{heterogeneity scenario} \\
\cmidrule(l{3pt}r{3pt}){3-8}
\multicolumn{2}{c}{ } & \multicolumn{3}{c}{strong} & \multicolumn{3}{c}{slight} \\
\cmidrule(l{3pt}r{3pt}){3-5} \cmidrule(l{3pt}r{3pt}){6-8}
Method & $k$ & $N=2000$ & $N=5000$ & $N=10000$ & $N=2000$ & $N=5000$ & $N=10000$ \\
\midrule
\multirow{2}{*}{BCF-IV}
 & 0.5 & 5.839 (0.084) & 12.427 (0.097) & 23.610 (0.806) & 5.911 (0.104) & 12.458 (0.091) & 23.810 (1.042) \\
 & 4.0 & 5.861 (0.072) & 12.413 (0.148) & 23.603 (0.275) & 5.911 (0.094) & 12.426 (0.127) & 23.511 (0.572) \\
\addlinespace
\multirow{2}{*}{DRRF-IV}
 & 0.5 & 1.745 (0.085) & 5.857 (0.157) & 15.198 (0.353) & 1.684 (0.059) & 5.556 (0.122) & 14.654 (0.247) \\
 & 4.0 & 1.667 (0.079) & 5.562 (0.147) & 14.487 (0.723) & 1.602 (0.066) & 5.307 (0.099) & 14.312 (0.521) \\
\addlinespace
\multirow{2}{*}{GRF-IV}
 & 0.5 & 0.375 (0.026) & 1.095 (0.046) & 2.817 (0.067) & 0.369 (0.022) & 1.048 (0.031) & 2.660 (0.106) \\
 & 4.0 & 0.374 (0.023) & 1.058 (0.044) & 2.697 (0.029) & 0.368 (0.026) & 1.035 (0.039) & 2.631 (0.074) \\
\bottomrule
\end{tabular}}
\caption*{\footnotesize\textit{Notes:} Entries report mean total runtime in seconds, with standard deviations in parentheses, over \(MC=100\) simulation replications. Lower values indicate better computational performance.}
\end{table}
\FloatBarrier
\end{document}